\documentclass[11pt,letterpaper]{article}

\usepackage[linesnumbered,ruled,vlined]{algorithm2e}

\usepackage{tikz}
\usetikzlibrary{positioning}

\usepackage{amsthm,amsfonts}

\usepackage{latexsym,nicefrac,bbm}
\usepackage{xspace}
\usepackage[top=1in, bottom=1in, left=1in, right=1in]{geometry}
\usepackage{hyperref} 
\usepackage{multicol}
\usepackage{sidecap}
\usepackage{tcolorbox}

\usepackage{rotating}

\usepackage{changepage}

\usepackage{siunitx,multirow, booktabs, longtable,rotating}

\usepackage{xcolor}

\usepackage{MnSymbol}

\usepackage{float}
\renewcommand{\epsilon}{\varepsilon}

\newcommand{\argmin}{\operatornamewithlimits{argmin}}

\newcommand{\coloneqq}{:=}

\newtheorem{theorem}{Theorem}[section]
\newtheorem{lemma}[theorem]{Lemma}
\newtheorem{definition}[theorem]{Definition}
\newtheorem{corollary}[theorem]{Corollary}
\newtheorem{proposition}[theorem]{Proposition}
\newtheorem{remark}[theorem]{Remark}
\newtheorem{example}[theorem]{Example}

\newcommand{\eps}{\varepsilon}

\usepackage{verbatim}

\usepackage{algpseudocode}

\SetKwInOut{KwHyperparameters}{Hyperparameters \hspace{-5.8mm}}

\usepackage{mhequ}
\def \be{\begin{equs}}
\def \ee{\end{equs}}

\newcommand{\ai}{i}
\newcommand{\aj}{j}
\newcommand{\ak}{k}
\newcommand{\aell}{\ell}

\usepackage{times}

\title{\bf Greedy Adversarial Equilibrium: An Efficient Alternative to Nonconvex-Nonconcave Min-Max Optimization}

\author{Oren Mangoubi \\ Worcester Polytechnic Institute \and Nisheeth K. Vishnoi \\ Yale University}

\date{}
\begin{document}

\maketitle

\begin{abstract}

Min-max optimization of an objective function $f: \mathbb{R}^d \times \mathbb{R}^d \rightarrow \mathbb{R}$ is an important model for robustness in an adversarial setting, 
 with applications to many areas including optimization, economics, and deep learning.
In many applications $f$ may be nonconvex-nonconcave, and finding a global min-max point may be computationally intractable. 
There is a long line of work that seeks computationally tractable algorithms for alternatives to the min-max optimization model.
However, many of the alternative models have solution points which are only guaranteed to exist under strong assumptions on $f$, such as convexity, monotonicity, or special properties of the starting point. 
We propose an optimization model, the $\varepsilon$-greedy adversarial equilibrium, and show that it can serve as a computationally tractable alternative to the min-max optimization model.
Roughly, we say that a point $(x^\star, y^\star)$ is an $\varepsilon$-greedy adversarial equilibrium if $y^\star$ is an $\varepsilon$-approximate local maximum for $f(x^\star,\cdot)$, and $x^\star$ is an $\varepsilon$-approximate local minimum for a ``greedy approximation" to the function $\max_z f(x, z)$ which can be efficiently estimated using second-order optimization algorithms.
We prove the existence of such a point for any smooth function which is bounded and has Lipschitz Hessian. To prove existence, we introduce an algorithm that converges from any starting point to an $\varepsilon$-greedy adversarial equilibrium in a number of evaluations of the function $f$, the max-player's gradient $\nabla_y f(x,y)$, and its Hessian  $\nabla^2_y f(x,y)$, that is polynomial in the dimension $d$, $1/\varepsilon$, and the bounds on $f$ and its Lipschitz constant.

\end{abstract}

\newpage
\tableofcontents
\newpage

\section{Introduction} \label{sec_introduction}

Min-max optimization of functions $f: \mathbb{R}^d \times \mathbb{R}^d \rightarrow \mathbb{R}$, where $f(x,y)$ may be nonconvex and nonconcave in both $x$ and $y$, is an important model for robustness which arises in  optimization and game theory \cite{sion1957game} with  recent applications in machine learning such as generative adversarial networks (GANs) \cite{goodfellow2014generative} 
 and robust training \cite{madry2017towards}. 
Specifically, in a min-max problem, one wishes to find a global min-max point $(x^\star, y^\star)$ that is a solution to the following optimization problem:
 $$\min_{x \in \mathbb{R}^d} \max_{y \in \mathbb{R}^d} f(x,y).$$
 \smallskip
In other words, $$f(x^\star,y^\star) = \max_{y\in \mathbb{R}^d} f(x^\star,y) \ \ \mbox{and} $$ $$ \max_{y\in \mathbb{R}^d} f(x^\star, y) = \min_{x \in \mathbb{R}^d} \max_{y\in \mathbb{R}^d} f(x,y).$$ 
 We consider the setting where $f$ is a $C^2$-smooth nonconvex-nonconcave function that is uniformly bounded by some $b>0$ with $L$-Lipschitz Hessian for some $L >0$, and we are given access to oracles for $f$, its gradient, and its Hessian.
  The setting where $f$ is a bounded function with unconstrained domain arises in many machine learning applications, including generative adversarial networks (GANs).\footnote{In GANs, the objective function value is bounded both above and below if one uses a mean-squared-error loss \cite{mao2017least}.  For the cross-entropy loss \cite{goodfellow2014generative}, the objective function is uniformly bounded above by 0 (but need not be bounded below).}

\smallskip
In the unconstrained setting, a global min-max point may not exist, even when $f$ is bounded above and below. 
However, in the setting where $f$ is bounded, the extreme value theorem guarantees that one can find a point where each player's objective is very close to its min-max optimal value. 
Namely, for every $\epsilon>0$ one can always find a point $(x^\star,y^\star)$ such that 
$$  f(x^\star, y^\star) \geq \max_{y \in \mathbb{R}^d} f(x^\star, y) - \epsilon \ \ \mbox{and}$$
$$
 \max_{y\in \mathbb{R}^d} f(x^\star, y) \leq \min_{x \in \mathbb{R}^d} \max_{y\in \mathbb{R}^d} f(x,y) + \epsilon 
$$

\noindent
 However, even in the special case of minimization, finding a point whose function value is within a fixed  $\epsilon > 0$ of the minimum value is hard.
For instance, this problem is NP-hard in settings when the objective function is given by a depth-2 neural network with %
  mean-squared error loss \cite{manurangsi2018computational}; see also  \cite{blum1989training}.
 The problem of minimizing a function or determining if it can achieve a minimum value of $0$ remains hard when $f$ is uniformly $b$-bounded with $L$-Lipschitz Hessian when one is given access to oracles for the gradient and Hessian of $f$, and requires a number of oracle queries which is exponential in  $d$ (see Section \ref{sec:hardness} for a detailed discussion).
Consequently, there has been  interest in finding computationally tractable alternatives to min-max optimization.

  One popular alternative to min-max optimization  is to consider a model where the min- and max- players are only allowed to make small ``local" updates, rather than requiring each player to solve a global optimization problem \cite{heusel2017gans,daskalakis2018limit,adolphs2018local,daskalakis2020complexity}.
 A stationary point for such a model, sometimes referred to as a local min-max point, is a point where the min-player is unable to decrease the loss, and the max-player is unable to increase the loss, if they are restricted to local updates inside a ball of some small radius. 
 More specifically, for $\epsilon, \delta>0$, an $(\epsilon, \delta)$-local min-max point $(x^\star, y^\star)$ is a point where $ \forall x,y \in \mathbb{R}^d$ such that  $\|y - y^\star\| \leq \delta$  and $\|x - x^\star\| \leq \delta$,
\be
  f(x^\star, y^\star) \geq f(x^\star, y) - \epsilon \ \ \  \mbox{and}  \ \ \ 
 \ f(x^\star, y^\star) \leq f(x, y^\star) + \epsilon.
\ee
One can also define a notion of $(\epsilon, \delta)$-local minimum and maximum point in a similar manner.
 In the ``local" regime where $\delta < O(\sqrt{\epsilon})$, if $f \in [-1,1]$ with $O(1)$-Lipschitz gradient,  any point which is an $(\epsilon,\delta)$-local minimum of the function $f$ with respect to the variable $(x,y)$, or an $(\epsilon, \delta)$-local maximum of $f$ with respect to $(x,y)$, will also be a $(\Omega(\epsilon), O(\delta))$-local min-max point of this function; thus, in this regime the problem of finding local min-max points is equivalent  to the problem of finding a local minimum or maximum point.
Unfortunately, outside of the regime $\delta < O(\sqrt{\epsilon})$, local min-max points may not exist even for functions $f$ where the value of $f \in [-1,1]$ and $f$ is $O(1)$-Lipschitz with $O(1)$-Lipschitz gradient and Hessian.\footnote{For instance, the function $f(x,y) \coloneqq \sin(x+y)$ has no $(\epsilon, \delta)$-local min-max points when $\frac{1}{100} > \delta > \sqrt{\epsilon}$.}
Thus, in the regime  $\delta > \Omega(\sqrt{\epsilon})$ local min-max points are not guaranteed to exist, and, when $\delta < O(\sqrt{\epsilon})$ 
 finding an $(\epsilon, \delta)$-local min-max point is not as interesting since it reduces to finding a minimum (or maximum) point of $f$.

\subsection{Our contributions} 
We depart from prior approaches and make a novel assumption on the adversary (the max-player), namely, that the adversary is computationally bounded.  
This is motivated from real-world applications  where the adversary itself may be an algorithm.
 Roughly, we show that when the adversary is restricted to computing a greedy approximation to the global maximum $\max_z f(x, z)$, a type of equilibrium --{\em greedy adversarial equilibrium}-- always exists and can be found efficiently for general $f$ in time polynomial in $d, b$, and  $L$.
This is in contrast to previous works which seek local min-max points and make strong assumptions on $f$, for instance assuming that $f(x,y)$ is concave in $y$ but possibly nonconvex in $x$ \cite{thekumparampil2019efficient, rafique2018non}, that $f$ is sufficiently bilinear \cite{abernethy2019last}, or that the gradient of $f$ satisfies a monotonicity property \cite{lin2018solving, gidel2018variational}. 

Our greedy adversarial equilibrium builds on the {second-order} notion of approximate local minimum  introduced by \cite{nesterov2006cubic}.\footnote{We often refer to this second-order approximate local minimum by  approximate local minimum.} 
Roughly, a second-order $(\epsilon, \theta)$-approximate local minimum of a function $\psi: \mathbb{R}^d \rightarrow \mathbb{R}$ is a point $x^\star$ which satisfies the following second-order conditions
$$\|\nabla \psi(x^\star)\| \leq \epsilon \ \  \mbox{and} 
 \ \ \lambda_{\mathrm{min}}( \nabla^2 \psi(x^\star)) \geq -\theta.$$
\cite{nesterov2006cubic} and other recent works \cite{agarwal2017finding, allen2018natasha, NonconvexOptimizationForML, carmon2018accelerated, fang2018near, arjevani2020second}, have shown that one can find an $(\epsilon, \theta)$-approximate local minimum in time roughly $\mathrm{poly}\left(\frac{1}{\epsilon},\log d, \frac{1}{\theta}, b, L\right)$ gradient evaluations.

In our model, the min-player is empowered to simulate updates of the max-player by computing a tractable second-order approximation to the global max function $\max_z f(x, z)$, which we refer to as the {\em greedy max function} $g_\epsilon(x,y)$. 
Here, the parameter $\eps>0$ is a measure of approximation.
Ideally, we would like a point $(x^\star, y^\star)$ to be a greedy adversarial equilibrium if $y^\star$ is a second-order approximate local maximum for $f(x^\star, \cdot)$, and $x^\star$ is a second-order approximate local  minimum for  $g_\epsilon(\cdot, y^\star)$.
 However, the function $g_\epsilon(x, y)$ that arises is hard to evaluate and also discontinuous.
  We overcome these issues in part by ``truncating" and ``smoothing'' $g_\epsilon$ by convolving it with a Gaussian $N(0, \sigma^2 I_d)$ for some $\sigma>0$ to obtain a smooth approximation $S_{\epsilon, \sigma}(\cdot)$ to $g_\epsilon(\cdot, y)$.  
  This allows us to apply the definition of approximate local minimum above to $g_\epsilon$, and to obtain our definition of $(\eps,\sigma)$-greedy adversarial equilibrium; see  Definition \ref{def:local_min-max_formal}.
  %

Our main technical result is an algorithm which finds an $(\epsilon, \sigma)$-greedy adversarial equilibrium in a number of gradient, Hessian, and function evaluations that is polynomial in $\frac{1}{\epsilon}, \frac{1}{\sigma}, b, L, d$; see Theorem \ref{thm:GreedyMin-max}.
In particular, providing such an algorithm proves the existence of an approximate greedy adversarial equilibrium.
Our algorithm requires access to a zeroth-order oracle for the value of $f$, and to oracles for the gradient $\nabla_y f(x,y)$ and Hessian $\nabla_y^2 f(x,y)$ for the max-player variable $y$, but {\em not} to oracles $\nabla_ x f(x,y)$ and $\nabla_x^2 f(x,y)$ for the min-player variable $x$.  
Note that the polynomial dependence on $d$ in our bounds comes from the fact that we do not assume that the $x$-player has access to a gradient oracle $\nabla_x f$ or Hessian oracle $\nabla^2_x f$.

\subsection{Discussion of our contributions}

\vspace{-2mm}
\paragraph{Computationally bounded adversaries.}  
 The main conceptual insight in this paper is that we can obtain a model which is an efficient alternative to min-max optimization by placing computational restrictions on the adversary. %
In comparison to models where each player is restricted to local updates, this can allow the model to be robust to a greater diversity of adversaries from a much larger set of parameters $y$ than just the current value of $y$ namely, all the values of $y$ reachable by the tractable approximation--while still allowing for efficient algorithms for modeling the adversary (in particular, this can allow for more stable training of machine  learning algorithms).
We note that analogous computationally bounded adversaries can lead to useful models in many other settings including coding theory  \cite{lipton1994new, micali2005optimal, 5671339} and cryptography \cite{goldwasser1984probabilistic, canetti2000security}.

\vspace{-3mm}
\paragraph{Results hold for any bounded and Lipschitz $f$.}  Aside from the bounded and Lipschitz assumptions, Theorem \ref{thm:GreedyMin-max} does not make any additional assumptions on $f$. 
As mentioned earlier, prior results which seek solutions to other alternative models to min-max optimization (such as local min-max points), assume that either $f(x,y)$ is concave  \cite{thekumparampil2019efficient, rafique2018non} in $y$, or monotone \cite{lin2018solving, gidel2018variational}, or sufficiently bilinear \cite{abernethy2019last}.
  Although there are other prior works which do not assume that $f$ is convex-concave or monotone, many of these works instead assume that there exists a stationary point for their algorithm on the function $f$, and that their algorithm is initialized somewhere in the region of attraction for this stationary point \cite{heusel2017gans, adolphs2018local, wang2020ridge, mazumdar2019finding}. 
  In contrast, Theorem \ref{thm:GreedyMin-max} guarantees that our algorithm  converges from {\em any} initial point $(x,y).$

\vspace{-3mm}
\paragraph{Extension of second-order local minimum definition to discontinuous functions.} 
 To handle minimization of the discontinuous greedy max function $g_\epsilon$, when defining our greedy adversarial equilibrium we introduce a second-order notion of approximate local minimum which applies to discontinuous functions. 
  This leads to an algorithm which, in the special case when the objective function depends only on $x$, reduces to a ``derivative-free" minimization algorithm, that is, it does not require access to derivatives of the objective function.
  The novel techniques and definitions we develop here for minimization of the greedy max function may be of interest to other problems in discontinuous or derivative-free minimization of non-convex objectives (for applications of derivative-free methods to adversarial bandit convex optimization, see for instance \cite{flaxman2005online}).

\vspace{-3mm}
\paragraph{Greedy adversarial equilibria corresponds to global min-max under strong convexity/concavity.}   If $f$ is $1$-strongly convex- strongly concave, then for  $\epsilon>0$ and small enough $\sigma$, we show that at any $(\epsilon, \sigma)$-greedy adversarial equilibrium $(x^\star, y^\star)$ the duality gap satisfies $$\max_{y \in \mathbb{R}^d} f(x^\star, y)-\min_{x \in \mathbb{R}^d} f(x, y^\star) \leq O(\epsilon^2);$$ see Theorem \ref{Thm_strongly_convex_concave} in Appendix \ref{sec:convex_concave} for a precise statement of this result and its proof.

\vspace{-3mm}
\paragraph{Applications to GANs.}  In a subsequent paper, \cite{OurAppliedPaper} use a  related first-order version of our greedy adversarial equilibrium to obtain an algorithm and show that it can enable more stable training of generative adversarial networks (GANs).
Roughly speaking, the first-order equilibrium in \cite{OurAppliedPaper} is a point $(x^\star,y^\star)$ such that 
$\|\nabla_y f(x^\star, y^\star)\| \leq \epsilon$ and $\|\nabla_x \hat{g}(x^\star, y^\star)\| \leq \epsilon,$
where $\hat{g}$ is a first-order approximation to the greedy-max function. 
This means that, unlike here, in \cite{OurAppliedPaper} min-min points (points where both players are at a local minimum) are included in the local equilibrium proposed.
Including second-order conditions for both the maximizing and minimizing players in our Definition \ref{def:local_min-max_formal}  allows us to ensure that our definition excludes points which may be (approximate) min-min points.
The second-order conditions also end up making the proofs in this paper significantly harder.

\vspace{-3mm}
\paragraph{Difference between constrained and unconstrained settings.}
Finally, we note that, in a subsequent work, \cite{daskalakis2020complexity} prove PPAD-hardness results for finding approximate local min-max points in the constrained setting.
Their result does not have any implication to our framework as we consider the unconstrained setting (domain is $\mathbb{R}^d \times \mathbb{R}^d$).

\section{Greedy adversarial equilibrium} \label{sec_Greedy_Definitions}
\paragraph{Preliminaries.}
 In the following,  we say that a function is $C^2$-smooth if its second derivatives are continuous on its domain.
  $\lambda_{\mathrm{max}}(A)$ denotes the largest eigenvalue of any square matrix $A$, and $\lambda_{\mathrm{min}}(A)$ is its smallest eigenvalue.  
   $\| \cdot \|$ denotes the Euclidean $\ell_2$ norm, and $\|A\|_{\mathrm{op}} = \sup_{v \neq 0} \frac{v^\top A v}{\|v\|^{2}} $ the operator norm of any square matrix $A$. 
We assume\footnote{We note that a uniform bound on a function and the Lipschitz constant of its Hessian also implies a bound on the Lipschitz constants of the function and its gradient.
Namely, if $f$ is $b$-bounded with $L_{}$-Lipschitz Hessian, it is also $L_1$-Lipschitz with  $L_1 \leq 4 b^{\nicefrac{2}{3}} L^{\nicefrac{1}{3}}$ and has $L_2$-Lipschitz gradient with $L_2 \leq 2 b^{\nicefrac{1}{3}} L_{}^{\nicefrac{2}{3}}$.
 We say $f: \mathbb{R}^d \times \mathbb{R}^d \rightarrow \mathbb{R}$ is $L_1$-Lipschitz if $|f(x,y) -  f(\tilde{x},\tilde{y})| \leq L_1 \sqrt{\| x- \tilde{x}\|^2 + \| y- \tilde{y}\|^2}$, and that $f$ has $L_2$-Lipschitz gradient if $\|\nabla f(x,y) - \nabla f(\tilde{x},\tilde{y})\| \leq L_2 \sqrt{\| x- \tilde{x}\|^2 + \| y- \tilde{y}\|^2}$.} that for some $b, L_{}>0$,  $f:\mathbb{R}^d \times \mathbb{R}^d \rightarrow \mathbb{R}$ is $b$-bounded, i.e.,  $|f(x,y)| \leq b$, and has $L_{}$-Lipschitz Hessian: $$\|\nabla^2 f(x,y) - \nabla^2 f(\tilde{x},\tilde{y})\|_{\mathrm{op}} \leq L_{} \sqrt{\| x- \tilde{x}\|^2 + \| y- \tilde{y}\|^2}.$$

\noindent
 We start by considering the special case of minimization.
 We say that a point $x^\star \in \mathbb{R}^d$  is an {\em exact} local minimum point of  a function  $\psi: \mathbb{R}^d \rightarrow \mathbb{R}$ if there exists $\delta>0$ such that 
\noindent
\begin{equation} \label{eq:ExactLocalMin_formal}
\psi(x^\star) \leq \psi(x), \ \  \forall x \in \mathbb{R}^d\textrm{ such that }\|x - x^\star \| \leq \delta.
\end{equation}
Unfortunately, even if the objective function $\psi:\mathbb{R}^d \rightarrow \mathbb{R}$ is bounded and Lipschitz, it is not always possible to find an exact local minimum for $\psi$ in $\mathrm{poly}(d)$ gradient evaluations (see Remark \ref{local_minimum_hardness} in Section \ref{sec:hardness} for a detailed discussion).

On the other hand, suppose we just wanted to minimize a function $\psi$, and we start from any point $x$ where $$\|\nabla \psi(x)\| > \epsilon \ \  \mbox{or} \ \   \lambda_{\mathrm{min}}( \nabla^2 \psi(x)) < -\theta$$ for some $\epsilon, \theta>0$.
Then we can always find a direction to travel in along which either $\psi$ decreases rapidly at a rate of at least $\epsilon$, or the second derivative of $\psi$ is less than $-\theta$ (see Remark \ref{rem_escaping_saddle_points}). 
 By searching in such a direction we can easily find a new point which has a smaller value of $\psi$ using only local information about the gradient and Hessian of $\psi$.
 This means that we can keep decreasing $\psi$ until we reach a point where $\|\nabla \psi(x)\| \leq \epsilon$ and $\lambda_{\mathrm{min}}( \nabla^2 \psi(x)) \geq -\theta$.
  If $\psi$ is Lipschitz smooth and bounded, we will reach such a point in polynomial time from any starting point \cite{nesterov2006cubic, ge2015escaping}.
This fact, together with the fact that any point which satisfies these conditions for $\eps = \theta = 0$ is also an exact local minimum, motivates the second-order notion of an approximate local minimum of  \cite{nesterov2006cubic}. %
    For any $\epsilon, \theta \geq 0$, say that a point $x^\star$  is an $(\epsilon, \theta)$-approximate local minimum for  a $C^2$-smooth function $\psi: \mathbb{R}^d \rightarrow \mathbb{R}$ if
    \noindent
\begin{equation} \label{eq:ApproximateLocalMin_formal}
 \|\nabla \psi(x^\star)\| \leq \epsilon \qquad \mbox{and} \quad \lambda_{\mathrm{min}}( \nabla^2 \psi(x^\star)) \geq -\theta.
\end{equation}
 We say that $x^\star$ is an $(\epsilon, \theta)$-approximate local {\em maximum} of $\psi$ if  $x^\star$ is an $(\epsilon, \theta)$-approximate local {\em minimum} of $-\psi$.
 We use two different values of $\theta$: when referring to an $(\epsilon, \theta)$-approximate local maximum on $f(x,\cdot)$, we  use  $\theta = \sqrt{L_{} \epsilon}$ and, roughly, when defining an $(\epsilon, \theta)$-approximate local minimum on $g_\epsilon$, we  use  $\theta = \sqrt{\epsilon}$.  We explain these choices of $\theta$ in Remark  \ref{rem_escaping_saddle_points} .

Importantly, one can view the definition given by Inequality \eqref{eq:ApproximateLocalMin_formal} as being motivated by a class of second-order optimization algorithms as,
 roughly speaking, a second-order optimization algorithm can rapidly decrease the value of $\psi$ when starting from any point which is {\em not} an approximate local minimum.

\subsection{Greedy path and greedy max}

When defining a greedy path, we restrict the max-player to updating $y$ by traveling along continuous paths which start at the current value of $y$ and along which either $f$ is increasing or the second derivative of $f$ is positive.  %

\begin{definition}[{Greedy path}] \label{def_greedy_path}
Let $x\in \mathbb{R}^d$, and suppose a continuous path $\varphi_t:[0,\tau] \rightarrow \mathbb{R}^d$ is differentiable except at a finite number of points, and at the points where it is differentiable $\left\|\frac{\mathrm{d}}{\mathrm{d} t} \varphi_t \right\| = 1$ (i.e., the path travels at unit speed).  Then for any $\epsilon \geq 0$, we say $\varphi$ is an $\epsilon$-greedy path for $f(x, \cdot)$ if at all points where $\varphi$ is differentiable $\frac{\mathrm{d}}{\mathrm{d}t} f(x, \varphi_t)\geq - \epsilon$ and
\begin{equation} \label{eq_greedy_path}
 \frac{\mathrm{d}}{\mathrm{d}t} f(x, \varphi_t) >\epsilon \qquad \textrm{ or } \qquad
 \frac{\mathrm{d}^2}{\mathrm{d}t^2} f(x, \varphi_t) > \sqrt{L_{} \epsilon}.
\end{equation}
\end{definition}

\noindent Roughly speaking, when restricted to updates obtained from $\epsilon$-greedy paths, the max-player will always be able to reach a point which is a second-order $(\epsilon, \sqrt{L \epsilon})$-approximate local maximum for $f(x,\cdot)$, although there may not be an $\epsilon$-greedy path which leads the max-player to a global maximum (Figure \ref{fig_greedy_paths}).
\begin{figure}
\includegraphics[scale=0.3]{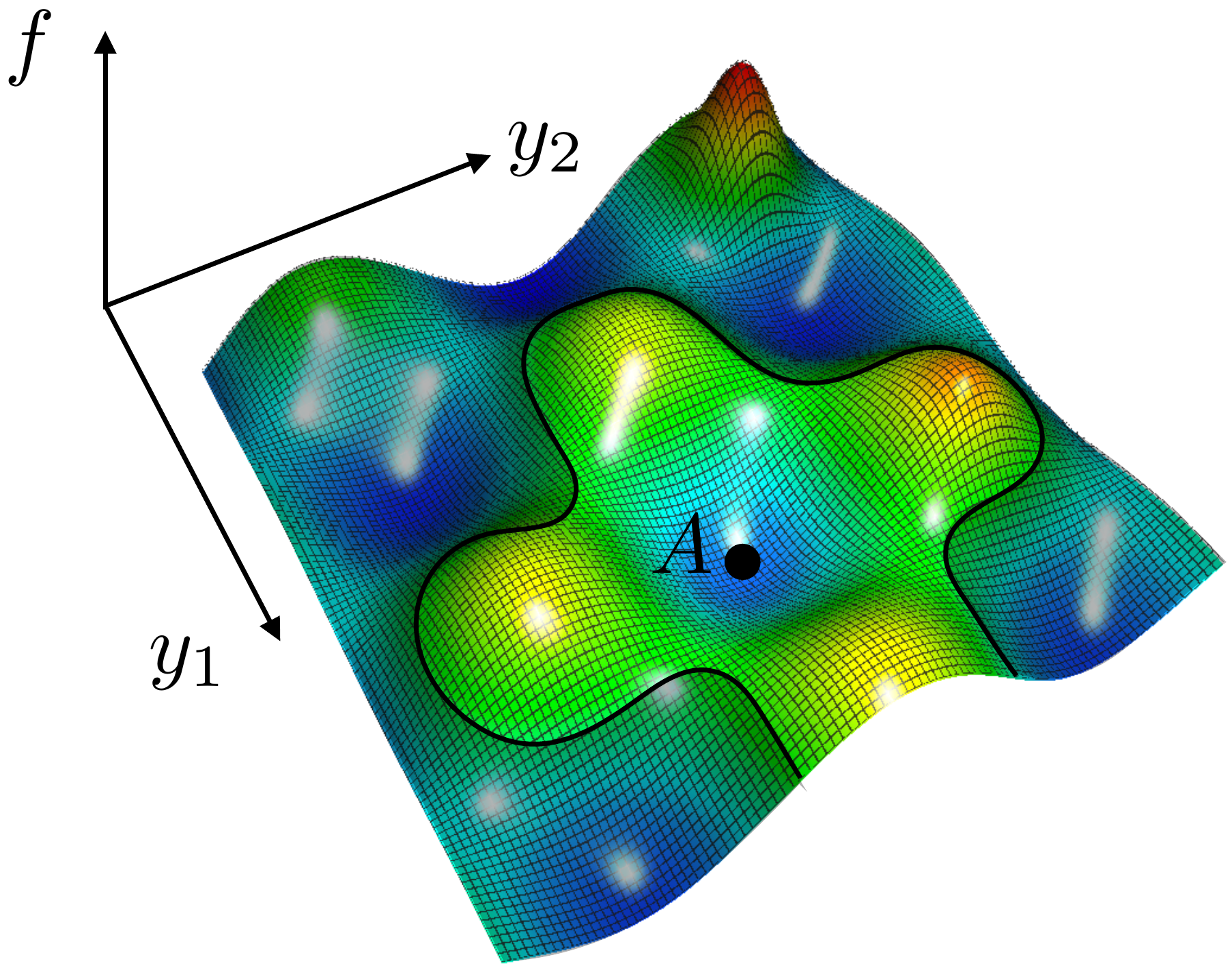}
\includegraphics[scale=0.3]{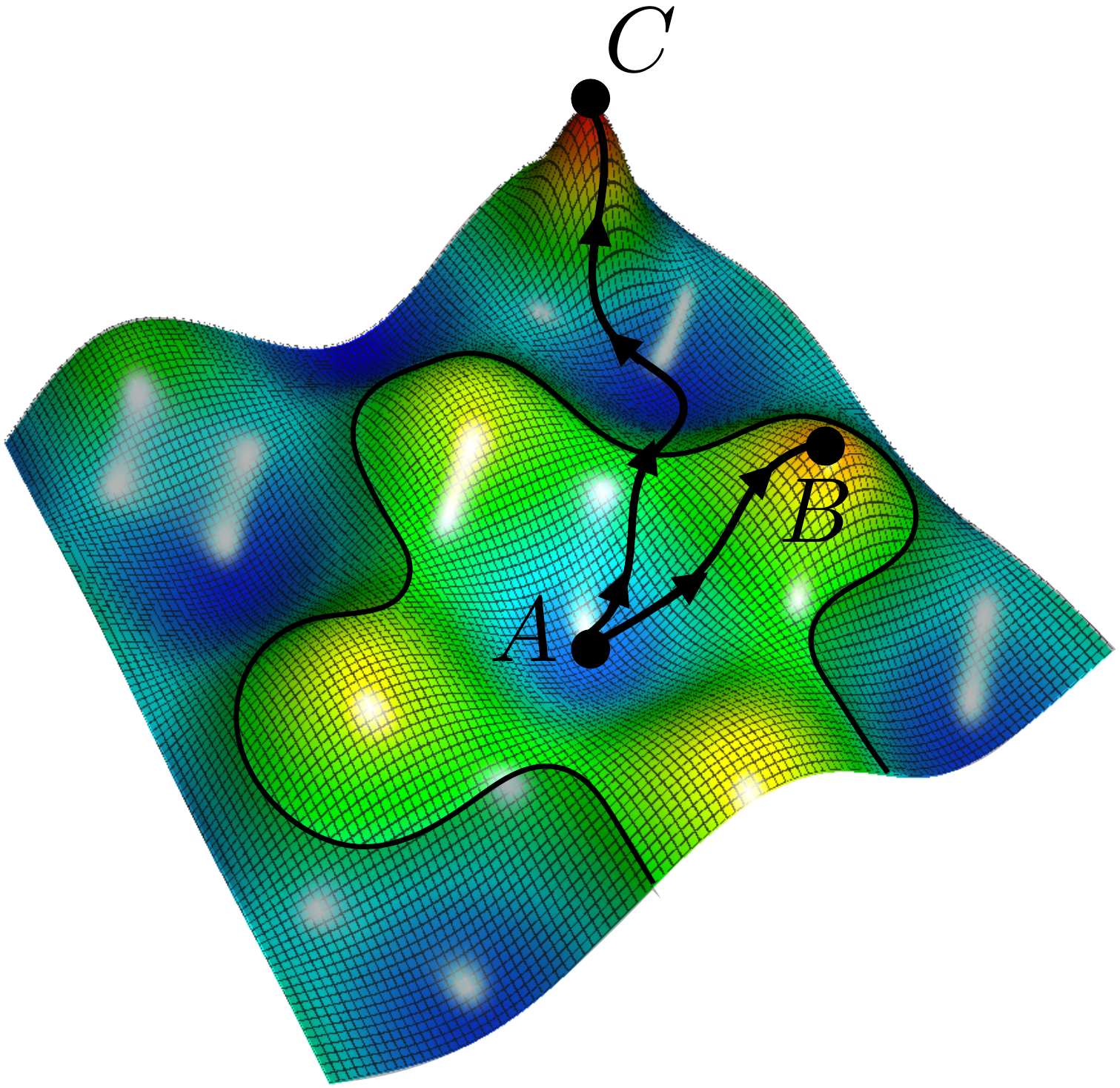}
\includegraphics[scale=0.3]{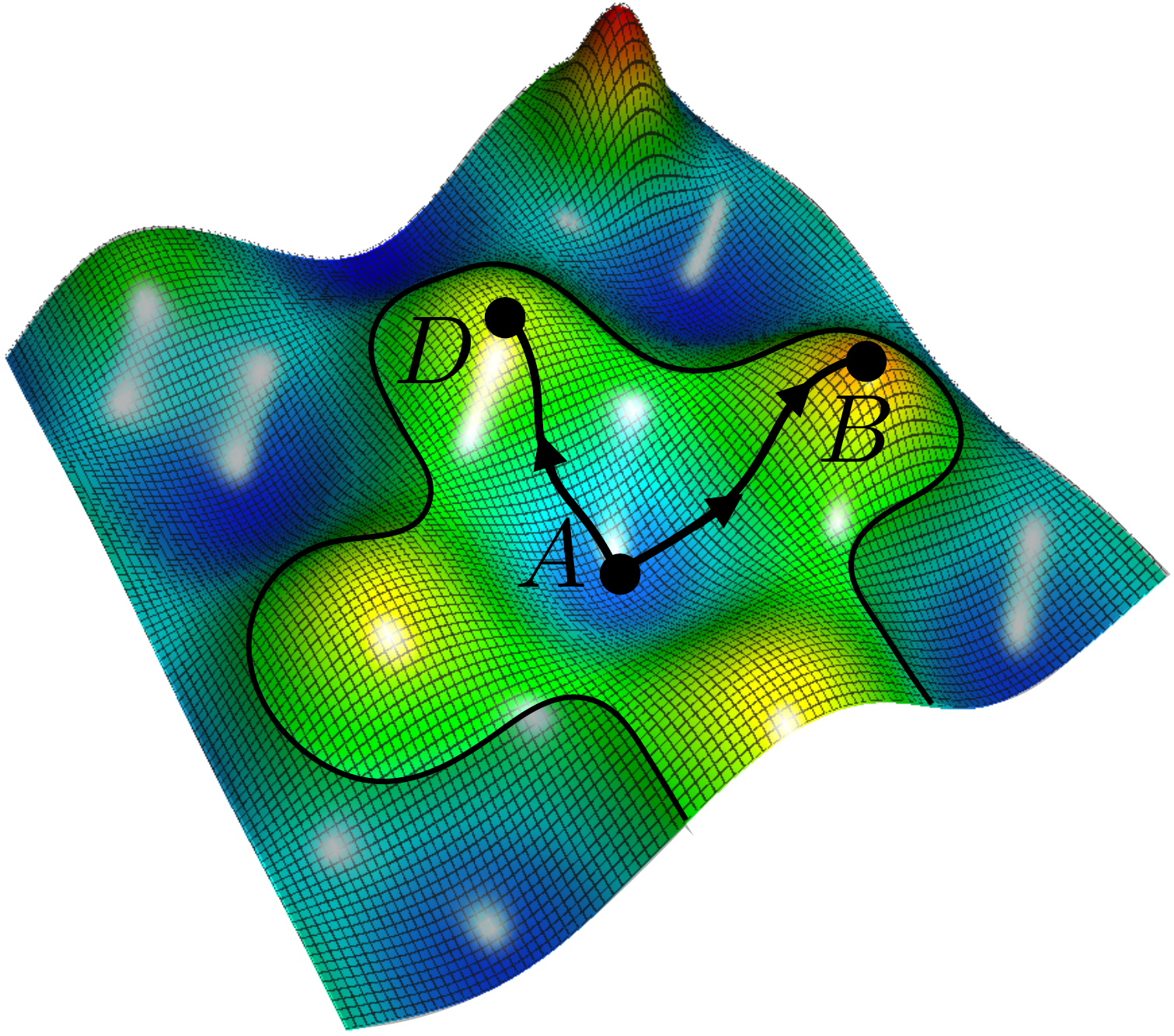}
\caption{\small Left figure: The left figure shows the region reachable by greedy paths starting from an initial point $A$ (the lighter, non-shaded region), for $\epsilon = 0$, on the function $f(x,y) =  - \sin(y_1)\cos(y_2) + 2 e^{-\frac{(y_1 - 2)^2 + y_2^2}{10}} - 1.5 e^{-\frac{(y_1 - 2)^2 + y_2^2}{4}} + 0.5 e^{-2((y_1 - 2)^2 + (y_2-3)^2)} + 2 e^{-2((y_1 +3.8)^2 + (y_2-3.3)^2)}  $  (for simplicity, in this figure we have chosen a function which has no dependence on $x$). 
Middle figure: The point $B$ with the largest value of $f$ that is reachable from a greedy path starting at $A$ (black curve from $A$ to $B$) is a local maximum of $f(x, \cdot)$.  However, in this example, to reach the global maximum point $C$ when starting form the point $A$, one must take a path which is not greedy and where the value of $f$ may decrease over long stretches (black curve from $A$ to $C$).  
 Right figure: There are many different greedy paths of maximal length that start at the same point $A$ but which end up at different local maxima.  Two of these paths are shown here, with one greedy path reaching the local maximum point $B$ and a different greedy path reaching the local maximum point $D$, which has a smaller value of $f$ than the point $B$. 
  The local maximum at $B$ is the maximum value attainable from  {\em any} greedy path starting at $A$, and the value of $f$ at the end of this path determines the value of the greedy max function $g_\epsilon(x,A) = f(x, B)$.
}\label{fig_greedy_paths}
\end{figure}

To define an alternative to $\max_z f(\cdot,z)$, we consider the local maximum point with the largest value of $f(x,\cdot)$ attainable from a given starting point $y$ by any $\epsilon$-greedy path.
Towards this end, we define the set $S_{\epsilon, x, y} \subseteq \mathbb{R}^d$ of endpoints of $\epsilon$-greedy paths, for any $x,y \in \mathbb{R}^d$ and $\epsilon>0$. 
   We say that a point $z \in S_{\epsilon, x, y}$ if there is a number $\tau \geq 0$ and a path $\varphi:[0,\tau] \rightarrow \mathbb{R}^d$ which is an $\epsilon$-greedy path for $f(x,\cdot)$, with initial point $\varphi_0 = y$ and endpoint $\varphi_\tau = z$.
The greedy max function $g_\epsilon(x, y)$ is the maximum value of $f(x, \cdot)$ attainable by {\em any} $\epsilon$-greedy path in the set  $S_{\epsilon, x, y}$:
%
\begin{equation} \label{eq_greedy_max}
 g_\epsilon(x,y)\coloneqq \sup \{f(x,z) : z\in  S_{\epsilon, x, y}\}.
\end{equation}

\begin{remark} [Greedy paths can escape saddle points and local minima] \label{rem_escaping_saddle_points}
Equations \eqref{eq_greedy_path} together ensure that for any  $y$ where either 
 (i) the gradient  $\nabla_y f(x,y)$ has magnitude greater than $\epsilon$ or 
 (ii) the eigenvalues of the Hessian $\nabla_y^2 f(x,y)$ are bounded below by $-\sqrt{L_{} \epsilon}$, there is always a unit-speed greedy path (with parameter $\epsilon$) starting at $y$ which can increase the value of $f$ at an average rate
 \footnote{By ``average rate" of at least $\frac{1}{2} \epsilon$ we mean that the increase in $f$ divided by the length of the path is $\geq \frac{1}{2} \epsilon$.} 
 of at least $\frac{1}{2} \epsilon$ by traveling a distance of at most $\frac{1}{2}\frac{\sqrt{\epsilon}}{\sqrt{L_{}}}$. 
 Moreover, since one such greedy path is always a straight line in the direction of either the gradient $\nabla_y f(x,y)$ or 
  the largest eigenvector of $\nabla_y^2 f(x,y)$, all one needs to compute such a path is access to the gradient and Hessian of $f(x,\cdot)$. 
 This fact can also be viewed as a motivation for the definition of approximate local maximum (Inequality \eqref{eq:ApproximateLocalMin_formal}): roughly, any point
   which does not satisfy both conditions (i) and (ii) (up to a constant factor) is an approximate local maximum. 
   Thus, starting from any point $y$ which is not an approximate local maximum of $f(x, \cdot)$ (with parameters $(\epsilon, \sqrt{L_{} \epsilon})$), 
    there is always an easy-to-compute greedy path (with parameter $\epsilon$) which allows one to increase the value of $f$.
\end{remark}

\subsection{Dealing with discontinuities and other difficulties of the greedy max function} \label{sec_discontinuities1}
    Unfortunately, even if $f$ is smooth, the greedy max function may not be differentiable with respect to $x$ and may even be discontinuous (see Example \ref{discontinuity_example} for a simple example of a smooth function $f$ whose greedy max function is discontinuous).  %
This lack of smoothness creates a problem, since the current definition of approximate local minimum (Inequality \eqref{eq:ApproximateLocalMin_formal}) only applies to $C^2$-smooth functions. 
 To solve this problem we would ideally like to smooth a discontinuous function by convolution with a Gaussian (see Section \ref{sec_discontinuous_localmin} for further discussion on why we use convolution for smoothing).
    
    Another difficulty is that the value of $g_\epsilon(x,y)$ may be intractable to compute at some points $(x,y)$, since one may need to compute a very large number\footnote{Even in the setting where $f$ is bounded with Lipschitz Hessian, the number of $\epsilon$-greedy paths that share a given starting point may be infinite.} of $\epsilon$-greedy paths (possibly an infinite number of paths), each with the same initial point $y$, before finding the $\epsilon$-greedy path with the largest value of $f$.  
    This is because, starting from a point near a local minimum or saddle point, there may be many directions to choose from which allow one to increase the value of $f$, and, depending on which direction one chooses, one may end up at a different local maximum. 
   Realistically, this means that in general we cannot hope to give our algorithm access to the exact value of $g_\epsilon$.
   Our algorithm overcomes this by instead computing a lower bound $h_\epsilon$ for $g_\epsilon$, and uses only access to this lower bound to minimize $g_\epsilon$ (In Sections \ref{sec_local_min_local_max}-\ref{sec_exact_local_min} of our technical overview we show how this can be done by using some additional properties of the greedy max function). 
      To allow us to handle this more difficult setting, %
       we would like our notion of approximate local minimum to satisfy the property that any point which is an {\em exact} local minimum, is also an {\em approximate} local minimum under our definition.
    
Unfortunately, convolution can cause the local minima of a function to ``shift"-- a point which is a local minimum for a function $\psi: \mathbb{R}^d \rightarrow \mathbb{R}$  may no longer be a local minimum for the convolved version of $\psi$ (for instance, in Example \ref{ex_shifted_min1}, we show that this happens if we convolve the function $\psi(x) = x - 3x \mathbbm{1}(x\leq 0) + \mathbbm{1}(x \leq 0)$ with a Gaussian $N(0,\sigma^2)$ for any $\sigma>0$).
To avoid this, we instead consider a ``truncated" version of $\psi$, and convolve this function with a Gaussian to obtain our smoothed version of $\psi$ (Definition \ref{def_discontinuos_local_min}).

  \begin{definition} [Approximate local minimum for discontinuous functions]  \label{def_discontinuos_local_min}
  For any $\epsilon, \sigma \geq 0$, we say that $x^\star$ is an $(\epsilon, \sigma)$-approximate local minimum for a uniformly bounded function $\psi$ if
  \begin{equation} \label{eq_local_lemma2}
 \| \nabla_{x}  \mathcal{S}(x^\star) \| \leq \epsilon \qquad \textrm{ and } \qquad
   \lambda_{\mathrm{min}}( \nabla^2_{x}   \mathcal{S}(x^\star) ) \geq -\sqrt{\epsilon},
  \end{equation}
  where $\mathcal{S}(x) \coloneqq \mathbb{E}_{\zeta \sim N(0,I_d)}\left [\min(\psi(x  + \sigma \zeta),  \psi(x^\star)) \right].$
  \end{definition}

\begin{example}[A simple example of a discontinuous greedy max function] \label{discontinuity_example}
As a simple example  (Figure \ref{fig:discontinuity}), consider the function
\begin{equation*}
    f(x,y) = \cos(x+y)\sin(2x+2y) - e^{-x^2}.
    \end{equation*}
 For any $0< \epsilon < 0.1$, the greedy max function $g_{\epsilon}(x,y)$ is discontinuous at the (parallel) lines $x+y = -2.52$ and $x+y = -0.62$, with $g_{\epsilon} (x,y) = -e^{-x^2}$ in the region enclosed between the two lines and $g_{\epsilon} (x,y) = -e^{-x^2}+0.77$ on each side of that region. 
  Such examples are easy to come by and extend to higher dimensions.
  \end{example}
  
  \begin{figure}
\includegraphics[scale=0.35]{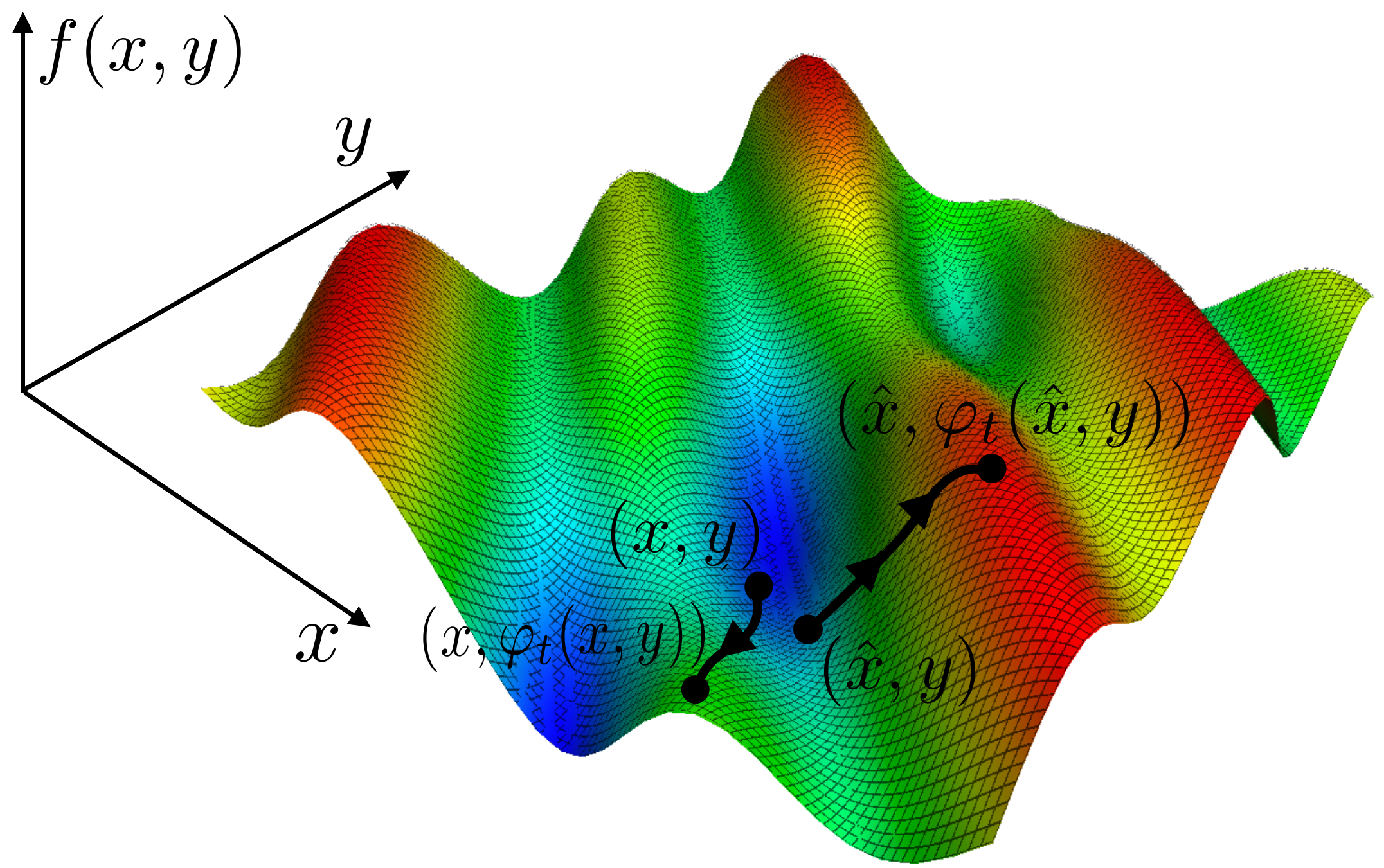}
\includegraphics[scale=0.35]{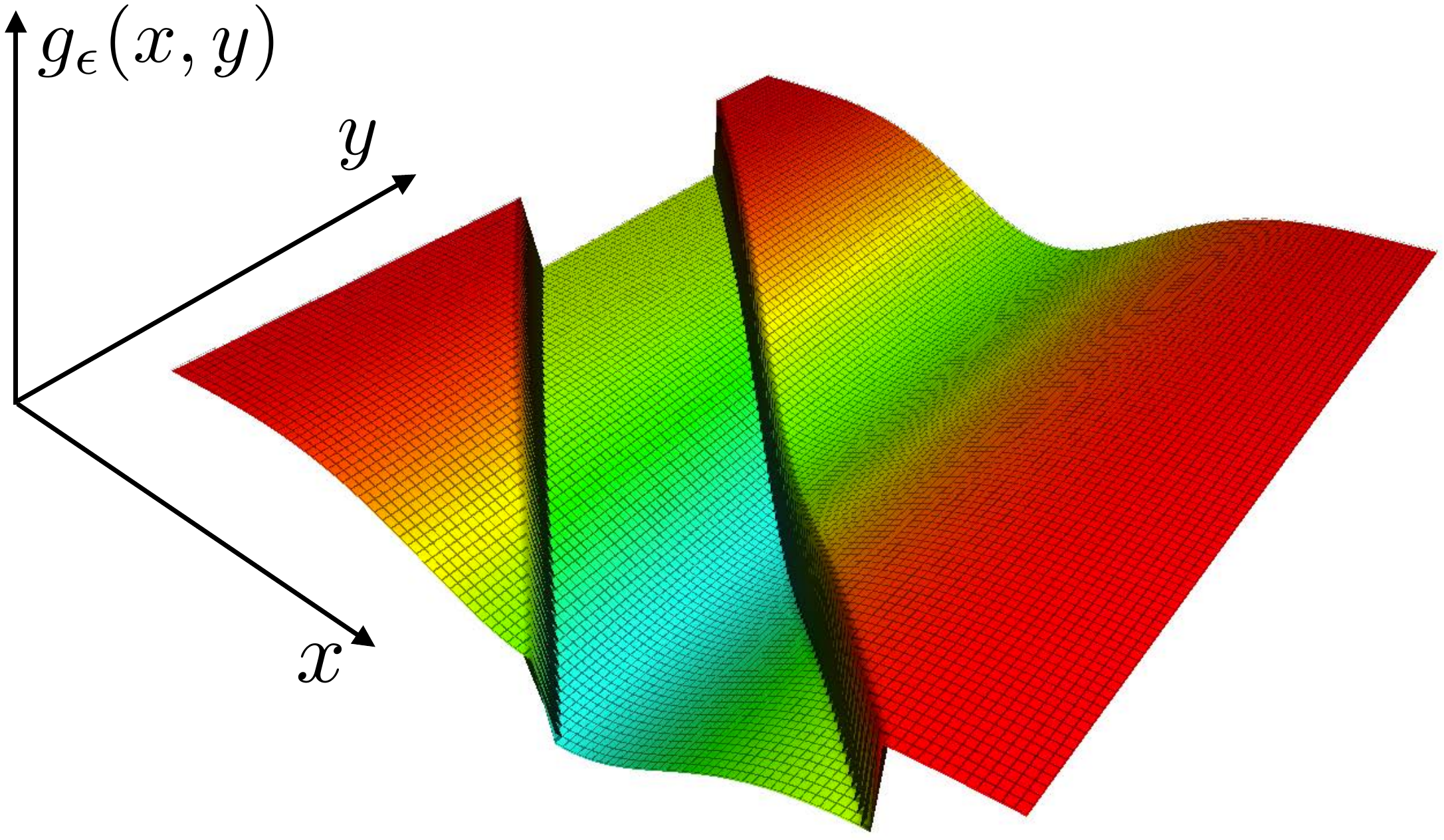}
\caption{\small In this example we have $f(x,y) = \cos(x+y) \sin(2x + 2y) - e^{-x^2}$ (left). 
 We see that if we change $x$ from one value $x$ to a very close value $\hat{x}$, the ``best" greedy path (i.e., the greedy path whose endpoint has the largest value of $f$) undergoes a very large, discontinuous change. 
  This implies that the ``greedy max" function $g_{\epsilon} (x,y)$ (right) is discontinuous in $x$ along the parallel lines $x+y = -2.52$ and $x+y = -0.62$.}\label{fig:discontinuity}
\end{figure}

\subsection{Greedy adversarial equilibrium} \label{sec_greedy_equilibirum}
  We say that $(x^\star, y^\star)$ is an $(\epsilon, \sigma)$-greedy adversarial equilibrium of a function $f: \mathbb{R}^d \times \mathbb{R}^d \rightarrow \mathbb{R}$ with $L$-Lipschitz Hessian if $y^\star$ is an $(\epsilon, \sqrt{L\epsilon})$-approximate local maximum of $f(x^\star, \cdot)$ (in the sense of Inequality \eqref{eq:ApproximateLocalMin_formal}), and if $x^\star$  is an $(\epsilon, \sigma)$-approximate local minimum of the (possibly) discontinuous function $g_\epsilon(\cdot, y^\star)$ (in the sense of Definition \ref{def_discontinuos_local_min}). 
  See Figure \ref{fig_GAE_example} for an example of greedy adversarial equilibria.

\begin{definition}[ Greedy adversarial equilibrium] \label{def:local_min-max_formal} 
For any $\epsilon, \sigma \geq 0$, we say that $(x^\star, y^\star)\in \mathbb{R}^d \times \mathbb{R}^d$ is an $(\epsilon, \sigma)$-greedy adversarial equilibrium of a $C^2$-smooth function $f: \mathbb{R}^d \times \mathbb{R}^d \rightarrow \mathbb{R}$ with $L$-Lipschitz Hessian, if we have
\begin{equation} \label{eq:first_order_y}
 \|\nabla_y f(x^\star, y^\star)\| \leq \epsilon
 \quad \textrm{and} \quad
 \lambda_{\mathrm{max}}( \nabla^2_{y} f(x^\star, y^\star)) \leq \sqrt{L \epsilon}, \textrm{ and }
\end{equation}
%
  \begin{equation} \label{eq:first_order_x}
  \| \nabla_{x} S(x^\star)\| \leq \epsilon  \quad \textrm{ and }  \quad
 \lambda_{\mathrm{min}}( \nabla^2_{x} S(x^\star)) \geq -\sqrt{ \epsilon},
  \end{equation} \label{eq_S}
     where $S(x) \coloneqq  \mathbb{E}_{\zeta \sim N(0,I_d)}\left [\min(g_\epsilon(x  + \sigma \zeta, y^\star),  g_\epsilon(x^\star, y^\star)) \right]$.
\end{definition}

\noindent
We can view the point $(x^\star, y^\star)$ in Definition \ref{def:local_min-max_formal} as a type of equilibrium.
 Namely, suppose that the max-player can only make updates in the set $S_{\epsilon, x, y^\star}$ of points attainable by an $\epsilon$-greedy path initialized at $y^\star$.  Then under this constraint, the max-player cannot make any update to $y^\star$ that will increase the value of $f(x^\star, \cdot)$. 
  Moreover, we have that $x^\star$ is an $(\epsilon, \sigma)$-approximate local minimum (in the sense of Definition \ref{def_discontinuos_local_min}) of the function $\max_z f(x,z)$ if the maximum is taken over the set  $S_{\epsilon, x, y^\star}$ of updates available to the max-player.

A key feature of greedy adversarial equilibrium is that it empowers the min-player to simulate the updates of the max-player via a class of second-order optimization algorithms which we model using greedy paths.  This is in contrast to previous models, such as the local min-max point considered in \cite{daskalakis2018limit, heusel2017gans, adolphs2018local} or  \cite{minmax_Jordan}, which restrict the min-player and max-player to making updates inside a small ball.

\begin{figure}
\begin{center}
\includegraphics[scale=0.38]{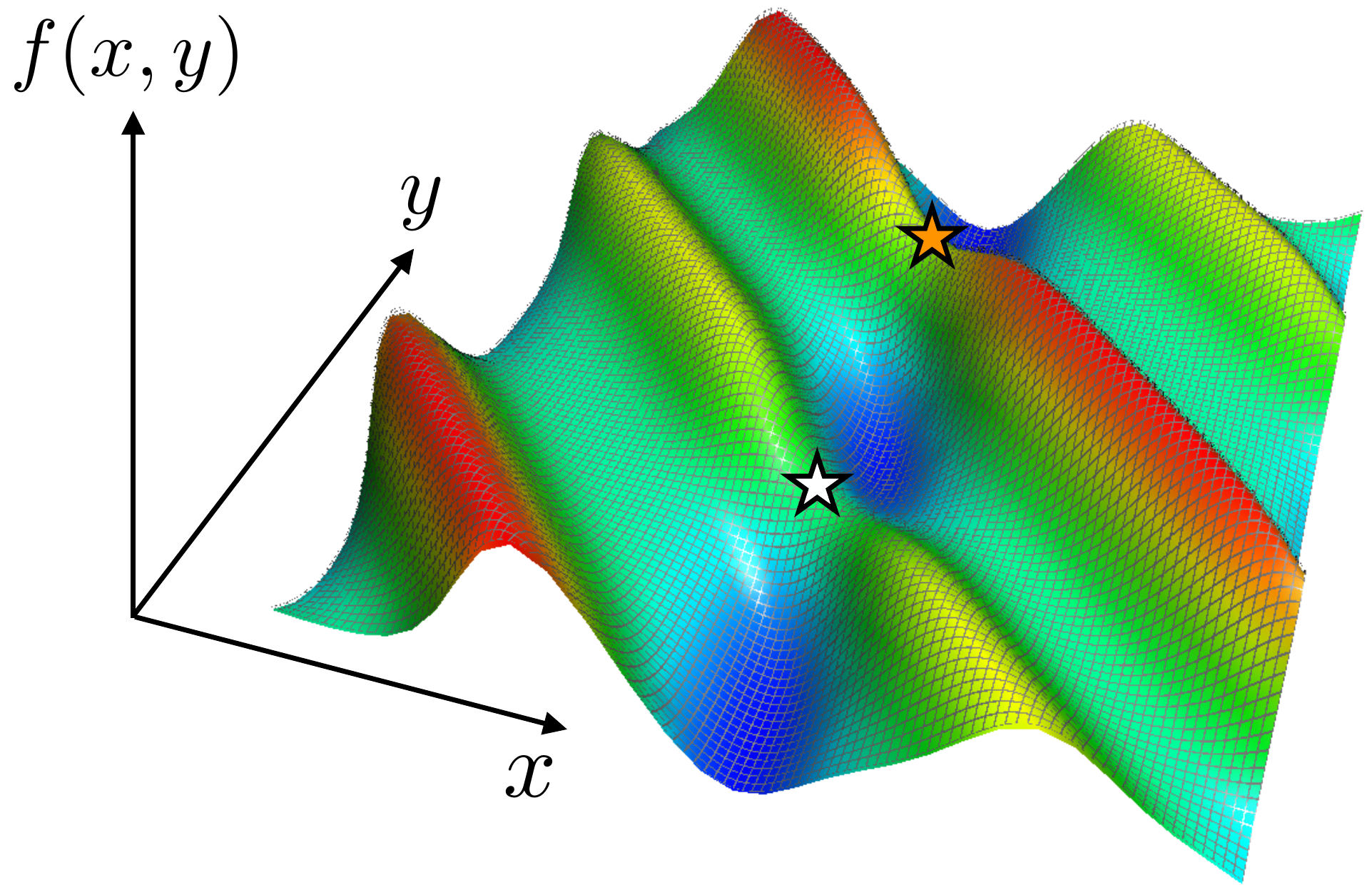}
\end{center}
\caption{\small In this example we have the periodic function $f:\mathbb{R}^1\times \mathbb{R}^1 \rightarrow \mathbb{R}$ where $f(x,y) = \sum_{k \in \mathbb{Z}} 1.2 e^{-(x+y+2 +9k)^2} + 2 e^{-(x+y-2 + 9k)^2}-e^{-(x +6k)^2}$.
This function has $(\eps,\sigma)$-greedy adversarial equilibria (for $\eps=0$ and any $0<\sigma \leq \frac{1}{100}$) at the points in the set $\bigcup_{k,\ell \in \mathbb{Z}} \{(6k, 2- 9 \ell)\} \cup \{(6k, 2+ 9 \ell) \}$; the two equilibria visible in the region shown in the figure are at the points $(0,-2)$ (white star) and $(0,2)$ (orange star).
In particular, the set of greedy adversarial equilibria at the points $\bigcup_{k,\ell \in \mathbb{Z}} \{(6k, 2+ 9 \ell) \}$ (including the orange star at $(0,2)$ in the figure) exactly coincides with the set of global min-max point of $f$.
On the other hand, $f$ has no $(\eps, \theta)$-local min-max points for any $\eps, \theta \leq 10^{-3}$: there does not exist a point $(x,y)$ where $y$ is an $(\eps, \theta)$-approximate local maximum of $f(x,\cdot)$ and $x$ is an $(\eps, \theta)$-approximate local minimum of $f(\cdot,y)$  (as defined in Equation \eqref{eq:ApproximateLocalMin_formal}).
}\label{fig_GAE_example}
\end{figure}

\section{Main result} \label{sec_main_result}

\begin{theorem}[Main result]\label{thm:GreedyMin-max}
Let $\epsilon, \sigma>0$, with $\sigma \leq \frac{1}{\sqrt{\epsilon d}}$, and consider any $C^2$-smooth  uniformly bounded function $f: \mathbb{R}^d \times \mathbb{R}^d \rightarrow \mathbb{R}$ with Lipschitz Hessian. 
 Then there exists a point $(x^\star, y^\star) \in \mathbb{R}^d \times \mathbb{R}^d$  which is an ($\epsilon^\star, \sigma)$-greedy adversarial equilibrium for $f$, for some  $\epsilon^\star \leq \epsilon$.
Moreover, there exists an algorithm which, given access to oracles for the value of a $C^2$-smooth function $f: \mathbb{R}^d \times \mathbb{R}^d \rightarrow [-b,b]$, and to oracles for $\nabla_y f$ and $\nabla^2_y f$, where $f$ has $L_{}$-Lipschitz Hessian for some $b, L_{} >0$, and numbers $\epsilon, \sigma \geq 0$, with probability at least $\frac{9}{10}$ generates a  point $(x^\star,y^\star) \in \mathbb{R}^{d} \times \mathbb{R}^{d}$ which is an $(\epsilon^\star, \sigma)$-greedy adversarial equilibrium for $f$, for some  $\epsilon^\star \leq \epsilon$. 
 Moreover, this algorithm takes a number of gradient, Hessian, and function evaluations which is polynomial in $\frac{1}{\epsilon}, d, b, L_{}, \frac{1}{\sigma}$.
\end{theorem}

\noindent
As noted earlier, our result does not require additional assumptions on $f$ such as convexity or monotonicity \cite{thekumparampil2019efficient, rafique2018non, nouiehed2019solving} or sufficient bilinearity \cite{abernethy2019last}.
Our algorithm also converges from any initial point. 
 This is in contrast to many previous works \cite{heusel2017gans, minmax_Jordan, adolphs2018local, wang2020ridge}, which assume that there exists a stationary point for their algorithm on the function $f$, and that their algorithm is initialized somewhere in the region of attraction for this stationary point.
 To the best of our knowledge,  greedy adversarial equilibrium is the first model for optimization in the presence of an adversary which is both guaranteed to exist and can be found efficiently in a general setting where $f$ is $C^2$-smooth, bounded, and has Lipschitz Hessian. 
 We expect it to find further use in learning in the presence of adversarial agents.
We note that we have not tried to optimize the order of the polynomial running time bound in Theorem \ref{thm:GreedyMin-max}. 
Finally, the greedy adversarial equilibrium our algorithm finds depends on the initial point $(x_0, y_0)$.  
 To  find other greedy adversarial equilibria, one can start from different initial points.

\section{Technical overview} \label{ses_alg_overview}
In this section we give an overview of our algorithm and of the proof of Theorem \ref{thm:GreedyMin-max}. In Section \ref{sec_algorithm_overview_full} we give a full description of the algorithm. 
In Section \ref{sec_proof_of_main_theorem} we give the full proof of Theorem \ref{thm:GreedyMin-max}.
 To simplify the exposition, we set $L = b = 1$,  $0<\epsilon< 1$, and $\sigma = \frac{1}{\sqrt{d}}$.
 In particular, if $f$ is $1$-bounded with $1$-Lipschitz Hessian, it is also $4$-Lipschitz with $2$-Lipschitz gradient.
 
 \subsection{An efficiently computable second-order local approximation to $\max_{z} f(x,z)$} \label{sec_local_approx}
 Ideally, we would like our algorithm to be able to compute the global maximum $\max_{z} f(x,z)$ at any point $x$.
However, since $f$ may be nonconvex-nonconcave, finding the global maximum may be intractable in our setting.
 
 Instead, starting from some initial point $z \leftarrow y$, we use a second-order maximization algorithm (Algorithm \ref{alg:InnerMaxLoop}) to find an $(\epsilon, \sqrt{\epsilon})$-approximate local maximum of $f(x,\cdot)$.
At each step, we would like our maximization algorithm to be able to rapidly decrease the value of $f(x,z)$ from any point $z$ that is not an $(\epsilon, \sqrt{\epsilon})$-approximate local maximum of $f(x,\cdot)$, that is, if $z$ is such that either $\|\nabla_y f(x,z)\| > \epsilon$ or $\lambda_{\mathrm{max}}(\nabla_y^2 f(x,z)) > \sqrt{\epsilon}$.
Towards this end, if $\|\nabla_y f(x,z)\| > \epsilon$, we have the max-player make an update $$z \leftarrow z+\mu_1 \nabla_y f(x,z)$$ for some step size $\mu_1>0$. %
  If we set
 \begin{equation}  \label{eq_mu1} 
 \mu_1 \leq \frac{1}{4},
 \end{equation}
  we have that, since $f$ has $2$-Lipschitz gradient, $$\nabla_y f(x,p)^\top \frac{\nabla_y f(x,z)}{\|\nabla_y f(x,z)\|} \geq \frac{1}{2} \|\nabla_y f(x,z)\|  > \frac{1}{2} \epsilon$$ at every point $p$ on the line segment $[z, z+\mu_1 \nabla_y f(x,z)]$.  
  This means that $f$ increases by at least $\frac{1}{2} \epsilon \times \mu_1  \|\nabla_y f(x,z)\|  > \frac{1}{2} \mu_1 \epsilon^2$  if the max-player makes this update.
On the other hand, if $\|\nabla_y f(x,z)\| \leq \epsilon$ but $\lambda_{\mathrm{max}}(\nabla_y^2 f(x,z)) > \sqrt{\epsilon}$, we have the max-player make an update $ z \leftarrow z+ \mu_3 \mathsf{a} v$ in the direction of the largest eigenvector $v$ of $\nabla_y^2 f(x,z)$ (with $\|v\|=1$), for some step size $\mu_3>0$, where the sign $\mathsf{a} \in \{-,1,1\}$ is chosen such that $(\mathsf{a} v)^\top \nabla_y f \geq 0$. %
If we set
 \begin{equation}  \label{eq_mu3}  
\mu_3 \leq \frac{1}{2} \sqrt{\epsilon}
 \end{equation}
we have that, since $f$ has $1$-Lipschitz Hessian, $v^\top \nabla_y^2 f(x,p) v > \frac{1}{2} \sqrt{\epsilon}$ for every point $p$ on the line segment $[z, z+ \mu_3 \mathsf{a} v]$.
 This means that $f$ increases by at least $\frac{1}{2} (\mu_3)^2 \sqrt{\epsilon}$ if the max-player makes this update.
Finally, if $\|\nabla_y f(x,z)\| \leq \epsilon$ and $\lambda_{\mathrm{max}}(\nabla_y^2 f(x,z)) \leq \sqrt{\epsilon}$, our maximization algorithm (Algorithm \ref{alg:InnerMaxLoop}) has reached an $(\epsilon, \sqrt{\epsilon})$-approximate local maximum $y'$, and it returns this point $y'$.

The above discussion implies that the value of $f$ increases by at least  $\Delta \coloneqq \min(\frac{1}{2} \mu_1 \epsilon^2, \frac{1}{2} (\mu_3)^2 \sqrt{\epsilon})$ at each iteration.
Since $f$ is also $1$-bounded, and each step of our method requires $O(1)$ oracle calls to the gradient and Hessian of $f$, and our method only stops once it reaches an  $(\epsilon, \sqrt{\epsilon})$-approximate local maximum, it uses at most $O(\frac{1}{\Delta})$ oracle calls to find an $(\epsilon, \sqrt{\epsilon})$-approximate local maximum $y'$ (Lemmas \ref{lemma:RunningTime}, \ref{lemma:Player2}).
 Thus, the value of $f$ at this  $(\epsilon, \sqrt{\epsilon})$-approximate local maximum $y'$, which we denote by the function $h_\epsilon$,
 \begin{equation} \label{eq_h_eps}  
 h_\epsilon(x,y) = f(x, y'),
 \end{equation}
 gives us a local approximation for $\max_{z} f(x,z)$, which can be computed in $O(\frac{1}{\Delta})$ oracle calls for the gradient, Hessian, and value of $f$. %
Moreover, as we explain in Remark \ref{rem_greedy_path}, the steps of our optimization method form an $\epsilon$-greedy path from $y$ to $y'$.%

\begin{remark}[{Our second-order maximization algorithm computes an $\epsilon$-greedy path}] \label{rem_greedy_path}
We have shown that at each point $p$ on a line segment $[z, z']$ connecting two consecutive steps $z$ and $z'$ of our maximization algorithm, we have that either  $\nabla_y f(x,p)^\top u > \frac{1}{2} \epsilon$  or $u^\top \nabla_y^2 f(x,p) u > \frac{1}{2} \sqrt{\epsilon}$, where $u$ is the unit vector $u = \frac{z'-z}{\|z'-z\|}$.
 Therefore, at any point on the unit-speed path $\varphi_t$ made up of the line segments connecting the consecutive steps of our algorithm, we have that either $\frac{\mathrm{d}}{\mathrm{d}t} f(x, \varphi_t) > \frac{1}{2} \epsilon$, or $\frac{\mathrm{d}^2}{\mathrm{d}t^2} f(x, \varphi_t)  > \frac{1}{2} \sqrt{\epsilon}$.
This implies that the path traced by our algorithm is a $\frac{1}{2} \epsilon$-greedy path. 

We show that, for a smaller choice of step sizes than required by \eqref{eq_mu1} and \eqref{eq_mu3}, namely for  $\mu_1, \mu_3 = \mathrm{poly}\left(\frac{1}{d}, \epsilon\right) $, this path is $\epsilon'$-greedy where $|\epsilon- \epsilon'| = \mathrm{poly}\left(\frac{1}{d}, \epsilon\right)$.
 The fact that $\epsilon' \neq \epsilon$ causes some technical issues which we deal with in the full algorithm\footnote{In the full algorithm we deal with this issue by initially computing an $\epsilon'$-greedy path for $\epsilon' = \frac{\epsilon}{2}$ at the first iteration, and then slowly increasing $\epsilon'$ at each iteration.} in Section \ref{sec:algorithm} and in our proof (Section 
 \ref{sec_proof_of_main_theorem}).  To improve readability, we ignore these issues in this technical overview and assume that the path is $\epsilon$-greedy.
 \end{remark}

\subsection{Finding a point $(x^\star, y^\star)$ which is a local min for $h_\epsilon(\cdot, y^\star)$ and a local max for $f(x^\star, \cdot)$} \label{sec_local_min_local_max}

Now that we have a subroutine for computing $h_\epsilon$, our next goal is to find an $(\epsilon, \sigma)$-approximate local minimum in the $x$ variable for $h_\epsilon$.
If we were able to compute the global maximum $\max_{z} f(x,z)$, it would be enough for our algorithm to find a global minimizer $x_{\mathrm{global}}$ for $\max_{z} f(x,z)$, and to then find a global maximizer $y_{\mathrm{global}}$ for $f(x_{\mathrm{global}}, \cdot)$.
Since $\max_{z} f(x,z)$ is only a function of $x$, the point $(x_{\mathrm{global}}, y_{\mathrm{global}})$ would still be a global min-max point regardless of which global maximizer $y_{\mathrm{global}}$ we find.
Here we encounter a difficulty:
\paragraph{Obstacle 1:}
Since the function $h_\epsilon(x,y)$ is a {\em local} approximation for $\max_z f(x,z)$, it depends both on $x$ and on the initial point $y$.
 If our algorithm were to first find an $(\epsilon, \sigma)$-approximate local minimum $x^\star$ for $h_\epsilon(\cdot, y)$, and then search for an  $(\epsilon, \sqrt{ \epsilon})$-approximate local maximum $y^\star$ for $f(x^\star, \cdot)$, the point  $x^\star$ may not be an $(\epsilon, \sigma)$-approximate local minimum for  $h_\epsilon(\cdot, y^\star)$ even though it is an $(\epsilon, \sigma)$-approximate local minimum for $h_\epsilon(\cdot, y)$.
To get around this problem, we use a different update rule for the min-player and max-player:
 \paragraph{Idea 1:} Alternate between a step where the min-player makes an update $w$ to $x$ which decreases the value of $h_\epsilon(\cdot, y)$ by some amount $\gamma_1$, and a step where the max-player uses the maximization subroutine discussed in Section \ref{sec_local_approx}  (Algorithm \ref{alg:InnerMaxLoop}), with initial point $y$, to find a $(\epsilon, \sqrt{ \epsilon})$-approximate local maximum $y'$ for $f(x, \cdot)$.

Since $h_\epsilon$ is the value of $f$ at the $(\epsilon, \sqrt{ \epsilon})$-approximate local maximum $y'$, we therefore have that $h_\epsilon(w, y) = f(w, y')$.
Moreover, since $y'$ is an $(\epsilon, \sqrt{ \epsilon})$-approximate local maximum for $f$, and Algorithm \ref{alg:InnerMaxLoop} stops whenever it reaches an $(\epsilon, \sqrt{ \epsilon})$-approximate local maximum, $y'$ is a stationary point for Algorithm \ref{alg:InnerMaxLoop}, which means that $f(w, y') = h_\epsilon(w, y')$.
Thus, 
\begin{equation} \label{h_unchanged} 
h_\epsilon(w, y) = f(w, y') = h_\epsilon(w, y').
\end{equation}
Moreover, since $f$ is $b$-bounded, and $h_\epsilon(x,y)$ is the value of $f(x,\cdot)$ at the approximate local maximum obtained by the maximization subroutine in Section \ref{sec_local_approx}, $h_\epsilon$ must also be $b$-bounded.
Thus, since the max-player's update does not change the value of $h_\epsilon$  \eqref{h_unchanged}, if we can show that whenever $x$ is not an $(\epsilon, \sigma)$-approximate local minimum of $f(\cdot, y)$ the min player can find an update to $x$ which decreases the value of $h_\epsilon(x,y)$ by at least some fixed amount $\gamma_1>0$, then we would have that the value of $h_\epsilon$ decreases monotonically at each iteration by at least $\gamma_1$. In that case, our algorithm would have to converge after  $\frac{2b}{\gamma_1}$ iterations to a point $(x,y')$ where $x$ is an $(\epsilon, \sigma)$-approximate local minimum for $h_\epsilon(\cdot, y')$ and $y'$ is an $(\epsilon, \sqrt{\epsilon})$-approximate local maximum for $f(x,\cdot)$.

\subsection{Escaping ``saddle points" of the discontinuous function $h_\epsilon$} \label{sec_escape_saddle}
 Before we can apply Idea 1, we would like to find a way for the min-player to find updates for $x$ which decrease the value of $h_\epsilon(x,y)$ by some amount at least $\gamma_1$  whenever the current value of $x$ is not an  $(\epsilon, \sigma)$-approximate local minimum (in the sense of Definition \ref{def_discontinuos_local_min}). 
  However, since $h_\epsilon$ is discontinuous we encounter a second obstacle:
 \paragraph{Obstacle 2:}
Finding an $(\epsilon, \sigma)$-approximate local minimum of $h_\epsilon$ requires our algorithm to escape saddle points of the truncated and smoothed function\footnote{We use a lowercase $s$ for the truncated and smoothed version of $h_\epsilon$ to distinguish it from the truncated and smoothed version of $g_\epsilon$ used in Definition \ref{def:local_min-max_formal}.} $s(w) = \mathbb{E}_{\zeta \sim N(0,I_d)}[\min(h_\epsilon(w + \sigma \zeta, y), h_\epsilon(x ,y))]$. 
 Ideally, we would like to run a ``noisy" version of gradient descent, which can allow us to escape saddle points (see e.g., \cite{ge2015escaping}), but we do not have access to the gradient of $s$.

 To get around this problem, we compute a stochastic gradient for $s(w)$ which can be computed without access to a gradient:
\paragraph{Idea 2:}
We use a stochastic gradient $\Gamma(w)$ for $s(w)$ which can be computed with access only to the value of $h_\epsilon$; roughly
\begin{equation} \label{eq_gradient_free_SG}  
\Gamma(w) = \frac{\zeta}{\sigma} \min\left(h_\epsilon(w + \sigma \zeta, y), h_\epsilon(x ,y)\right)
\end{equation}
where $\zeta \sim N(0, I_d)$, and $\mathbb{E}[ \Gamma(w)] = s(w)$ (see e.g. \cite{flaxman2005online}). 
This allows us to use a noisy version of stochastic gradient descent (SGD) to escape saddle points of $s(w)$.%

More specifically, starting at the initial point $w \leftarrow x$, each step of this ``noisy" SGD is given by
\begin{equation}\label{eq_noisySGD}
w \leftarrow w -\eta \Gamma(w) + \alpha \xi
\end{equation}
where $\xi \sim N(0,I_d)$ and $\eta, \alpha$ are hyperparameters .
We then apply concentration bounds for our stochastic gradient (Proposition %
 \ref{prop_SGVariance}) to results about noisy SGD \cite{NonconvexOptimizationForML} to show that, whenever $x$ is not an $(\epsilon, \sigma)$-approximate local minimum for $h_\epsilon(\cdot,y)$, with high probability, this noisy SGD, with hyperparameters $\eta, \alpha = \mathrm{poly}\left(\frac{1}{d}, \epsilon\right)$, can find an update for $x$ which decreases the value of $h_\epsilon$ by at least $\gamma_1=\mathrm{poly}\left(\frac{1}{d},  \epsilon\right)$ after $\mathcal{I} =  \mathrm{poly}\left(d, \frac{1}{\epsilon}\right)$ iterations of Equation \eqref{eq_noisySGD}    (Proposition \ref{Prop_NoisySGD}).

\subsection{Using $h_\epsilon$ to minimize the greedy max function} \label{sec_exact_local_min}
Although we have shown how to find an $(\epsilon, \sigma)$-approximate local minimum of $h_\epsilon$ (In the sense of Definition \ref{def_discontinuos_local_min}), our goal is to find an $(\epsilon, \sigma)$-approximate local minimum of the greedy max function $g_\epsilon$.   
For simplicity, let us start by supposing that we were able to find an {\em exact} local minimum for $h_\epsilon$. %
Since $g_\epsilon(x, y)$ is the maximum value of $f(x, z)$ that is attainable at the endpoint $z$ of {\em any} $\epsilon$-greedy path that starts at $y$, and, as we have shown in Remark \ref{rem_greedy_path}, the steps of the second-order optimization algorithm we use to compute $h_\epsilon(x,y)$ form one such greedy path whose endpoint $y'$ determines the value of $h_{\epsilon}$ (Equation \eqref{eq_h_eps}), we have that (Proposition \ref{Prop_lower_bound})
\begin{equation} \label{eq_Tech1}
h_{\epsilon} (x,y) \leq g_{\epsilon}(x,y) \qquad \forall (x,y) \in \mathbb{R}^d \times \mathbb{R}^d.
\end{equation}
However, it is still not clear how $h_{\epsilon}$ can help us find a local minimizer for $g_{\epsilon}$:

\paragraph{Obstacle 3:}
We want to minimize the greedy max function $g_\epsilon$, but computing the value of $g_\epsilon$ may be intractable, and we only have access to a lower bound $h_\epsilon \leq g_\epsilon$.
 We would like to somehow use our ability to compute $h_\epsilon$ to find a local minimum of $g_\epsilon$.
Towards this end, we observe that if any point $y^\star$ is an $(\epsilon, \sqrt{\epsilon})$-approximate local maximum, then 
 the conditions from the definition of approximate local maximum, $\|\nabla_y f(x,y^\star)\| \leq \epsilon$ and $\lambda_{\mathrm{min}}( \nabla^2 f(x, y^\star)) \leq \sqrt{\epsilon}$, imply that there is no unit speed path $\varphi: [0,\tau] \rightarrow \mathbb{R}^d$ starting at the point $\varphi_0 = y^\star$ for which $\frac{\mathrm{d}}{\mathrm{d}t} f(x, \varphi_0) >\epsilon$ or $\frac{\mathrm{d}^2}{\mathrm{d}t^2} f(x, \varphi_t) > \sqrt{\epsilon}$ on all $t \in (0, \tau]$.
This means that, if $y^\star$ is an $(\epsilon, \sqrt{\epsilon})$-approximate local maximum of $f(x, \cdot)$, then it is the {\em only} $(\epsilon, \sqrt{\epsilon})$-approximate local maximum reachable by any $\epsilon$-greedy path starting at $y^\star$.
In other words, we have that,
\begin{equation} \label{eq_Tech2}
h_\epsilon(x, y^\star) = g_\epsilon(x,y^\star)
\end{equation}
whenever $y^\star$ is an $(\epsilon, \sqrt{\epsilon})$-approximate local maximum for $f(x,\cdot)$ (Proposition \ref{Prop_fixed_point}).
Together, \eqref{eq_Tech1} and \eqref{eq_Tech2} imply the following:
\paragraph{Idea 3:}
 For any pair of points $(x^\star, y^\star)$ where $y^\star$ is an $(\epsilon, \sqrt{\epsilon})$-approximate local maximum for $f(x^\star, \cdot)$, we have that if $x^\star$ is an exact local minimum for $h_\epsilon(\cdot, y^\star)$ it must also be an exact local minimum for $g_\epsilon(\cdot, y^\star)$.

 This is because, if $x^\star$ is an exact local minimum, then there is an open ball $B$ containing $x^\star$ where $h_\epsilon(x^\star, y^\star) = \min_{w \in B} h_\epsilon(w, y^\star)$.
This implies that
\begin{equation}
g_\epsilon(x^\star, y^\star) \stackrel{\textrm{Eq.}\eqref{eq_Tech2}}{=} h_\epsilon(x^\star, y^\star) = \min_{w \in B} h_\epsilon(w, y^\star)
\stackrel{\textrm{Eq.}\eqref{eq_Tech1}}{\leq} \min_{w \in B} g_\epsilon(w, y^\star) 
\end{equation}
and hence that $x^\star$ minimizes $g_\epsilon(\cdot, y^\star)$ on the ball $B$.

Finally, we extend the result in Idea 3, which holds for exact local minima, to a similar result (Lemma \ref{Lemma_SharedLM}) that holds for {\em approximate} local minima.
 We show that if, roughly, the variance of our stochastic gradient $\Gamma$ (Equation \eqref{eq_gradient_free_SG}) satisfies a $\mathrm{poly}(\frac{1}{d}, \epsilon)$ upper bound (Proposition \ref{prop_SGVariance}), then for any pair of points $(x^\star, y^\star)$ where $y^\star$ is an $(\epsilon, \sqrt{\epsilon})$-approximate local maximum of $f(x^\star, \cdot)$ (in the sense of Equation \eqref{eq:ApproximateLocalMin_formal}), if $x^\star$ is an $(\epsilon, \sigma)$-approximate local minimum of $h_\epsilon(\cdot, y^\star)$ then it is also an $(\epsilon, \sigma)$-approximate local minimum for $g_\epsilon(\cdot, y^\star)$ (in the sense of Definition \ref{def_discontinuos_local_min}).
 The proofs of Lemma \ref{Lemma_SharedLM} and Proposition \ref{prop_SGVariance} are technical and summarized in Section \ref{sec_proof_of_main_theorem}.

\subsection{Showing convergence to a greedy adversarial equilibrium in $\mathrm{poly}\left(d, \frac{1}{\epsilon}\right)$ oracle calls} \label{sec_runtime}
From Idea 1 we have that our algorithm terminates after $O\left(\frac{1}{\gamma_1}\right)$ iterations consisting of an update for the min-player and max-player, if one can bound by some number $\gamma_1>0$ the amount by which each update for the min-player decreases the value of $h_\epsilon(x,y)$.

  From Idea 2 (Section \ref{sec_escape_saddle}) we have that, with high probability, noisy SGD can allow the min-player to find an update which decreases the value of $h_\epsilon$ by an amount  $\gamma_1=\mathrm{poly}\left(\frac{1}{d},  \epsilon\right)$, and that this can be accomplished in $\mathcal{I} =  \mathrm{poly}\left(d, \frac{1}{\epsilon}\right)$  computations of $h_\epsilon$.
 
 In Section \ref{sec_local_approx}, we show that our maximization subroutine can compute the value of $h_\epsilon$ in at most $O\left(\frac{1}{\Delta}\right)$ oracle calls, where $\Delta = \min\left\{\frac{1}{2} \mu_1 \epsilon^2, \frac{1}{2} (\mu_3)^2 \sqrt{\epsilon}\right\}$.
From Equations \eqref{eq_mu1} and \eqref{eq_mu3}, roughly speaking, we may set $\mu_1 = \frac{1}{4}$ and $\mu_3 =  \frac{1}{2} \sqrt{\epsilon}$. \footnote{As mentioned in Remark \ref{rem_greedy_path}, in the full algorithm we use somewhat smaller hyperparameter values $\mu_1, \mu_3 = \mathrm{poly}(\frac{1}{d}, \epsilon)$ to ensure that the path computed by Algorithm \ref{alg:InnerMaxLoop} is $\epsilon$-greedy instead of $\frac{1}{2} \epsilon$-greedy.}  %

This means that, with high probability, the number of oracle calls until our algorithm terminates is 
$
O\left(\frac{1}{\Delta} \times \mathcal{I} \times \frac{1}{\gamma_1}\right) = \mathrm{poly}\left(d, \frac{1}{\epsilon}\right).
$
Finally, from Idea 1 we also have that, if our algorithm terminates, it returns a pair of points $(x^\star, y^\star)$ where $x^\star$ is an $(\epsilon, \sigma)$-approximate local minimum for $h_\epsilon(\cdot, y^\star)$ (in the sense of Definition \ref{def_discontinuos_local_min}) and $y^\star$ is an $(\epsilon, \sqrt{ \epsilon})$-approximate local maximum for $f(x^\star,\cdot)$ (in the sense of Equation \eqref{eq:ApproximateLocalMin_formal}).

Applying Idea 3 (or rather its extension to {\em approximate} local minima), we have that  $x^\star$ is an $(\epsilon, \sigma)$-approximate local minimum for $h_\epsilon(\cdot, y^\star)$, which implies that $(x^\star, y^\star)$ is a $(\epsilon, \sigma)$-greedy adversarial equilibrium.
In other words, our algorithm returns an $(\epsilon, \sigma)$-greedy adversarial equilibrium after at most $ \mathrm{poly}\left(d, \frac{1}{\epsilon}\right)$ oracle calls for the gradient, Hessian, and value of $f$.

\subsection{Summary of algorithm}
The discussion in Sections \ref{sec_local_approx}-\ref{sec_runtime} leads us to the following algorithm (see Algorithm \ref{alg:LocalMin-max} for the full description). 
In addition to oracles for $f$, $\nabla_y f$ and $\nabla^2_y f$, our algorithm also takes as input an initial point $(x_0, y_0)$ in $\mathbb{R}^d \times \mathbb{R}^d$, and parameters $\epsilon, \sigma >0$  (recall we have set $\sigma = \frac{1}{\sqrt{d}}$ in this section). \footnote{In the full description of the algorithm we set  $(x_0, y_0) = (0, 0)$ for simplicity.}

\begin{enumerate}
\item  Starting at the initial point $(x_0, y_0)$, our algorithm first uses the second-order optimization method described in Section \ref{sec_local_approx} (Algorithm \ref{alg:InnerMaxLoop}) to find a point $y_1$ which is an  $(\epsilon, \sqrt{\epsilon})$-approximate local maximum for $f(x_0, \cdot)$.

\item Next, starting from iteration $i=1$, and setting $x_1 \leftarrow x_0$,  our algorithm uses noisy SGD (Equation \eqref{eq_noisySGD})\footnote{In the full algorithm we combine noisy SGD with a a random hill-climbing method (Lines \ref{RebootStart_hill}-\ref{RebootEnd_hill}  of Algorithm \ref{alg:LocalMin-max}).}
 to search for a point $x_{i+1}$ for which,
\be \label{eq_decrease}
 h_\epsilon(x_{i+1}, y_{i}) \leq h_\epsilon(x_i, y_{i})  - \gamma_1,
\ee
   where $\gamma_1 = \mathrm{poly}(\frac{1}{d},  \epsilon)$ (Lines \ref{RebootStart}-\ref{RebootEnd} of Algorithm \ref{alg:LocalMin-max}).
   When running noisy SGD, roughly speaking, our algorithm uses the stochastic gradient $\Gamma$ (the same stochastic gradient as in Equation \eqref{eq_gradient_free_SG}),
   \begin{equation}
 \Gamma(w) = \frac{\zeta}{\sigma} \min \left(h_\epsilon(w + \sigma \zeta, y_{i}), h_\epsilon(x_i,y_{i})\right),
\end{equation}
where $\zeta \sim N(0,I_d)$ and $h_\epsilon$ is computed using the second-order optimization method of Section \ref{sec_local_approx} (Algorithm \ref{alg:InnerMaxLoop}).

  \item
If our algorithm is able to find an update $x_{i+1}$ which satisfies Inequality \eqref{eq_decrease}, it uses Algorithm \ref{alg:InnerMaxLoop} to compute a point $y_{i+1}$, which is an $(\epsilon, \sqrt{\epsilon})$-approximate local maximum for $f(x_{i+1},\cdot)$,
 sets $i\leftarrow i+1$,  and goes back to Step 2.
  Otherwise, if it cannot find such an update, it concludes that $x_i$ is an $(\epsilon, \sigma)$-approximate local minimum for $h_\epsilon(\cdot ,y_{i}$), and, hence, that $(x_i,y_{i})$ is a  $(\epsilon, \sigma)$-greedy adversarial equilibrium.
\end{enumerate}

\section{Discussions and limitations} \label{Sec_Discussions_and_Limitations}

\subsection{How does our greedy adversarial equilibrium compare to previous models?} \label{sec_comparison_to_previous_notions}

In previous papers different models which can be seen as alternatives to min-max optimization have been considered in the nonconvex setting. 
 A number of papers \cite{daskalakis2018limit, heusel2017gans, adolphs2018local} consider the local min-max point model (sometimes called a ``local Nash" point or ``local saddle" point). %
Any point which is a local min-max point is also a greedy adversarial equilibrium for small enough $\sigma>0$ (Corollary \ref{corollary_exact_greedy_minmax}).

  To prove Corollary \ref{corollary_exact_greedy_minmax}, we use the following lemma, which states that any exact local minimum $x^\star$ of a possibly discontinuous function is also an approximate local minimum for the function $\psi$ for small enough $\epsilon, \sigma>0$ (in the sense of Definition \ref{def_discontinuos_local_min}).  We then use this Lemma to show that any local min-max point is also a greedy adversarial equilibrium for small enough $\epsilon, \sigma>0$ (Corollary \ref{corollary_exact_greedy_minmax}).

\begin{lemma} \label{lemma_local_nash}
Suppose that $x^\star$ is an exact local minimum for $\psi: \mathbb{R}^d \rightarrow \mathbb{R}$, and that there is a number $b>0$ such that $|\psi(x)| \leq b$ for all $x \in \mathbb{R}^d$.  Then for any $\epsilon>0$ there exists $\sigma^\star>0$ such that for any $0< \sigma \leq \sigma^\star$,  $x^\star$ is an approximate local minimum for the function $\psi$ with smoothing $\sigma$.
\end{lemma}

\noindent
We defer the proof of Lemma \ref{lemma_local_nash} to Appendix \ref{Sec_Auxilliary}.

  \begin{corollary} \label{corollary_exact_greedy_minmax}
  Suppose that for any $\delta>0$, $(x^\star, y^\star)$ is a $(0, \delta)$-local min-max point for a $C^2$-smooth function $f: \mathbb{R}^d \times \mathbb{R}^d \rightarrow \mathbb{R}$, and that there is a number $b>0$ such that $|f(x,y)| \leq b$ for all $x,y \in \mathbb{R}^d$.  Then for any $\epsilon>0$ there exists  $\sigma^\star>0$ such that for every $0< \sigma \leq  \sigma^\star$ we have that $(x^\star, y^\star)$ is a  $(\epsilon, \sigma)$-greedy adversarial equilibrium for $f$.
\end{corollary}

  \begin{proof}
  Any point $(x^\star, y^\star)$ which is a $(0, \delta)$-local min-max point has the property that $x^\star$ is an exact local minimum of $f(\cdot, y^\star)$ and that $y^\star$ is an exact local maximum of $f(x^\star, \cdot)$. 
    This implies that $y^\star$ is an approximate local maximum of $f(x^\star, \cdot)$ for parameter $\epsilon$ and any parameter $\theta$ (in the sense of Inequality \eqref{eq:ApproximateLocalMin_formal}).  Thus, the only greedy path starting at $y^\star$ consists only of the point $\{y^\star\}$ itself, and hence we have that $g_\epsilon(x^\star, y^\star) = f(x^\star, y^\star)$. 
    Therefore, since $g_\epsilon(x,y) \geq f(x,y)$ for all $x,y$, the fact that $x^\star$ is  an exact local minimum of  $f(\cdot, y^\star)$ implies that it is also an exact local minimum of $g_\epsilon(x^\star, y^\star)$. 
    By Lemma \ref{lemma_local_nash} we therefore have that there exists $\sigma^\star > 0$ such that $x^\star$ is an $(\eps, \sigma)$-approximate local minimum of $g_\eps(\cdot, y^\star)$ for any $0<\sigma \leq \sigma^\star$  (in the sense of Definition \ref{def_discontinuos_local_min}).

    Since $y^\star$ is an approximate local maximum of $f(x^\star, \cdot)$ for parameter $\epsilon$ and any parameter $\theta$  (in the sense of Inequality \eqref{eq:ApproximateLocalMin_formal}) and  $x^\star$ is an $(\epsilon, \sigma)$-approximate local minimum of $g_\eps(\cdot, y^\star)$ (in the sense of Definition \ref{def_discontinuos_local_min}), we must have that $(x^\star, y^\star)$ is an $(\epsilon, \sigma)$-greedy adversarial equilibrium.
  \end{proof}

\noindent
Since local min-max points are not guaranteed to exist in general, previous algorithms which seek local min-max points oftentimes make strong assumptions on the function $f$.
For instance, \cite{abernethy2019last} show that if $f$ has $1$-Lipschitz gradient and satisfies a ``sufficient bilinearity'' condition-- that is, roughly speaking, if the cross derivative $\nabla^2_{xy} f(x,y)$ has all its singular values greater than some $\gamma>2$ at every $x,y \in \mathbb{R}^d$--then their algorithm reaches a point $(x^\star, y^\star)$ where $\|\nabla_x f(x^\star, y^\star)\| \leq \epsilon$ and  $\|\nabla_y f(x^\star, y^\star)\| \leq \epsilon$ 
 in $O\left(\frac{1}{\gamma^2} \log\frac{M}{\epsilon}\right)$ evaluations of a Hessian-vector product of $f$, where $M$ is the magnitude of  $\nabla f$ at the point where their algorithm is initialized.\footnote{
Note that, since such algorithms' dynamics should ideally not be attracted to minima or maxima of $f$, additional assumptions on $f$ are typically needed even to show convergence to a first-order stationary point--that is, a point $(x^\star, y^\star)$ where $\|\nabla_x f(x^\star, y^\star)\| \leq \epsilon$ and  $\|\nabla_y f(x^\star, y^\star)\| \leq \epsilon$.}

 In \cite{minmax_Jordan} the authors consider an alternative model to min-max optimization which incorporates the fact that in min-max optimization the min-player reveals her strategy before the max-player.  In their notion, both players are restricted to making updates in vanishingly small neighborhoods of the optimum point (although the size of the neighborhood for the min-player is allowed to vanish at a much faster rate than the neighborhood for the max-player). 
 One difference between our greedy adversarial equilibrium model and the model considered in \cite{minmax_Jordan} is that in \cite{minmax_Jordan} the max-player is able to compute a  global maximum (albeit when restricted to a ball of vanishingly small radius), while in our greedy adversarial equilibrium model the max-player is constrained to points reachable by a greedy path of any length.  That being said, our main result still holds if we restrict the greedy path of the max-player to be proportional to the updates made by the minimizing player (see Remark \ref{remark_neighborhood}). Another difference is that, while the solution point for the model in \cite{minmax_Jordan} is not guaranteed to exist in a general nonconvex-nonconcave setting (see Remark \ref{example_nonexistance}), our main result (Theorem \ref{thm:GreedyMin-max}) guarantees that any uniformly bounded function with Lipschitz Hessian has a greedy adversarial equilibrium.

  \begin{remark} \label{remark_neighborhood}
Theorem \ref{thm:GreedyMin-max} still holds if we restrict the greedy path to a ball whose radius is proportional to $\epsilon$, $L_1$ and the distance $\|x-x^\star \|$ between $x^\star$ and the minimizing player's update $x$. 
This is because, roughly speaking, any greedy path that leaves this ball would reach a point $y$ for which the value of $f$ at $(x,y)$ is greater than the value of $f$ at $(x^\star, y^\star)$. 
This implies that the truncated greedy max function $\min(g_\epsilon(x,y), g_\epsilon(x^\star,y^\star))$ would have the exact same value regardless of whether we restrict the max-player to such a ball, and the point $(x^\star, y^\star)$ guaranteed by Theorem \ref{thm:GreedyMin-max} would therefore still satisfy Definition \ref{def:local_min-max_formal}.
\end{remark}

\begin{remark} \label{example_nonexistance}
 The authors of \cite{minmax_Jordan} note that their model does not have any solution point on the function $f(x,y) = y^2 - 2xy$ on $[-1,1] \times [-1,1]$. 
  We can extend this function to a continuous function on all of $\mathbb{R}^2$ as follows. 
  Defining $\tilde{z} := | ((z+3) \mod 4) -2|-1$ for $z \in \mathbb{R}^2$, we have $f(x,y) = \tilde{y}^2 -2\tilde{x}\tilde{y}$  (in other words, we imagine that we put mirrors at each of the edges of the square $[-1,1] \times [-1,1]$, which reflect the value of the function across all of $\mathbb{R}^2$). 
    This function is uniformly bounded on all of $\mathbb{R}^2$ but does not have any point which is a solution point for the model in \cite{minmax_Jordan}.  
  We also note that one can smooth this function by convolution with a small-radius bump function of radius $r<\frac{1}{100}$, and these properties still hold for the smoothed function $\tilde{f}$.
  We can extend this example to many dimensions, for example by considering the objective function $\hat{f}(x,y) = \sum_{i=1}^d \tilde{f}(x_i,y_i)$.

\end{remark}

\subsection{Applicability and limitations of our definition}

 The class of algorithms that our definition allows the players to use includes a range of algorithms, e.g.,  gradient descent and negative curvature descent \cite{liu2018adaptive, reddi2018generic}, which only take steps in directions where the gradient or second derivative is above some threshold value.
  
Moreover, one can expand our definition to allow the max-player to also use randomized algorithms such as noisy gradient descent \cite{NonconvexOptimizationForML}, as long as the algorithm stops once an approximate local maximum is reached. 
For this class of algorithms, any point $(x^\star, y^\star)$ which satisfies our original Definition \ref{def:local_min-max_formal}  also is a greedy adversarial equilibrium under this expanded definition.  
Roughly speaking, this is because as long as the max-player is at a local maximum for the function $f(x^\star,\cdot)$, expanding the choice of algorithms available to the max-player may increase the value of the greedy max function at points other than $x^\star$ but will not increase the value of the greedy max function at the current point $x^\star$. 
In other words, the minimizing player will not have an incentive to deviate from $x^\star$ if more algorithms are made available to the max-player.

On the other hand, if we allow the max-player to use algorithms which do not stop at local maxima, for instance algorithms such as simulated annealing, a solution  $(x^\star, y^\star)$ which satisfies our current definition may no longer be a solution in this expanded sense.  
This is because, giving the max-player the option to use algorithms which do not stop once a local maximum is reached may cause the greedy max function to increase at $x^\star$ more than at neighboring points, incentivizing the minimizing player to deviate from $x^\star$.

  \subsection{The necessity of dealing with discontinuities in the greedy max function} \label{sec_discontinuity}

At first glance, it may seem that we can simply restrict ourselves to considering functions $f(x,y)$ for which the greedy max function $g_{\epsilon} (x,y)$ is continuous. 
 This would greatly simplify our proof, since we could exclude ``unstable" situations where the min-player proposes a small change in $x$ which would then cause the max-player to respond by making a large change in her strategy. 
  A second difficulty involving discontinuous greedy max functions is that, since we allow our algorithm to start at any point, even greedy max functions with discontinuities far from the greedy adversarial equilibrium point(s) are challenging to analyze. 
   Unfortunately, even very simple functions $f(x,y)$ oftentimes have discontinuous greedy max functions $g_{\epsilon} (x,y)$ (see Example \ref{discontinuity_example}). 
    Excluding functions where such discontinuities arise would greatly restrict the applicability of our results, and a large part of our proof is devoted to dealing with the possibility of discontinuities in the greedy max function.

\subsection{Additional discussion of approximate local minimum for discontinuous functions} \label{sec_discontinuous_localmin}

When choosing our definition for approximate local minimum of discontinuous functions, we would like this definition to be as close as possible to the notion of approximate local minimum for $C^2$-smooth functions (Inequality \eqref{eq:ApproximateLocalMin_formal}).
 This  allows us to more easily relate our results to past work in the optimization literature. 
  For instance, in our proof, we would like to adapt results from \cite{NonconvexOptimizationForML} about escaping saddle points in polynomial time to the setting of discontinuous functions.  
  However, we cannot expect our algorithm to have direct access to the discontinuous function $g_\epsilon(\cdot, y)$ we wish to minimize. 
   To allow us to handle this more difficult setting, we would like our notion of approximate local minimum to satisfy the property that any point which is an exact local minimum is also an approximate local minimum under our definition.

To obtain a definition which applies to discontinuous functions yet is as close as possible to the definition in \eqref{eq:ApproximateLocalMin_formal}, we would like to approximate any discontinuous function $\psi$ with a $C^2$-smooth function.%
 When choosing which $C^2$-smooth approximation to use, we would like it to satisfy the following three properties.
\begin{enumerate}
\item {\bf $C^2$-smooth with Lipschitz Hessian.} We would like each function in our family of approximation functions to be $C^2$-smooth with Lipschitz Hessian.  This would allow us to apply the definition of approximate local minimum for $C^2$-smooth functions (Inequality \eqref{eq:ApproximateLocalMin_formal}) to any function in this family.
\item {\bf Shared local minima.} We would like our family of $C^2$-smooth approximation functions to have the property that for any $x^\star$ which is an (exact) local minimum of the objective function $\psi$, and any $\epsilon>0$, there is a function in this family such that $x^\star$ is also an $(\epsilon, \sqrt{\eps})$-approximate local minimum of this $C^2$-smooth function (in the sense of Inequality \eqref{eq:ApproximateLocalMin_formal}).
\item {\bf Easy to compute.} We would like each function in our family of approximation functions to be easily computed within some error $\epsilon$ at any point $x$ in $\mathrm{poly}(d,\nicefrac{1}{\epsilon}, b)$  evaluations of $\psi$.
\end{enumerate}
 
 \noindent
Towards this end, we consider the family of functions $\mathcal{F}$ where we convolve $\psi$ with a Gaussian density $N(0, \sigma^2 I_d)$ of some variance $\sigma^2$ and zero mean. 
 That is, we consider functions of the form $$\psi_\sigma(x) := \mathbb{E}_{\zeta \sim N(0,I_d)}\left [\psi(x  + \sigma \zeta) \right]$$ for some $\sigma>0$. 
  This family of functions is $C^2$-smooth and has Lipschitz Hessian, which satisfies our first property (1). 
  This is because a Gaussian density is $C^2$-smooth, and any function convolved with a $C^2$-smooth function is also $C^2$-smooth. 
    Moreover, if $\psi$ is $b$-bounded, then convolving $\psi$ with a Gaussian  gives a $b$-bounded function with the magnitude of its $k$'th-derivatives bounded by $2b$ times an upper bound on the $k$'th derivative of the standard Gaussian density, that is, $2b \times \frac{1}{\sigma^{2k+1} \sqrt{2 \pi}}$ for every $k>0$. 
     In particular, this means our smoothed function $\psi_\sigma(x)$ is also $b$-bounded, with $b \times \frac{1}{\sigma^{7} \sqrt{2 \pi}}$-Lipschitz Hessian (see Remark \ref{Rem_Lipschitz_convolution}).

The family of functions $\mathcal{F}$ also has the advantage that, if $\psi$ is $b$-bounded, it can be computed within error $\epsilon$ in $\mathrm{poly}(d,\nicefrac{1}{\epsilon}, b)$ evaluations of $\psi$ with high probability if one uses a Monte-Carlo computation of the expectation $\mathbb{E}_{\zeta \sim N(0,I_d)}\left [\psi(x  + \sigma \zeta)\right]$, which satisfies our third property (3).

To satisfy our second property (2), we would ideally like to ensure that, for every exact local minimum $x^\star$ of $\psi$, and every $\epsilon>0$, there is a small enough $\sigma>0$  such that $x^\star$ is an  $(\epsilon, \sqrt{\eps})$-approximate local minimum (in the sense of Inequality \eqref{eq:ApproximateLocalMin_formal}) of the smoothed function $\psi_{\sigma} = \mathbb{E}_{\zeta \sim N(0,I_d)}\left [\psi(x  + \sigma \zeta) \right]$. 
 Unfortunately, smoothing $\psi$ by convolution alone does not directly allow us to satisfy property (2). The following example illustrates this problem.
\begin{example}[Convolution can shift local minima] \label{ex_shifted_min1}
Consider the function $\psi: \mathbb{R} \rightarrow \mathbb{R}$, where
 $   \psi(x)    = x - 3x \mathbbm{1}(x\leq 0) + \mathbbm{1}(x \leq 0).$
   This function is discontinuous at $x=0$, and has an exact local minimum at the point $x=0$ (which also happens to be its global minimum point). 
    If we smooth $\psi$ by convolving it with a Gaussian distribution $N(0,\sigma^2)$ for any $\sigma>0$,  we get the smooth function 
    \begin{equation*}
        \psi_\sigma(x) = 3 \sigma \frac{1}{\sqrt{2\pi}} e^{-\frac{x^2}{2 \sigma^2}} - x + x \Phi\left(\frac{1}{\sigma} x\right) + \Phi\left(-\frac{1}{\sigma} x\right),
         \end{equation*}
        where $\Phi(\cdot)$ is the standard Gaussian cumulative distribution function. 
    This function is $C^{2}$-smooth since $\Phi(\cdot)$ is $C^{2}$-smooth. 
    However, for any $\sigma >0$, the gradient at $x=0$ of the smoothed function is $-1.5 - \frac{1}{\sigma \sqrt{2 \pi}}$. 
     Thus, for any $\sigma>0$, $x=0$ is not an approximate local minimum of the smoothed function for any parameter $\epsilon \leq 1.5$.
\end{example}

  \noindent
  In Example \ref{ex_shifted_min1} the gradient of the smoothed function $\psi_\sigma$ at $x^\star=0$ has magnitude at least $1.5$ for any $\sigma>0$ even though $x^\star=0$ is a local minimum of $\psi$. 
   To understand how this is possible, we consider the following stochastic gradient (see e.g. \cite{flaxman2005online}) for the smoothed function $\psi_\sigma$:
\begin{equation} \label{eq_stochastic_gradient_formula}
\nabla \psi_\sigma(x) =   \frac{1}{\sigma} \mathbb{E}_{\zeta \sim N(0,I_d)}[(\psi(x + \sigma \zeta) - \psi(x)) \zeta].
\end{equation}
 One can obtain a non-zero gradient $\nabla \psi_\sigma(x)$ even if all of the sampled points $x + \sigma \zeta$ in \eqref{eq_stochastic_gradient_formula} give values $\psi(x + \sigma \zeta)$ greater than $\psi(x)$.

 If $\psi$ were smooth, finding a small step $\sigma \zeta$ which increases the value of $\psi$ (by at least some amount proportional to the step size) would imply that $\psi$ decreases in the direction $-\sigma \zeta$.
For smooth objective functions one can therefore find a descent direction (a direction in which $\psi$ decreases) simply by first finding an ascent direction $\sigma \zeta$ and then moving in the opposite direction $-\sigma \zeta$.
 Unfortunately, this is not true for discontinuous functions, since if $\psi$ is discontinuous, it may be that $\psi(x^\star + \sigma \zeta) > \psi(x^\star)$ does not imply that $\psi(x^\star - \sigma \zeta) < \psi(x^\star)$ no matter how small a step $\sigma \zeta$ we take. 
  In other words, for discontinuous objective functions the presence of an ``ascent direction" along which the objective function increases does not imply the existence of a ``descent direction" along which the objective function decreases.
  The only thing that matters when determining whether a discontinuous function has a local minimum at some point $x^\star$ is whether, in every ball containing $x^\star$, there are points $x^\star + \sigma \zeta$ for which  $\psi(x^\star + \sigma \zeta) < \psi(x^\star)$.
  
  To enable our definition of approximate local minimum to only consider those directions which decrease the value of $\psi$, when determining whether a point $x^\star$ is an (approximate) local minimum we instead consider the truncated function $\min(\psi(x),  \psi(x^\star))$.
 We then smooth this truncated function by convolving it with a Gaussian, to obtain the following smoothed function of $x$:
  \be
\mathbb{E}_{\zeta \sim N(0,I_d)}\left [\min(\psi(x  + \sigma \zeta),  \psi(x^\star)) \right].
  \ee
  This function has the property that it is both $C^2$-smooth and has $\frac{b}{\sigma^{7}}$-Lipschitz Hessian, since it is the convolution of a $b$-bounded function $\psi$ with a Gaussian of variance $\sigma^2$. 
  This leads us to Definition \ref{def_discontinuos_local_min}, which says that $x^\star$ is an approximate local minimum ``with smoothing $\sigma$" for a discontinuous function $\psi$, if $x^\star$ is an approximate local minimum of the smooth function $\mathbb{E}_{\zeta \sim N(0,I_d)}\left [\min(\psi(x  + \sigma \zeta),  \psi(x^\star)) \right]$.

\begin{remark}[Lipschitz and smoothness properties of convolution] \label{Rem_Lipschitz_convolution}
If $\psi$ is a function and $\rho$ a probability distribution, then the convolution $\psi \ast \rho (x)$ of $\psi$ with $\rho$ is  defined as $\psi \ast \rho (x) := \mathbb{E}_{\zeta \sim \rho}\left [\psi(x  + \zeta)\right]$.  In one dimension, we can write this as the integral $\psi \ast \rho (x) = \int_{-\infty}^\infty \psi(x-t) \rho(t) \mathrm{d}t =  \int_{-\infty}^\infty \psi(t) \rho(x-t) \mathrm{d}t$. 
 Hence, if $\psi$ and $\rho$ are uniformly bounded then $\frac{\mathrm{d}^k}{\mathrm{d} x^k} \psi \ast \rho (x) =  \int_{-\infty}^\infty \psi(t) \frac{\mathrm{d}^k}{\mathrm{d} x^k} \rho(x-t) \mathrm{d}t$.
  Thus, if the distribution $\rho$ is $C^k$-smooth, then the convolution $\psi \ast \rho$ is also $C^k$-smooth. 
   Moreover, if $|\psi(x)| \leq b$  for all $x \in \mathbb{R}$, and for any $k>0$ there is a number $c_k$ such that $\frac{\mathrm{d}^k}{\mathrm{d} x^k} \rho(x)| \leq c_k$ for all $x \in \mathbb{R}$, then we must have that $|\frac{\mathrm{d}^k}{\mathrm{d} x^k} \psi \ast \rho (x)| \leq b \times c_k$ for all $x \in \mathbb{R}$. 
    And, since $\psi$ is a probability distribution, we also have that $|\psi \ast \rho (x)| \leq b$ for all $x \in \mathbb{R}$ whenever $|\psi(x)| \leq b$  for all $x \in \mathbb{R}$. 
     In particular, if $\rho$ is the Gaussian distribution with variance $\sigma^2$, then its $k$'th derivative is bounded by $\frac{1}{\sigma^{2k+1} \sqrt{2 \pi}}$ for every $k$. 
     This implies that $\psi \ast \rho$ is $b$-bounded,  $2b \times \frac{1}{\sigma^{3} \sqrt{2 \pi}}$-Lipschitz, with $2b \times \frac{1}{\sigma^{5} \sqrt{2 \pi}}$-Lipschitz gradient and $2b \times \frac{1}{\sigma^{7} \sqrt{2 \pi}}$-Lipschitz Hessian. 
      The same argument can be extended to functions $\psi: \mathbb{R}^d \rightarrow \mathbb{R}$ of dimensions $d>1$.
\end{remark}

\section{The full algorithm} \label{sec_algorithm_overview_full} \label{sec:algorithm}
In this section we present the full algorithm for computing a greedy adversarial equilibrium (Algorithm \ref{alg:LocalMin-max}), as well as an algorithm for computing a greedy path (Algorithm \ref{alg:InnerMaxLoop}) which Algorithm \ref{alg:LocalMin-max} uses as a subroutine.

\begin{algorithm}
\caption{Computing a greedy path \label{alg:InnerMaxLoop}}
\KwIn{Oracles for the value of a function $f: \mathbb{R}^d \times \mathbb{R}^d \rightarrow \mathbb{R}$, the gradient $\nabla_yf$ and the Hessian $\nabla_y^2f$}
\KwIn{$\mathsf{x}, \mathsf{y}^0, \epsilon'$} 
\KwHyperparameters{$\delta, \mu_1, \mu_3, \mu_4$}
  
Set $\aell \leftarrow 0$,  $\mathsf{Stopy}  \leftarrow \textrm{False}$

\While{$\mathsf{Stopy} = \mathsf{False}$}{ \label{InnerWhileStart}

\uIf{$\| \nabla_yf(\mathsf{x}, \mathsf{y}^{\aell})\| > \epsilon'$}{ \label{NoislessSGD_Start}

 Set $\mathsf{y}^{\aell+1} \leftarrow \mathsf{y}^{\aell} + \mu_1 \nabla_yf(\mathsf{x},\mathsf{y}^{\aell})$ \label{SGD_NoNoise} 
 
   Set $\aell \leftarrow \aell+1$} \label{NoislessSGD_End}

\Else{
Compute an eigenvalue-eigenvector pair $(\lambda, v)$ of $\nabla^2_{{y}}f(\mathsf{x}, \mathsf{y}^{\aell})$, s.t. $\lambda \geq \lambda_{\mathrm{max}}(\nabla^2_{{y}}f(\mathsf{x}, \mathsf{y}^{\aell})) - \mu_4$ \label{Compute_eigenvalue} 

\uIf{$\lambda > \sqrt{L_{} \epsilon'}$}{  \label{GreedyEnpoint_Start}

Set $\mathsf{a} = \mathrm{sign}(\nabla_yf(\mathsf{x},\mathsf{y}^{\aell})^\top v)$

Set $\mathsf{y}^{\aell+1} \leftarrow \mathsf{y}^\ell + \mu_3  \mathsf{a} v$  \label{saddle_escape}

 Set $\aell \leftarrow \aell+1$

\Else{Set $\mathsf{Stopy} = \mathsf{True}$}}
\label{GreedyEnpoint_End}
}
\label{InnerWhileEnd}}

\Return{$\mathsf{y}_{\mathrm{LocalMax}} \leftarrow \mathsf{y}^{\aell}$}

\end{algorithm}

  \begin{algorithm}
\caption{Computing a greedy adversarial equilibrium \label{alg:LocalMin-max}}
\KwIn{Oracle for a function $f: \mathbb{R}^d \times \mathbb{R}^d \rightarrow \mathbb{R}$, and oracles for the gradient $\nabla_yf$ and Hessian $\nabla_y^2f$. $\sigma,\epsilon >0$}
\KwHyperparameters{\, \, $\eta, \gamma_1, \mathcal{I}_2, \delta, \mathcal{I}_3, \mathcal{I}_4, \alpha,  \epsilon_0$}

 Initialize $(x_0, y_0) \leftarrow (0, 0)$
 
 Set $x_1 \leftarrow x_0$. 
 
Run Algorithm \ref{alg:InnerMaxLoop} with inputs $\mathsf{x} \leftarrow x_1$,  $\mathsf{y}^0  \leftarrow y_0$,  $\epsilon' \leftarrow \epsilon_{0}(1+\delta)$. \label{InitializationStart}

Set $y_1 \leftarrow \mathsf{y}_{\mathrm{LocalMax}}$ to be the output $\mathsf{y}_{\mathrm{LocalMax}}$ of Algorithm \ref{alg:InnerMaxLoop}.  \label{InitializationEnd}   
                    
Set $h^0 \leftarrow f(x_1, y_1)$

  Set $\mathsf{Stop}\leftarrow \mathsf{False}$, \, $\ai \leftarrow 0$\\
\While{$\mathsf{Stop}=\mathsf{False}$}{ \label{OuterWhileStart}
Set $\ai \leftarrow \ai+1$, \, \, \,
   $\mathsf{NoProgress} \leftarrow \mathsf{True}$, \, \, \,
 
 Set $\epsilon_{\ai} \leftarrow \epsilon_{\ai-1}(1+\delta)^2$
 
 Set $X_0 \leftarrow x_{\ai}$

\For{$\aj=1$ to $\mathcal{I}_3$\label{RebootStart_hill}}{
\If{$\mathsf{NoProgress} = \mathsf{True}$}{
Set $\zeta_{\ai \aj} \sim N(0,I_d)$

Run Algorithm \ref{alg:InnerMaxLoop} with inputs $\mathsf{x} \leftarrow x_{\ai} + \sigma \zeta_{\ai \aj}$,  $\mathsf{y}^0 \leftarrow y_{\ai}$, and $\epsilon' \leftarrow \epsilon_{\ai}(1+\delta)$.

Set $\mathcal{Y} \leftarrow \mathsf{y}_{\mathrm{LocalMax}}$ to be the output $\mathsf{y}_{\mathrm{LocalMax}}$ of Algorithm \ref{alg:InnerMaxLoop}.     

\If{$f(x_{\ai} + \sigma \zeta_{\ai \aj}, \mathcal{Y}) \leq f(x_{\ai}, y_{\ai}) - \gamma_1$}{ \label{HillClimbing}
Set $x_{\ai+1} \leftarrow x_{\ai} + \sigma \zeta_{\ai \aj}$

Set $y_{\ai+1} \leftarrow \mathcal{Y}$,  

Set $h^0 \leftarrow f(x_{\ai}, y_{\ai})$

Set  $\ai \leftarrow \ai+1$, \, \, \, $\mathsf{NoProgress} \leftarrow \mathsf{False}$,

}
}
}  \label{RebootEnd_hill}
\For{$\aj=1$ to $\mathcal{I}_4$ \label{RebootStart}}{

\If{$\mathsf{NoProgress} = \mathsf{True}$}{

\For{$\ak = 1$ to $\mathcal{I}_2$}{ \label{k_for_Start}

Set $u \sim N(0,I_d)$

Run Algorithm \ref{alg:InnerMaxLoop} with inputs $\mathsf{x} \leftarrow X_{\ak-1} + \sigma u$,  $\mathsf{y}^0 \leftarrow y_{\ai}$,  $\epsilon' \leftarrow \epsilon_{\ai}(1+\delta)$.

Set $\mathcal{Y} \leftarrow \mathsf{y}_{\mathrm{LocalMax}}$ to be the output $\mathsf{y}_{\mathrm{LocalMax}}$ of Algorithm \ref{alg:InnerMaxLoop}.     

Set $h^{\ak} = \min(f(X_{\ak-1} + \sigma u, \mathcal{Y}), f(x_{\ai}, y_{\ai}))$ \label{PSGD_start}
 
Set $\Gamma_{\ak} = (h^{\ak} - h^{\ak-1}) \frac{1}{\sigma} u$ \label{PSGD_end}

Set $\xi \sim N(0,I_d)$, \, \, \,

Set $X_{\ak} \leftarrow X_{\ak-1} - \eta \Gamma_{\ak} + \alpha \xi$   
}  \label{k_for_End}

Run Algorithm \ref{alg:InnerMaxLoop} with inputs $\mathsf{x} \leftarrow X_{\ak}$, $\mathsf{y}^0 \leftarrow y_{\ai}$, and $\epsilon' \leftarrow \epsilon_{\ai}(1+\delta)$.

Set $\mathcal{Y} \leftarrow \mathsf{y}_{\mathrm{LocalMax}}$ to be the output $\mathsf{y}_{\mathrm{LocalMax}}$ of Algorithm \ref{alg:InnerMaxLoop}.  

\If{$f(X_{\ak}, \mathcal{Y}) \leq f(X_0, y_{\ai}) - \gamma_1 $}{ \label{SuccessStart}

Set $x_{\ai+1} \leftarrow X_{\ak}$

Set $y_{\ai+1} \leftarrow \mathcal{Y}$,  

Set $h^0 \leftarrow f(x_{\ai}, y_{\ai})$, \, \, \,

Set  $\ai \leftarrow \ai+1$, \,  \, \,
$\mathsf{NoProgress} \leftarrow \mathsf{False}$
} \label{SuccessEnd}
}
} \label{RebootEnd}

\If{$\mathsf{NoProgress} = \mathsf{True}$}{
Set $\mathsf{Stop} = \mathsf{True}$ 

}

}
\Return{$\ai^\star \leftarrow \ai$,  $\epsilon^\star \leftarrow \epsilon_{\ai^\star}$, and $(x^\star, y^\star) \leftarrow (x_{\ai^\star}, y_{\ai^\star})$} \label{Output}
\end{algorithm}

\vspace{-3mm}

\section{Proof of Theorem \ref{thm:GreedyMin-max}} \label{sec_proof_of_main_theorem}

\subsection{Setting constants and notation} \label{section:constants}
Since $f$ is a $b$-bounded $C^2$-smooth function with $L$-Lipschitz Hessian, it is also $L_1$-Lipschitz with  $L_1 \leq 4 b^{\nicefrac{2}{3}} L^{\nicefrac{1}{3}}$ and has $L_2$-Lipschitz gradient with $L_2 \leq 2 b^{\nicefrac{1}{3}} L^{\nicefrac{2}{3}}$.
From now on, we set $L_1 = 4 b^{\nicefrac{2}{3}} L^{\nicefrac{1}{3}}$ and $L_2 = 2 b^{\nicefrac{1}{3}} L^{\nicefrac{2}{3}}$.
Without loss of generality, we may assume that $b \geq 1$ (since our goal is to prove that the number of gradient evaluations is polynomial in $\frac{1}{\epsilon}, d, b, L_{}, \frac{1}{\sigma}$).
In our proof, we set the following hyperparameters and constants.
\begin{enumerate}
    \item $\omega \coloneqq 10^{-3}$, 
\item $\gamma_1 \coloneqq \frac{\epsilon^{2.1} \sigma^{16.6}}{10^4 (1+b^{3.1}) d^{0.6} \log(b d \sigma \epsilon)}$, 
\item $\delta \coloneqq \frac{\gamma_1^2}{8 b^2}$, 
\item  $\mu_1 \coloneqq \delta \frac{1}{L_2(L_1+1)}$,
\item $\mu_3 \coloneqq \frac{1}{7}  \min\left( \frac{\delta \sqrt{\epsilon}}{\sqrt{L_{}}}, \frac{\epsilon}{\sqrt{L_{}}} \right)$, 
 \item $\mu_4 \coloneqq \frac{1}{7} \sqrt{\delta L_{} \epsilon}$, 
\item $\eta \coloneqq \frac{\sigma^9 }{  b^6 d^2 \left(1+ 10\frac{b d}{\sigma^{12}\epsilon^2}\right) \mathsf{c} \log^9( bd\sqrt{\sigma \epsilon})}$,  
\item $\mathcal{I}_2 \coloneqq \frac{\mathsf{c} \log(bd\sqrt{\sigma \epsilon})}{\eta \sqrt{\epsilon}}$, 
\item $\mathcal{I}_3 \coloneqq \frac{30 b}{\gamma_1}$,  
\item $\mathcal{I}_4 \coloneqq 6\log\left(\frac{2 b}{\gamma_1 \omega}\right),$ 
 \item$\alpha \coloneqq \eta \mathsf{c}\log( bd\sqrt{\sigma \epsilon}) \sqrt{1+ b^2 d^2 \sigma^{-2}},$
 \end{enumerate}
 where $\mathsf{c}$ is a large enough universal constant.

In particular, we have set $\delta = \frac{1}{4 i_{\mathrm{max}}^2}$, where $i_{\max} \coloneqq \frac{2b}{\gamma_1}$ is an upper bound on the number of iterations of the While loop in Algorithm \ref{alg:LocalMin-max}.  
This ensures that $(1+\delta)^i \leq 2$ for all $i \in [i_{\mathrm{max}}]$.

In the following sections we let $(x_i, y_i)$ denote the points  $(x_i, y_i)$  at each iteration $i$ of the While loop in Algorithm \ref{alg:LocalMin-max}, and we set $\epsilon_i \coloneqq \epsilon_0 (1+\delta)^{2i}$ for all $i \in \mathbb{N}$.

For any $\epsilon^{\circ}>0$, and any $(x,y) \in \mathbb{R}^d \times \mathbb{R}^d$, we define 
\be \label{eq:z1}
h_{\epsilon^{\circ}} (x,y) \coloneqq f(x,\mathcal{Y}),
\ee
 where $\mathcal{Y} \leftarrow \mathsf{y}_{\mathrm{LocalMax}}$ is the output of Algorithm \ref{alg:InnerMaxLoop} with inputs $\mathsf{x} \leftarrow x$, $\mathsf{y}^0 \leftarrow y$, and
 $\epsilon' \leftarrow (1+\delta)\epsilon^{\circ}$. 

\subsection{Proof outline}
The proof of Theorem \ref{thm:GreedyMin-max} has three main components:

\begin{enumerate}
    \item We start by showing that Algorithm \ref{alg:LocalMin-max} halts after $i^\star = O\left(\frac{b}{\gamma_1}\right)$ iterations, which allows us to bound the number of oracle calls until Algorithm \ref{alg:LocalMin-max} halts (Lemma \ref{lemma:RunningTime}).

\item We then show that the point $(x^\star, y^\star)$ returned by Algorithm \ref{alg:LocalMin-max} is an $(\epsilon^\star, \sigma)$-greedy adversarial equilibrium for $f$, for some 
\begin{equation*}
\epsilon^\star = \epsilon_{i^\star} = \epsilon_0 \leq  (1+\delta)^{O(i^\star)} \leq \epsilon.
\end{equation*}
Towards this end we first show that $y^\star$ is an $(\epsilon^\star, \sqrt{L \epsilon^\star})$-approximate local maximum of $f(x^\star, \cdot)$ (in the sense of the definition in \eqref{eq:ApproximateLocalMin_formal}) and therefore satisfies Inequalities \eqref{eq:first_order_y} of Definition \ref{def:local_min-max_formal}.

\item Next, we show that $x^\star$ is an $(\eps^\star, \sigma)$-approximate local minimum of  $g_\epsilon(\cdot, y^\star))$ (in the sense of Definition \ref{def_discontinuos_local_min}), and therefore satisfies Inequalities \eqref{eq:first_order_x} of Definition \ref{def:local_min-max_formal} (see the proof in Section \ref{concluding_the_proof}).
Towards this end, we prove (in Lemma \ref{Lemma_SharedLM}) that to guarantee  $x^\star$ is an $(\eps^\star, \sigma)$-approximate local minimum of  $g_\epsilon(\cdot, y^\star)$, it is sufficient to show that $x^\star$ is an $(\eps^\star, \sigma)$-approximate local minimum of  $h_\epsilon(\cdot, y^\star)$ (Propositions \ref{Prop_first_order_x} and \ref{Prop_NoisySGD}), provided that we can also show that $h_\eps$ and $g_\eps$ satisfy a number of conditions.
Namely, these conditions are:
\begin{enumerate}
    \item $h_\epsilon(x,y) \leq g_\epsilon(x,y)$ for all $(x,y) \in \mathbb{R}^d \times \mathbb{R}^d$ (Proposition \ref{Prop_lower_bound}).
    \item $h_\epsilon(x^\star,y^\star) = g_\epsilon(x^\star,y^\star)$ (Proposition \ref{Prop_fixed_point}).
    \item A stochastic gradient for a smoothed version of $h_\eps$  (defined in Equation \eqref{eq_SG_H}) has a very small expected magnitude (Proposition \ref{prop_SGVariance}).

    \end{enumerate}
    \end{enumerate}
\noindent    
    In the remainder of this section, we give proofs of the main lemmas and propositions we use to prove Theorem \ref{thm:GreedyMin-max}.  We conclude the proof of Theorem \ref{thm:GreedyMin-max} in Section \ref{concluding_the_proof}.
    
     A diagram of the proof structure is given in Figure \ref{fig:proof_diagram}.  

\begin{figure}
\begin{tikzpicture}[
roundnode/.style={circle, draw=green!60, fill=green!5, very thick, minimum size=7mm},
squarednode/.style={rectangle, draw=black!60, rounded corners=.2cm, fill=black!0, very thick, minimum size=5mm},
]
\node[squarednode, text width=10cm,align=center]      (Main_Theorem)                              {\textbf{Theorem \ref{thm:GreedyMin-max}:} Algorithm \ref{alg:LocalMin-max} converges in polynomial time to a point $(x^\star, y^\star)$ which is a greedy adversarial equilibrium.};
\node[squarednode,text width=8.4cm,align=center]      (Lemma_3)       [below=of Main_Theorem] {\textbf{Lemma \ref{Lemma_SharedLM}:}  If Props. \ref{Prop_lower_bound}-\ref{prop_SGVariance} hold for an approximate local minimizer $x^\star$ of $\mathfrak{h}_{\eps, \sigma}(\cdot, y^\star)$, the $x^\star$ is also an approximate local minimizer of $\mathfrak{g}_{\eps, \sigma}(\cdot, y^\star)$.};
\node[squarednode, text width=3cm,align=center]      (Lemma_2)       [right=of Lemma_3, xshift=-0.2cm] {\textbf{Lemma \ref{lemma:Player2}:} $y^\star$ is an approximate local minimizer of $f(x^\star, \cdot)$};
\node[squarednode, text width=3cm,align=center]      (Lemma_1)       [left=of Lemma_3, xshift=0cm] {\textbf{Lemma \ref{lemma:RunningTime}:} Algorithm \ref{alg:LocalMin-max} converges in polynomial time.};
\node[squarednode, text width=3cm,align=center]      (Prop_8)       [below=of Lemma_3, xshift=-3cm, yshift=-1.5cm] {\textbf{Prop. \ref{prop_SGVariance}:} Stochastic gradient has low variance at $(x^\star, y^\star)$.};
\node[squarednode,text width=3cm,align=center]      (Prop_9)       [right=of Prop_8,xshift=1cm, yshift=1.5cm] {\textbf{Prop. \ref{Prop_first_order_x}:} $x^\star$ is a 1st-order stationary point of $\mathfrak{h}_{\eps, \sigma}^{x^\star}(\cdot, y^\star)$.};
\node[squarednode, text width= 3cm,align=center]      (Prop_10)       [right=of Prop_9,xshift=-0.6cm, yshift= -1.5cm] {\textbf{Prop. \ref{Prop_NoisySGD}:} $x^\star$ satisfies 2nd-order condition to be an approximate local minimizer of $\mathfrak{h}_{\eps, \sigma}^{x^\star}(\cdot, y^\star)$};
\node[squarednode,text width=4cm,align=center]      (Prop_4)       [below=of Prop_8,  yshift=-3.5cm] {\textbf{Prop. \ref{Lemma_Greedypath}:} Algorithm \ref{alg:InnerMaxLoop} computes a greedy path.};
\node[squarednode, text width=6cm,align=center]      (Prop_6)       [right=of Prop_4, xshift=-2cm,  yshift=3.5cm] {\textbf{Prop. \ref{Prop_fixed_point}:} fixed point property: $h(x^\star, y^\star) = g_\eps(x^\star, y^\star) = f(x^\star, y^\star)$};
\node[squarednode, text width=3cm,align=center]      (Prop_5)       [left=of Prop_4,  xshift=1.8cm, yshift=3.5cm] {\textbf{Prop. \ref{Prop_lower_bound}:} $h_\eps$ is a lower bound for greedy max function.};
\node[squarednode, text width=5cm,align=center]      (Prop_7)       [right=of Prop_4, xshift=2cm,  yshift=1cm] {\textbf{Prop. \ref{lemma_SG}:} $\Gamma_k$ is a stochastic gradient for $\mathfrak{h}_{\eps, \sigma}^{x^\star}(\cdot, y^\star)$};

\draw[<-, thick] (Main_Theorem.south) +(right:40mm) -- (Lemma_2.north);
\draw[<-, thick] (Main_Theorem.south) +(left:40mm) -- (Lemma_1.north);
\draw[<-, thick] (Main_Theorem.south) -- (Lemma_3.north);
\draw[<-,blue, thick] (Lemma_3.south) +(right:30mm) -- (Prop_10.north);
\draw[<-,blue, thick] (Lemma_3.south) +(left:30mm) -- (Prop_8.north);
\draw[<-,blue, thick] (Lemma_3.south) +(right:13mm)-- (Prop_9.north);
\draw[->, thick] (Prop_8.east)-- (Prop_9.west);
\draw[->, thick] (Prop_8.east) +(down:5mm)--(Prop_10.west);
\draw[<-,blue, thick] (Lemma_3.south) +(left:4cm)-- (Prop_5.north);
\draw[<-,blue, thick] (Lemma_3.south) +(left:1cm)-- (Prop_6.north);
\draw[<-, thick] (Prop_8.south) +(left:7mm) -- (Prop_5.east);
\draw[<-, thick] (Prop_8.south)+(left:3mm) -- (Prop_4.north);
\draw[<-, thick] (Prop_8.south)+(right:1mm) -- (Prop_6.west);
\draw[<-, thick] (Prop_5.south) -- (Prop_4.west);
\draw[->, thick]  (Prop_4.north)+(right:10mm)--(Prop_6.south);
\draw[<-, thick] (Prop_7.west) -- (Prop_4.east);
\draw[<-, thick] (Prop_10.south) -- (Prop_7.north);

\end{tikzpicture}
\caption{A diagram of the proof of the main theorem.  A black arrow means that a lemma or proposition was used to prove another lemma, proposition or theorem.  The blue arrows pointing from propositions \ref{Prop_lower_bound}-\ref{Prop_NoisySGD} to Lemma \ref{Lemma_SharedLM} mean that those propositions were used to satisfy the conditions of Lemma \ref{Lemma_SharedLM}.}\label{fig:proof_diagram}
\end{figure}
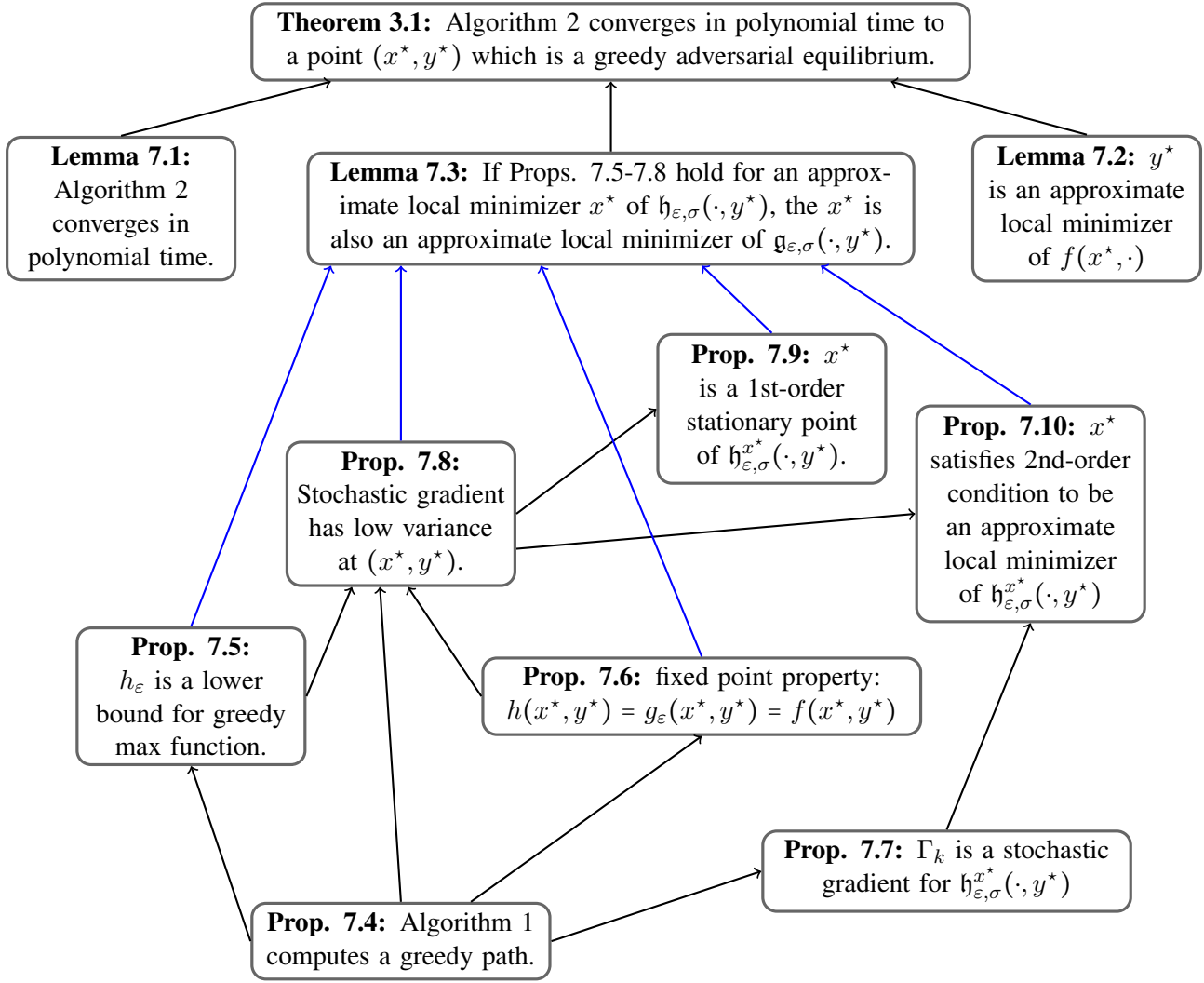

\newpage

\subsection{Gradient, function, and Hessian evaluations} 

\begin{lemma}[{Bounding the number of gradient, Hessian, and function evaluations}] \label{lemma:RunningTime} 
%
Algorithm \ref{alg:LocalMin-max} terminates after at most  $i^\star = O\left(\frac{b}{\gamma_1}\right)$ iterations, and thus takes at most $O\left(\frac{b}{\gamma_1}  \times  (\mathcal{I}_2 \mathcal{I}_4+ \mathcal{I}_3) \times \frac{b}{ \mu_1 \mu_3^2 L_{}}\right)$ gradient, Hessian, and function evaluations. 
\end{lemma}

\begin{proof}

\textbf{Bounding the iterations of Algorithm \ref{alg:InnerMaxLoop}}:
First, we bound the number of iterations of the  ``While" loop in Algorithm \ref{alg:InnerMaxLoop}. 
We begin by showing that $\| \nabla_y f(\mathsf{x}, \mathsf{y}^{\aell})\| > \epsilon'$ occurs at least once every $\frac{4 \epsilon'}{\mu_3^2 L_{}}$ iterations.

Consider any iteration $\ell$ where $\| \nabla_y f(\mathsf{x}, \mathsf{y}^{\aell})\| \leq \epsilon'$.  Then (unless Algorithm \ref{alg:InnerMaxLoop} terminates at step $\aell$ of the While loop) Line \ref{GreedyEnpoint_Start} of Algorithm \ref{alg:InnerMaxLoop} implies that we have both
\begin{equation*}
v^\top \nabla^2_{{y}}f(\mathsf{x}, \mathsf{y}^{\aell}) v \geq \sqrt{L_{} \epsilon'},
\quad \textrm{ and } \quad
 \mathsf{y}^{\aell+1} = \mathsf{y}^\ell + \mu_3 \mathsf{a} v,
\end{equation*}
 where $\mathsf{a} = \mathrm{sign}(\nabla_yf(\mathsf{x},\mathsf{y}^{\aell})^\top v)$.
\noindent 
  Hence, since $\mu_3 \leq \frac{\sqrt{\epsilon'}}{2\sqrt{L_{}}}$, we have that
\begin{align*}
\nabla_y f(\mathsf{x}, \mathsf{y}^{\aell+1}) - \nabla_y f(\mathsf{x}, \mathsf{y}^{\aell}) &\geq  [\nabla_y^2 f(\mathsf{x}, \mathsf{y}^{\aell}) -L_{} \mu_3 I_d] \mu_3 \mathsf{a} v\\
&= (\sqrt{L_{} \epsilon'} - L_{} \mu_3) \mu_3 \mathsf{a} v\\
&\geq  \frac{1}{2}\sqrt{L_{} \epsilon'} \mu_3 \mathsf{a} v.
\end{align*}

\noindent
Therefore, since $\mathsf{a} = \mathrm{sign}(\nabla_yf(\mathsf{x},\mathsf{y}^{\aell})^\top v)$, we have that,
\be \label{eq:a1}
\|\nabla_y f(\mathsf{x}, \mathsf{y}^{\aell+1})\|^2 - \|\nabla_y f(\mathsf{x}, \mathsf{y}^{\aell})\|^2 \geq  \left(\frac{1}{2}\sqrt{L_{} \epsilon'} \mu_3 \right)^2.
\ee
Therefore, we have that the magnitude of the gradient becomes $\geq \epsilon'$ at least once every $\mathfrak{c} := \frac{4 \epsilon'}{\mu_3^2 L_{}}$ iterations of the While loop of Algorithm \ref{alg:InnerMaxLoop}. 

Since $\mu_3 \leq \frac{\sqrt{\epsilon}}{4\sqrt{L_{}}} \leq \frac{\sqrt{\epsilon'}}{2\sqrt{L_{}}}$, for some unit vector $w$ we have that
\be \label{eq:a2}
 f(\mathsf{x}, \mathsf{y}^{\aell+1}) -  f(\mathsf{x}, \mathsf{y}^{\aell})
 &\geq \mu_3 \mathsf{a} \nabla_y f(\mathsf{x}, \mathsf{y}^{\aell}) ^\top v  + \frac{1}{2} \mu_3^2 v^\top (\nabla_y ^2 f(\mathsf{x}, \mathsf{y}^{\aell}) - \mu_3 L_{} I_d) v\\
  &\geq 0  + \frac{1}{2} \mu_3^2 \sqrt{L_{} \epsilon'}  - \frac{1}{2} \mu_3^3 L_{} 
  \geq \frac{1}{4} \mu_3^2 \sqrt{L_{} \epsilon'} \geq 0.
 \ee
Recall that we have also shown (Inequality \eqref{eq:a1}) that the gradient becomes $\geq \epsilon'$ at least once every $\mathfrak{c} := \frac{4 \epsilon'}{\mu_3^2 L_{}}$ iterations of the While loop of Algorithm \ref{alg:InnerMaxLoop}.  Therefore, since $\mu_1 \leq \frac{1}{2 L_2}$ and $f$ is uniformly bounded by $b$, we must have that the While loop in Algorithm \ref{alg:InnerMaxLoop} terminates after $O\left(\frac{b \mathfrak{c} }{\mu_1 \epsilon}\right)$ iterations, and hence that  Algorithm \ref{alg:InnerMaxLoop} terminates after $O\left(\frac{b}{ \mu_1 \mu_3^2 L_{}}\right)$ oracle evaluations.

\bigskip

\noindent
\textbf{Bounding the iterations of Algorithm \ref{alg:LocalMin-max}}:
At each iteration $\ai$ of the While loop in Algorithm \ref{alg:LocalMin-max} except for the last iteration $i^\star$, Lines \ref{HillClimbing} and \ref{SuccessStart} of Algorithm \ref{alg:LocalMin-max} imply that
\be \label{eq:a5}
f(x_{\ai+1}, y_{\ai+1}) \leq f(x_{\ai}, y_{\ai}) - \gamma_1.
\ee

\noindent
Therefore, since $f$ is uniformly bounded  by $b$, Inequality \eqref{eq:a5} implies that the While loop of Algorithm \ref{alg:LocalMin-max} terminates after at most $i_{\max}$ iterations for some number $i_{\max} = O\left(\frac{b}{\gamma_1}\right)$. 
 Therefore since we have already shown that Algorithm \ref{alg:InnerMaxLoop} terminates in at most $O\left(\frac{b}{ \mu_1 \mu_3^2 L_{}}\right)$ oracle calls each time it is called, and Algorithm \ref{alg:InnerMaxLoop} is called $\mathcal{I}_3 + \mathcal{I}_2 \mathcal{I}_4$ times at each iteration of the While loop, running Algorithm \ref{alg:InnerMaxLoop} contributes at most $O\left(\frac{b}{\gamma_1} \times (\mathcal{I}_3 + \mathcal{I}_2 \mathcal{I}_4)  \times \frac{b}{ \mu_1 \mu_3^2 L_{}} \right)$ oracle calls to the cost of Algorithm \ref{alg:LocalMin-max}. 
  Since the other parts of the While loop make at most $O(\mathcal{I}_3 + \mathcal{I}_2 \mathcal{I}_4)$ function evaluations, they contribute no more that $O\left(\frac{b}{\gamma_1} \times (\mathcal{I}_3 + \mathcal{I}_2 \mathcal{I}_4)\right)$ function evaluations to the cost of Algorithm \ref{alg:LocalMin-max}. 
  Therefore, Algorithm \ref{alg:LocalMin-max} terminates after at most $O\left(\frac{b}{\gamma_1} \times (\mathcal{I}_3 + \mathcal{I}_2 \mathcal{I}_4) \times \frac{b}{ \mu_1 \mu_3^2 L_{}}\right)$ oracle calls.
\end{proof}

\subsection{$y^\star$ an approximate local maximum of $f(x^\star, \cdot)$}

\begin{lemma} [{Approximate local maximum in $y$}] \label{lemma:Player2}
The output  $\mathsf{y}_{\mathrm{LocalMax}}$ of Algorithm \ref{alg:InnerMaxLoop} with inputs  $\mathsf{x}, \mathsf{y}^0, \epsilon'$ satisfies
\begin{equation}\label{eq:b1a}
\|\nabla_y f(\mathsf{x}, \mathsf{y}_{\mathrm{LocalMax}})\| \leq  \epsilon'
\end{equation}
and
\begin{equation}\label{eq:b1b}
\lambda_{\mathrm{max}}( \nabla^2_{{y}} f(\mathsf{x}, \mathsf{y}_{\mathrm{LocalMax}})) \leq \sqrt{L_{} \epsilon'}.
\end{equation}

\noindent In particular, this implies that the output $(x^\star, y^\star)$ of Algorithm \ref{alg:LocalMin-max} satisfies
\begin{equation} \label{eq:b1}
\|\nabla_y f(x^\star, y^\star)\| \leq  \epsilon_{\ai^\star},
\end{equation}
and
\begin{equation} \label{eq:b2}
\lambda_{\mathrm{max}}( \nabla^2_{{y}} f(x^\star, y^\star)) \leq \sqrt{L_{}  \epsilon_{\ai^\star}}.
\end{equation}

\end{lemma}

\begin{proof}
First, we note that Lines \ref{NoislessSGD_Start} and \ref{GreedyEnpoint_Start} of Algorithm \ref{alg:InnerMaxLoop} imply that for Algorithm \ref{alg:InnerMaxLoop} (with inputs  $\mathsf{x}, \mathsf{y}^0, \epsilon'$) to stop at a point $(\mathsf{x}, \mathsf{y}_{\mathrm{LocalMax}})$ we must have that both Inequalities \eqref{eq:b1a} and \eqref{eq:b1b} hold.

Therefore, the output $(x^\star, y^\star)$ of Algorithm \ref{alg:LocalMin-max} satisfies
\begin{equation} \label{eq:b3}
\|\nabla_yf (x^\star, y^\star)\| \leq \epsilon_{i^\star-1}(1+\delta),
\end{equation}
and
\begin{equation} \label{eq:b4}
 \lambda_{\mathrm{max}}(\nabla^2_y f(x^\star, y^\star)) \leq \sqrt{L_{}  \epsilon_{i^\star-1}(1+\delta)} + \mu_4.
\end{equation}

\noindent Thus, Inequality \eqref{eq:b3} implies that
\begin{equation*}
\|\nabla_y f(x^\star, y^\star)\| \leq  \epsilon_{i^\star-1}(1+\delta)  \leq   \epsilon_{i^\star-1}(1+\delta)^2   \leq \epsilon_{i^\star},
\end{equation*}
 and hence that Inequality \eqref{eq:b1} holds.

Inequality \eqref{eq:b4} implies that
  \be
\nabla^2_y f(x^\star, y^\star) \preceq \lambda_{\mathrm{max}}(\nabla^2_y f(x^\star, y^\star)) \preceq \sqrt{L_{}  \epsilon_{i^\star-1}(1+\delta)} I_d + \mu_4 I_d.
  \ee
Therefore,
\be \label{eq:b6}
\nabla^2_y f(x^\star, y^\star) \preceq (\sqrt{L_{}  \epsilon_{i^\star-1}(1+\delta)} + \mu_4) I_d \preceq \sqrt{L_{}  \epsilon_{i^\star-1}(1+\delta)^2} I_d = \sqrt{L_{}  \epsilon_{i^\star}} I_d,
\ee
 since $\mu_4 \leq \frac{1}{2}(\sqrt{1+\delta}- 1)\sqrt{L_{} \epsilon}$.
Therefore Inequality \eqref{eq:b6} implies that Inequality \eqref{eq:b2} holds.
\end{proof}

\subsection{$x^\star$ an approximate local minimum of $g_\eps(\cdot, y^
\star)$}

We start by defining some functions which will be useful in showing that  $x^\star$ is an approximate local minimum of $g_\eps(\cdot, y^
\star)$.

 \noindent
 For any $\hat{x} \in \mathbb{R}^d$, $\epsilon^{\circ}>0$, let 
\be
\mathbbm{g}^{\hat{x}}_{\epsilon^{\circ}}(x,y) \coloneqq \min (g_{\epsilon^{\circ}}(x,y), g_{\epsilon^{\circ}}(\hat{x},y)),\\
\mathbbm{h}^{\hat{x}}_{\epsilon^{\circ}}(x,y) \coloneqq \min (h_{\epsilon^{\circ}}(x,y), h_{\epsilon^{\circ}}(\hat{x},y)),
\ee
 and 
 \be
  \mathfrak{g}_{\epsilon^{\circ}, \sigma}^{\hat{x}}(x,y) \coloneqq \mathbb{E}_{\zeta \sim N(0,I_d)}\left [\mathbbm{g}_{\epsilon^{\circ}}^{\hat{x}}(x  + \sigma \zeta,y)\right],\\
 \mathfrak{h}_{\epsilon^{\circ}, \sigma}^{\hat{x}}(x,y) \coloneqq \mathbb{E}_{\zeta \sim N(0,I_d)}\left [\mathbbm{h}_{\epsilon^{\circ}}^{\hat{x}}(x  + \sigma \zeta,y)\right].
 \ee

\noindent
Finally, for any $\hat{x} \in \mathbb{R}^d$, $\epsilon^{\circ}>0$, define the stochastic gradient 
\be \label{eq_SG_H}
\mathcal{H}_{\epsilon^{\circ}}^{\hat{x}}(x,y)\coloneqq \frac{\zeta}{\sigma}(\mathbbm{h}_{\epsilon^{\circ}}^{\hat{x}}(x  + \sigma \zeta,y) - \mathbbm{h}_{\epsilon^{\circ}}^{\hat{x}}(x ,y)),
\ee
 where $\zeta \sim N(0,I_d)$.

\subsubsection{Showing that $g_\eps$ and $h_\eps$ have shared approximate local minima}

The following lemma allows us to guarantee that if we can show that our algorithm returns a point $(x^\star, y^\star)$ where $x^\star$ is an approximate local minimum for $h_\eps(\cdot, y^\star)$ for some $y^\star \in \mathbb{R}^d$, then $x^\star$ is also an approximate local minimum for the greedy max function $g_\eps(\cdot, y^\star)$,
provided that we can also show that $h_\eps$, $g_\eps$, $x^\star$ and $y^\star$ satisfy certain conditions.

\begin{lemma}[{Shared local minima of $g_\eps$ and $h_\eps$}] \label{Lemma_SharedLM}
Consider any $\epsilon>0$.  
Suppose that $\sigma \leq \frac{1}{\sqrt{\epsilon d}}$ and that for some point $(x^\star, y^\star) \in \mathbb{R}^d \times \mathbb{R}^d$ we have
\begin{equation}\label{eq:c11}\qquad h_\epsilon(x,y) \leq g_\epsilon(x,y) \qquad \forall x,y \in \mathbb{R}^d \qquad \textrm{(lower bound)},  \end{equation}
\begin{equation}\label{eq:c7} \quad h_\epsilon(x^\star,y^\star) = g_\epsilon(x^\star,y^\star)  \qquad \qquad \textrm{(fixed-point  property)}, \end{equation}
\begin{equation}\label{eq:c8}   \qquad \quad \mathbb{E}[\|\mathcal{H}_\epsilon^{x^\star}(x^\star,y^\star)\|] \leq \frac{1}{8000}  \frac{\sigma^{14} \epsilon^{1.5}}{b^2} \ \  \textrm{(low SG)}, \end{equation}
\begin{equation} \label{eq:c23}   \|\nabla_x \mathfrak{h}_{\epsilon, \sigma}^{x^\star}(x^\star, y^\star)\| \leq \frac{\epsilon^2\sigma^7}{8000b} \ \ \textrm{(1st-order stationarity for $\mathfrak{h}$) },
\end{equation}

\begin{equation}\label{eq:c10} \lambda_{\mathrm{min}}(\nabla^2_{{x}} \mathfrak{h}_{\epsilon, \sigma}^{x^\star}(x^\star, y^\star)) \geq - \frac{\sqrt{\epsilon} }{5} \ \ \textrm{(2nd-order stationarity)}. \end{equation}
Then
\be \label{eq:c2}
\|\nabla_x \mathfrak{g}_{\epsilon, \sigma}^{x^\star}(x^\star, y^\star)\| \leq \epsilon \ \ {\textrm and}
\ee
\be \label{eq:c3}
 \lambda_{\mathrm{min}}(\nabla^2_{{x}} \mathfrak{g}_{\epsilon, \sigma}^{x^\star}(x^\star, y^\star)) \geq - \sqrt{\epsilon}.
\ee
\end{lemma}

\begin{proof}

\textbf{Showing first-order condition holds:}
First, note that
\be
&\mathbbm{h}^{x^\star}_{\epsilon}(x,y^\star) = \left [\min(h_{\epsilon}(x, y^\star), h_{\epsilon}(x^\star, y^\star))\right]\\
&\stackrel{\textrm{Eq.} \ref{eq:c11}}{\leq} \min(g_{\epsilon}(x, y^\star), h_{\epsilon}(x^\star, y^\star))\\
&\stackrel{\textrm{Eq.} \ref{eq:c7}}{=} \min(g_{\epsilon}(x, y^\star), g_{\epsilon}(x^\star, y^\star))\\
&= \mathbbm{g}^{x^\star}_{\epsilon}(x,y^\star)  \quad \forall x \in \mathbb{R}^d.
\ee

\noindent
Define the stochastic gradient $\mathcal{G}(x,y):= \frac{\zeta}{\sigma}(\mathbbm{g}_\epsilon^{x^\star}(x  + \sigma \zeta,y) - \mathbbm{g}_\epsilon^{x^\star}(x ,y))$ for all $x,y \in \mathbb{R}^d$, where $\zeta \sim N(0,I_d)$.
 Since $\mathbbm{h}^{x^\star}_{\epsilon}$ and  $\mathbbm{g}^{x^\star}_{\epsilon}$  are uniformly bounded, by Lemma 7 of \cite{zhang2017hitting}, we have that 
\begin{equation*}
\nabla_{{x}} \mathfrak{h}_{\epsilon, \sigma}^{x^\star}(x^\star, y^\star) = \mathbb{E}\left[\mathcal{H}_\epsilon^{x^\star}(x^\star, y^\star)\right],
\end{equation*}
\begin{equation*}
 \qquad \nabla_{{x}} \mathfrak{g}_{\epsilon, \sigma}^{x^\star}(x^\star, y^\star) := \mathbb{E}\left[ \mathcal{G}(x^\star, y^\star) \right].
\end{equation*}

\noindent
Inequality \eqref{eq:c8} implies that
\be \label{eq:c1}
\frac{1}{8000}  \frac{\sigma^{14} \epsilon^{1.5}}{b^2} &\geq \mathbb{E}\left[\left\|\mathcal{H}_\epsilon^{x^\star}(x^\star,y^\star)\right\|\right]\\
&= \mathbb{E} [\| \zeta \sigma^{-1}  (\min(h_{\epsilon}(x^\star+ \sigma \zeta, y^\star), h_{\epsilon}(x^\star, y^\star)) - \min(h_{\epsilon}(x^\star, y^\star), h_{\epsilon}(x^\star, y^\star))\|]\\
&=\mathbb{E}[\|\zeta \sigma^{-1} (\min(h_{\epsilon}(x^\star+ \sigma \zeta, y^\star), h_{\epsilon}(x^\star, y^\star)) -  h_{\epsilon}(x^\star, y^\star)\|]\\
&= \mathbb{E} [\|\zeta\| \sigma^{-1}\times  |\min(h_{\epsilon}(x^\star+ \sigma \zeta, y^\star), h_{\epsilon}(x^\star, y^\star)) -  h_{\epsilon}(x^\star, y^\star) |]\\
&\geq\mathbb{E}[\|\zeta\| \sigma^{-1}] \times \mathbb{E}[  | \min(h_{\epsilon}(x^\star+ \sigma \zeta, y^\star), h_{\epsilon}(x^\star, y^\star)) -  h_{\epsilon}(x^\star, y^\star)  |]\\
&= \sqrt{d}\sigma^{-1} \mathbb{E}[  | \min(h_{\epsilon}(x^\star+ \sigma \zeta, y^\star), h_{\epsilon}(x^\star, y^\star)) -  h_{\epsilon}(x^\star, y^\star) |]\\
&\stackrel{\textrm{Eq.} \ref{eq:c7}}{=}  \sqrt{d}\sigma^{-1} \mathbb{E}[ \left | \min(h_{\epsilon}(x^\star+ \sigma \zeta, y^\star), g_{\epsilon}(x^\star, y^\star)) -  g_{\epsilon}(x^\star, y^\star) \right |]\\
&\stackrel{\textrm{Eq.} \ref{eq:c11}}{\geq}\sqrt{d}\sigma^{-1} \mathbb{E}[ |\min(g_{\epsilon}(x^\star+ \sigma \zeta, y^\star), g_{\epsilon}(x^\star, y^\star)) -  g_{\epsilon}(x^\star, y^\star) |],
\ee
where the second inequality holds since $E[XY] \geq E[X] E[Y]$ for any nonnegative random variables $X, Y$. 
The last inequality holds since $\min(h_{\epsilon}(x^\star+ \sigma \zeta, y^\star), g_{\epsilon}(x^\star, y^\star)) -  g_{\epsilon}(x^\star, y^\star) \leq 0$ and since Inequality \eqref{eq:c11} implies that $h_{\epsilon}(x^\star+ \sigma \zeta, y^\star) \leq g_{\epsilon}(x^\star+ \sigma \zeta, y^\star)$.
  Thus,
\be 
\|&\nabla_{{x}} \mathfrak{g}_{\epsilon, \sigma}^{x^\star}(x^\star, y^\star)\|\\
&=  \| \mathbb{E}_{\zeta\sim N(0,I_d)}[\zeta \sigma^{-1} (\mathbbm{g}^{x^\star}_{\epsilon}(x^\star+ \sigma \zeta,y^\star) - \mathbbm{g}^{x^\star}_{\epsilon}(x^\star,y^\star))]  \|\\
&=  \| \mathbb{E}_{\zeta\sim N(0,I_d)}[\zeta \sigma^{-1} (\min(g_{\epsilon}(x^\star+ \sigma \zeta, y^\star), g_{\epsilon}(x^\star, y^\star)) - \min(g_{\epsilon}(x^\star, y^\star), g_{\epsilon}(x^\star, y^\star))] \| \\
&=  \|\mathbb{E}_{\zeta\sim N(0,I_d)}[\zeta \sigma^{-1} (\min(g_{\epsilon}(x^\star+ \sigma \zeta, y^\star), g_{\epsilon}(x^\star, y^\star)) -  g_{\epsilon}(x^\star, y^\star))]\|\\
&\leq \mathbb{E}[ \|\zeta \sigma^{-1} (\min(g_{\epsilon}(x^\star+ \sigma \zeta, y^\star), g_{\epsilon}(x^\star, y^\star)) -  g_{\epsilon}(x^\star, y^\star))\| ]\\
&= \mathbb{E}\left[ \|\zeta \sigma^{-1}\| \times  | (\min(g_{\epsilon}(x^\star+ \sigma \zeta, y^\star), g_{\epsilon}(x^\star, y^\star)) -  g_{\epsilon}(x^\star, y^\star))| \right]\\
&\leq \mathbb{E}\left[ \sqrt{d} \sigma^{-1}  \log\left(\left(\frac{10 b \sqrt{d}}{\eps \sigma}\right)^2\right) \times | (\min(g_{\epsilon}(x^\star+ \sigma \zeta, y^\star), g_{\epsilon}(x^\star, y^\star)) -  g_{\epsilon}(x^\star, y^\star))| \right]\\
&\qquad \qquad \qquad +\mathbb{P}\left(\|\zeta\| \geq  \sqrt{d}  \log\left(\left(\frac{10 b \sqrt{d}}{\eps \sigma}\right)^2\right) \right) \times \mathbb{E}\left[ \| \zeta \sigma^{-1} \| \times 2b  \, \, \,\, \,  \,   \bigg | \, \,  \,  \, \,  \,  \|\zeta\| \geq  \sqrt{d}  \log\left(\left(\frac{10 b \sqrt{d}}{\eps \sigma}\right)^2\right) \right]\\
&\leq \mathbb{E}\left[ \sqrt{d} \sigma^{-1}  \log\left(\left(\frac{10 b \sqrt{d}}{\eps \sigma}\right)^2\right) \times | (\min(g_{\epsilon}(x^\star+ \sigma \zeta, y^\star), g_{\epsilon}(x^\star, y^\star)) -  g_{\epsilon}(x^\star, y^\star))| \right]\\
&\qquad \qquad \qquad + 2 \left(\frac{\eps \sigma}{10 b \sqrt{d}}\right)^2 \times \mathbb{E}\left[ \| \zeta \sigma^{-1} \| b  \, \, \, \, \, \, \bigg | \, \, \, \, \, \,  \|\zeta\| \geq  \sqrt{d}  \log\left(\left(\frac{10 b \sqrt{d}}{\eps \sigma}\right)^2\right) \right]\\
&\leq \mathbb{E}\left[ \sqrt{d} \sigma^{-1} \log\left(\left(\frac{10 b \sqrt{d}}{\eps \sigma}\right)^2\right) \times | (\min(g_{\epsilon}(x^\star+ \sigma \zeta, y^\star), g_{\epsilon}(x^\star, y^\star)) -  g_{\epsilon}(x^\star, y^\star))| \right]\\
&\qquad \qquad \qquad + 4 \left(\frac{\eps \sigma}{10 b \sqrt{d}}\right)^2 \sigma^{-1}  b \sqrt{d} \log\left(\left(\frac{10 b \sqrt{d}}{\eps \sigma}\right)^2\right)\\
&\stackrel{\textrm{Eq.} \ref{eq:c1}}{\leq}
\frac{1}{8000}  \frac{\sigma^{14} \epsilon^{1.5}}{b^2}  \log\left(\left(\frac{10 b \sqrt{d}}{\eps \sigma}\right)^2\right) + 4 \left(\frac{\eps \sigma}{10 b \sqrt{d}}\right)^2 \sigma^{-1}  b \sqrt{d} \log\left(\left(\frac{10 b \sqrt{d}}{\eps \sigma}\right)^2\right)\\
&\leq \epsilon,
\ee
\noindent
 where the third and fourth inequalities hold by a standard concentration bound for Gaussians \cite{hanson1971bound}, and the last inequality holds because $\sigma \leq \frac{1}{\sqrt{\epsilon d}}$, $b\geq 1$, and  $\sigma, \epsilon \leq 1$.
This shows inequality \eqref{eq:c2}.\\

\noindent \textbf{Showing that second-order condition holds:}
First, we will show that, roughly speaking, $\mathfrak{g}_{\epsilon, \sigma}^{x^\star}(x^\star, y^\star) \approx \mathfrak{h}_{\epsilon, \sigma}^{x^\star}(x^\star, y^\star)$.   Since $g_{\epsilon}(x^\star, y^\star) = h_{\epsilon}(x^\star, y^\star)$, we have, by Inequality \eqref{eq:c1} that

\be \label{eq:c5}
0 &\leq g_{\epsilon}(x^\star, y^\star) - \mathfrak{h}_{\epsilon, \sigma}^{x^\star}(x^\star, y^\star)\\
&=  \mathbb{E}[ g_{\epsilon}(x^\star, y^\star) - \min(h_{\epsilon}(x^\star+ \sigma \zeta, y^\star), h_{\epsilon}(x^\star, y^\star))]  \\
&\stackrel{\textrm{Eq.} \ref{eq:c7}}{=}  \mathbb{E}[ g_{\epsilon}(x^\star, y^\star) - \min(h_{\epsilon}(x^\star+ \sigma \zeta, y^\star), g_{\epsilon}(x^\star, y^\star))] \\
&= \mathbb{E}[|g_{\epsilon}(x^\star, y^\star) - \min(h_{\epsilon}(x^\star+ \sigma \zeta, y^\star), g_{\epsilon}(x^\star, y^\star))|] \\
&\stackrel{\textrm{Eq.} \ref{eq:c1}}{\leq} \frac{1}{8000} \frac{\sigma^{14} \epsilon^{1.5}}{b^2} \frac{\sigma}{\sqrt{d}},
\ee
and
\be  \label{eq:c21}
0 &\geq \mathfrak{g}_{\epsilon, \sigma}^{x^\star}(x^\star, y^\star) - g_{\epsilon}(x^\star, y^\star)\\
&=  \mathbb{E}[\min(g_{\epsilon}(x^\star+ \sigma \zeta, y^\star), g_{\epsilon}(x^\star, y^\star)) -  g_{\epsilon}(x^\star, y^\star)]  \\
&\stackrel{\textrm{Eq.} \ref{eq:c1}}{\geq} -\frac{1}{8000} \frac{\sigma^{14} \epsilon^{1.5}}{b^2} \frac{\sigma}{\sqrt{d}}.
\ee
Thus, \eqref{eq:c5} and \eqref{eq:c21} together  with \eqref{eq:c11} and \eqref{eq:c7} imply that

\begin{equation}  \label{eq:c6}
0 \leq \mathfrak{g}_{\epsilon, \sigma}^{x^\star}(x^\star, y^\star) -  \mathfrak{h}_{\epsilon, \sigma}^{x^\star}(x^\star, y^\star) \leq \frac{1}{4000} \frac{\sigma^{14} \epsilon^{1.5}}{b^2} \frac{\sigma}{\sqrt{d}}
\leq \frac{1}{4000} \frac{\sigma^{14} \epsilon^{1.5}}{b^2} \frac{1}{\sqrt{d}},
\end{equation}
where the last inequality holds since $\sigma \leq 1$.

\noindent \textbf{Contradiction argument:}
We will show \eqref{eq:c3} by contradiction.  Towards this end, suppose that the following statement were true
\be \label{eq:c4}
\lambda_{\mathrm{min}}(\nabla^2_{{x}} \mathfrak{g}_{\epsilon, \sigma}^{x^\star}(x^\star, y^\star)) < - \sqrt{\epsilon}.
\ee
Then there would exist a unit vector $v$ such that
\be \label{eq:c20}
v^\top \nabla^2_{{x}} \mathfrak{g}_{\epsilon, \sigma}^{x^\star}(x^\star, y^\star) v < - \sqrt{\epsilon}.
\ee
Let $\hat{\mathsf{a}} := \mathrm{sign}(\nabla_{{x}} \mathfrak{g}_{\epsilon, \sigma}^{x^\star}(x^\star, y^\star)^\top v)$.  Then, since $\mathfrak{g}_{\epsilon, \sigma}^{x^\star}$ has $\frac{2b}{\sigma^7}$-Lipschitz Hessian (see Remark \ref{Rem_Lipschitz_convolution}), for every $t \geq 0$ we have
\be \label{eq:c22}
&\mathfrak{g}_{\epsilon, \sigma}^{x^\star}(x^\star - t \hat{\mathsf{a}} v, y^\star) - \mathfrak{g}_{\epsilon, \sigma}^{x^\star}(x^\star, y^\star)\\
&\leq \frac{1}{2} t^2 v^\top \left(\nabla^2_{{x}} \mathfrak{g}_{\epsilon, \sigma}^{x^\star}(x^\star, y^\star) + t \frac{2b}{\sigma^7} I_d\right) v
\stackrel{\textrm{Eq.} \ref{eq:c20}}{\leq}  - \frac{1}{2}  t^2 \left(\sqrt{\epsilon} - t \frac{2b}{\sigma^7}\right).
\ee
\noindent
Consider the value $t=  \frac{\sigma^{7} \sqrt{\epsilon}}{20 b}$.  Then Inequality \eqref{eq:c22} implies that
\be \label{eq:c12}
\mathfrak{g}_{\epsilon, \sigma}^{x^\star}(x^\star - t \hat{\mathsf{a}}v, y^\star) - \mathfrak{g}_{\epsilon, \sigma}^{x^\star}(x^\star, y^\star) &\leq -\frac{\nicefrac{9}{2}}{4000} \frac{\sigma^{14} \epsilon^{1.5}}{b^2}.
\ee

\noindent
But we also have from our assumption (Inequality \eqref{eq:c10}) that
\be \label{eq:c13}
v^\top \nabla^2_{{x}} \mathfrak{h}_{\epsilon, \sigma}^{x^\star}(x^\star, y^\star) v \geq - \frac{\sqrt{\epsilon}}{5}.
\ee
 Then, since $\mathfrak{h}_{\epsilon, \sigma}^{x^\star}$ has $\frac{2b}{\sigma^7}$-Lipschitz Hessian (see Remark \ref{Rem_Lipschitz_convolution}), for $t= \frac{\sigma^{7} \sqrt{\epsilon}}{20 b}$ we have

\be \label{eq:c14}
&\mathfrak{h}_{\epsilon, \sigma}^{x^\star}(x^\star - t \hat{\mathsf{a}} v, y^\star) - \mathfrak{h}_{\epsilon, \sigma}^{x^\star}(x^\star, y^\star)\\
&\geq \frac{t^2}{2} v^\top \left(\nabla^2_{{x}} \mathfrak{h}_{\epsilon, \sigma}^{x^\star}(x^\star, y^\star) - t\frac{2b}{\sigma^7} I_d\right) v -t \hat{\mathsf{a}} \nabla_{{x}} \mathfrak{h}_{\epsilon, \sigma}^{x^\star}(x^\star, y^\star )^\top v \\
&\stackrel{\textrm{Eq.} \ref{eq:c13}, \ref{eq:c23}}{\geq}  - \frac{t^2 }{2} \left(\frac{1}{5}\sqrt{\epsilon} + t\frac{2b}{\sigma^7}\right) -  \frac{1}{8000}  \frac{\epsilon^2 \sigma^7}{ b} t
\geq -  \frac{2 \sigma^{14} \epsilon^{1.5}}{4000 \, b^2}.
\ee

\noindent Combining Inequalities \eqref{eq:c6},  \eqref{eq:c12} and \eqref{eq:c14}, we get that

\be
&\mathfrak{g}_{\epsilon, \sigma}^{x^\star}(x^\star - t \hat{\mathsf{a}}v, y^\star) - \mathfrak{h}_{\epsilon, \sigma}^{x^\star}(x^\star - t \hat{\mathsf{a}}v, y^\star)\\ &\stackrel{\textrm{Eq.} \ref{eq:c6},\ref{eq:c12},\ref{eq:c14}}{\leq}   \mathfrak{g}_{\epsilon, \sigma}^{x^\star}(x^\star, y^\star)  - \mathfrak{h}_{\epsilon, \sigma}^{x^\star}(x^\star, y^\star)  -\frac{5}{2}\times \frac{\sigma^{14} \epsilon^{1.5}}{4000 \, b^2}.\\
& \stackrel{\textrm{Eq.} \ref{eq:c6}}{\leq}  -\frac{3}{2}\times \frac{\sigma^{14} \epsilon^{1.5}}{4000 \, b^2},
\ee
which implies that
\be \label{eq:c15}
\mathfrak{g}_{\epsilon, \sigma}^{x^\star}(x^\star - t \hat{\mathsf{a}}v, y^\star) < \mathfrak{h}_{\epsilon, \sigma}^{x^\star}(x^\star - t \hat{\mathsf{a}}v, y^\star).
\ee

\noindent Now, we also have that 
\be \label{eq:c16}
&\mathbbm{h}^{x^\star}_{\epsilon}(x,y^\star) = \min(h_{\epsilon}(x, y^\star), h_{\epsilon}(x^\star, y^\star))\\
&\stackrel{\textrm{Eq.} \ref{eq:c11}}{\leq}  \min(g_{\epsilon}(x, y^\star), h_{\epsilon}(x^\star, y^\star))
\stackrel{\textrm{Eq.} \ref{eq:c7}}{=}  \min(g_{\epsilon}(x, y^\star), g_{\epsilon}(x^\star, y^\star))\\
&= \mathbbm{g}^{x^\star}_{\epsilon}(x,y^\star)  \quad \forall x \in \mathbb{R}^d.
\ee
and hence that
\be \label{eq:c17}
&\mathfrak{h}^{x^\star}_{\epsilon}(x,y^\star) = \mathbb{E}[\mathbbm{h}^{x^\star}_{\epsilon}(x+\sigma \zeta ,y^\star) ]\\
&\leq  \mathbb{E}[\mathbbm{g}^{x^\star}_{\epsilon}(x+\sigma \zeta ,y^\star) ]= \mathfrak{g}^{x^\star}_{\epsilon}(x,y^\star)  \qquad \forall x \in \mathbb{R}^d.
\ee

\noindent Since inequality \eqref{eq:c17} contradicts Inequality \eqref{eq:c15}, our Assumption (Inequality \eqref{eq:c4}) must be false.  Therefore we have
\be
\lambda_{\mathrm{min}}(\nabla^2_{{x}} \mathfrak{g}_{\epsilon, \sigma}^{x^\star}(x^\star, y^\star)) \geq - \sqrt{\epsilon},
\ee
which completes the proof of the second-order condition \eqref{eq:c3}.
\end{proof}


\begin{proposition} \label{Lemma_Greedypath}
The path consisting of the line segments  $[\mathsf{y}^{\aell}, \mathsf{y}^{\aell+1}]$ formed by the points $\mathsf{y}^{\aell}$ computed by Algorithm \ref{alg:InnerMaxLoop} is a $\frac{1}{1+\delta}\epsilon'$-greedy path.
\end{proposition}

\begin{proof}
We have the following continuous unit-velocity parametrized path $\phi_t$:
\be
\phi_t = \mathsf{y}^{\aell}+ t v_{\aell} \qquad t\in [t_\aell, t_{\aell+1}], \qquad \aell \in [\aell_{\mathrm{max}} -1],
\ee
where $v_{\aell} := \frac{\mathsf{y}^{\aell+1} - \mathsf{y}^{\aell}}{\|\mathsf{y}^{\aell+1} - \mathsf{y}^{\aell}\|}$,   $t_\aell := \sum_{s=1}^{\aell-1} \|\mathsf{y}^{\aell+1} - \mathsf{y}^{\aell}\|$,  and $\aell_{\mathrm{max}}$ is the number of iterations of the While loop of Algorithm \ref{alg:InnerMaxLoop}

First, we consider indices $\aell$ for which $\| \nabla_yf(\mathsf{x}, \mathsf{y}^{\aell})\| > \epsilon'$ in Line \ref{NoislessSGD_Start} of Algorithm \ref{alg:InnerMaxLoop}.  In this case, $v_{\aell}= \frac{\nabla_yf(\mathsf{x},\mathsf{y}^{\aell})}{\|\nabla_yf(\mathsf{x},\mathsf{y}^{\aell})\|}$ and we have
\be
&\frac{\mathrm{d}}{\mathrm{d}t} f(\mathsf{x}, \phi_t) \geq  [\nabla_yf(\mathsf{x}, \mathsf{y}^{\aell}) -L_2 \|\mathsf{y}^{\aell+1} - \mathsf{y}^{\aell}\| u  ]^\top v_{\aell}\\
&=   \left[\nabla_yf(\mathsf{x}, \mathsf{y}^{\aell}) -L_2 \mu_1\|\nabla_yf(\mathsf{x},\mathsf{y}^{\aell})\| u \right]^\top \frac{\nabla_yf(\mathsf{x},\mathsf{y}^{\aell})}{\|\nabla_yf(\mathsf{x},\mathsf{y}^{\aell})\|}\\
&=   \left[\nabla_yf(\mathsf{x}, \mathsf{y}^{\aell}) - L_2 \mu_1 \|\nabla_yf(\mathsf{x},\mathsf{y}^{\aell})\| u \right]^\top \frac{\nabla_yf(\mathsf{x},\mathsf{y}^{\aell})}{\|\nabla_yf(\mathsf{x},\mathsf{y}^{\aell})\|}\\
&=   \|\nabla_yf(\mathsf{x}, \mathsf{y}^{\aell})\| - (L_2 \mu_1 \|\nabla_yf(\mathsf{x},\mathsf{y}^{\aell})\| u)^\top \frac{\nabla_yf(\mathsf{x},\mathsf{y}^{\aell})}{\|\nabla_yf(\mathsf{x},\mathsf{y}^{\aell})\|}\\
&\geq   \|\nabla_yf(\mathsf{x}, \mathsf{y}^{\aell})\| -  L_2 \mu_1  \|\nabla_yf(\mathsf{x},\mathsf{y}^{\aell})\|\\
&=  (1 -  L_2 \mu_1)  \|\nabla_yf(\mathsf{x}, \mathsf{y}^{\aell})\|\\
&\geq  {\frac{1}{1+\delta} \epsilon' \qquad \qquad \qquad \qquad \qquad \forall t\in [t_\aell, t_{\aell+1})},
\ee
for some unit vectors $u, w$, since $\mu_1 \leq (1- \frac{1}{1+\delta}) \frac{1}{2 L_2}$.

Next, we consider indices $\aell$ for which $\| \nabla_yf(\mathsf{x}, \mathsf{y}^{\aell})\| \leq \epsilon'$ in Line \ref{NoislessSGD_Start} of Algorithm \ref{alg:InnerMaxLoop}.  Since $\mu_3 \leq (1-\frac{1}{1+\delta}) \frac{\sqrt{\epsilon}}{4\sqrt{L_{}}} \leq (1-\frac{1}{1+\delta}) \frac{\sqrt{\epsilon'}}{2\sqrt{L_{}}}$, we have that 
\be 
\frac{\mathrm{d}^2}{\mathrm{d}t^2} f(\mathsf{x}, \phi_t) &\geq v^\top (\nabla_y ^2 f(\mathsf{x}, \mathsf{y}^{\aell}) - \mu_3 L_{} I_d) v\\
  &\geq \sqrt{L_{} \epsilon'} - \mu_3 L_{}
  \geq \sqrt{L_{} (1+\delta)^{-1} \epsilon'}  \qquad \qquad \qquad \forall t\in [t_\aell, t_{\aell+1}),
 \ee
for some unit vector $v$.
\end{proof}

\bigskip
\subsubsection{Properties of $g_\epsilon$ and $h_\epsilon$}

 \begin{proposition}[{Greedy max lower bound}]\label{Prop_lower_bound}
  \be \label{eq:z2}
  h_{\epsilon^{\circ}} (x,y) \leq g_{\epsilon^{\circ}}(x,y), \qquad \forall x,y \in \mathbb{R}^d,  \forall \epsilon^{\circ} > 0.
  \ee
 \end{proposition}

 \begin{proof} Recall that  $h_{\epsilon^{\circ}} (x,y):= f(x,\mathcal{Y})$, where $\mathcal{Y} \leftarrow \mathsf{y}_{\mathrm{LocalMax}}$ is the output of Algorithm \ref{alg:InnerMaxLoop} with inputs $\mathsf{x} \leftarrow x$, $\mathsf{y}^0 \leftarrow y$, and $\epsilon' \leftarrow (1+\delta)\epsilon^{\circ}$.
 
  By Proposition \ref{Lemma_Greedypath}, the path traced by Algorithm \ref{alg:InnerMaxLoop} is a $\epsilon^{\circ}$-greedy path with starting point $y$.  Recall that $ g_{\epsilon^{\circ}}(x,y)$ is the supremum of the value of $f$ at the endpoints of all  $\epsilon^{\circ}$-greedy paths which seek to maximize $f(x, \cdot)$ from the starting point $y$.  Therefore, we have that
 \begin{equation*}
  h_{\epsilon^{\circ}} (x,y) \leq g_{\epsilon^{\circ}}(x,y), \qquad \forall x,y \in \mathbb{R}^d,  \forall \epsilon^{\circ} > 0.  \qedhere
  \end{equation*}
 \end{proof}
 
 \begin{proposition}[{Fixed point property}] \label{Prop_fixed_point} 
 Recall that $\epsilon_i = \epsilon_0 (1+\delta)^{2i}$, and consider the points $(x_i, y_i)$ generated at each iteration $i$ of the While loop in Algorithm \ref{alg:LocalMin-max}.  Then
\be \label{eq:z4}
   h_{\epsilon_i} (x_i,y_i) = g_{\epsilon_i}(x_i,y_i) = f(x_i, y_i), \qquad \forall i \in \mathbb{N}.
\ee
 \end{proposition}
 \begin{proof}

By Lemma \ref{lemma:Player2} we have that 
\begin{equation}  \label{eq:z3}
\|\nabla_y f(x_i, y_i)\| \leq  \epsilon_i  \quad \textrm{ and } \quad \lambda_{\mathrm{max}}( \nabla^2_{{y}} f(x_i, y_i)) \leq \sqrt{L_{} \epsilon_i},
\end{equation}
since $y_{\ai}$ is generated by Algorithm \ref{alg:InnerMaxLoop} with inputs $\mathsf{x} \leftarrow x_{\ai}$, $\mathsf{y}^0 \leftarrow y_{\ai-1}$, and $\epsilon' \leq \epsilon_{\ai-1}(1+\delta)$.

 Inequality \eqref{eq:z3} implies that there is only one $\epsilon_i$-greedy path which seeks to maximize $f(x_{\ai}, \cdot)$ with starting point $y_i$, namely, the path consisting of the single point $y_{\ai}$.  Therefore, 
  \begin{equation*}
  h_{\epsilon_i} (x_i,y_i) = g_{\epsilon_i}(x_i,y_i) = f(x_i, y_i), \qquad \forall i \in \mathbb{N},
  \end{equation*}
  where the first equality holds since, by the proof of Prop \ref{Prop_lower_bound}, $h_{\epsilon_i} (x_i,y_i)$ is the value of $f(x_{\ai}, \cdot)$ at the endpoint of an $\epsilon_i$-greedy path, initialized at $y_i$, which seeks to maximize $f(x_{\ai}, \cdot)$.
  \end{proof}

\begin{proposition}[{Stochastic gradient}] \label{lemma_SG}
For any $\epsilon^{\circ}>0$, $ x,y, \hat{x} \in \mathbb{R}^d$, define
$\hat{\Gamma}^{\hat{x}}_{\epsilon^{\circ}} (x,y) :=  [\mathbbm{h}_{\epsilon^{\circ}}^{\hat{x}} (x + \sigma u, y) - c]\frac{1}{\sigma} u,$
where  $u \sim N(0,I_d)$, for some $c\geq 0$ independent of $u$ with $|c| \leq b$.
Then
\be \label{eq:e1}
\mathbb{E}[\hat{\Gamma}^{\hat{x}}_{\epsilon^{\circ}} (x,y)] = \nabla_{{x}} \mathfrak{h}_{\epsilon^\circ, \sigma}^{\hat{x}}(x,y).
\ee
and, for every $t \geq 0$, we have
\be \label{eq:e3}
\mathbb{P}(\|\hat{\Gamma}^{\hat{x}}_{\epsilon^{\circ}} (x,y) - \nabla_{{x}} \mathfrak{h}_{\epsilon^\circ, \sigma}^{\hat{x}}(x,y)\| \geq t) \leq 2\exp\left(-\frac{t^2}{8\left(\frac{2b\sqrt{d}}{\sigma}\right)^2}\right).
\ee
\end{proposition}

\begin{proof}
First, note that, since $|f(x,y)|$ is bounded by $b$, we have 
\be \label{eq:e2}
&|\mathbbm{h}^{\hat{x}}_{\epsilon^{\circ}}(x,y)| = |\min (h_{\epsilon^{\circ}}(x,y), h_{\epsilon^{\circ}}(\hat{x},y))|\\
&= |f(x,\mathcal{Y})|
 \leq 2b \qquad x,y \in \mathbb{R}^d,
\ee  
 where $\mathcal{Y} \leftarrow \mathsf{y}_{\mathrm{LocalMax}}$ is the output of Algorithm \ref{alg:InnerMaxLoop} with inputs $\mathsf{x} \leftarrow x$ (or $\mathsf{x} \leftarrow \hat{x}$), $\mathsf{y}^0 \leftarrow y$, and $\epsilon' \leftarrow (1+\delta)\epsilon^{\circ}$.
Therefore,
\be
\mathbb{E}[\hat{\Gamma}^{\hat{x}}_{\epsilon^{\circ}} (x,y)] &= \mathbb{E}[\mathbbm{h}^{\hat{x}}_{\epsilon^{\circ}} (x + \sigma u, y) - c]\sigma^{-1} u\\
&=\mathbb{E}[\mathbbm{h}^{\hat{x}}_{\epsilon^{\circ}} (x + \sigma u, y) \sigma^{-1} u] - \sigma^{-1}\mathbb{E}[c] \times \mathbb{E}[u]\\
&= \mathbb{E}[\mathbbm{h}^{\hat{x}}_{\epsilon^{\circ}} (x + \sigma u, y) \sigma^{-1} u] - 0\\
&=\nabla_{{x}} \mathfrak{h}_{\epsilon^\circ, \sigma}^{\hat{x}}(x,y).
\ee
where the last equality follows from Lemma 7 in \cite{zhang2017hitting}, since $\mathbbm{h}_{\epsilon^{\circ}}$ is uniformly bounded by Inequality \eqref{eq:e2}.

Next, we prove Inequality \eqref{eq:e3}:
\be \label{eq:e4}
\|\hat{\Gamma}^{\hat{x}}_{\epsilon^{\circ}} (x,y)& - \nabla_{{x}} \mathfrak{h}_{\epsilon^\circ, \sigma}^{\hat{x}}(x,y)\|\\
&\leq  \| [\mathbbm{h}_{\epsilon^{\circ}}^{\hat{x}} (x + \sigma u, y) - c]\sigma^{-1} u \| + \|\nabla_{{x}} \mathfrak{h}_{\epsilon^\circ, \sigma}^{\hat{x}}(x,y)\| \\
&\leq  \| [\mathbbm{h}_{\epsilon^{\circ}}^{\hat{x}} (x + \sigma u, y) - c]\sigma^{-1} u \| + \|\mathbb{E}[\mathbbm{h}_{\epsilon^{\circ}}^{\hat{x}} (x + \sigma u, y) \sigma^{-1} u]\| \\
&\leq  \| 2b \sigma^{-1} u \| + 2b\sigma^{-1}\sqrt{d},
\ee
since $|\mathbbm{h}_{\epsilon^{\circ}}^{\hat{x}}| \leq b$ and $|c|\leq b$.  Therefore, since $u \sim N(0,I_d)$ is a gaussian random vector, Inequality \eqref{eq:e4} implies that
\begin{align} \label{eq:e5}
\mathbb{P}(\|\hat{\Gamma}^{\hat{x}}_{\epsilon^{\circ}} (x,y) - \nabla_{{x}} \mathfrak{h}_{\epsilon^\circ, \sigma}^{\hat{x}}(x,y)\| \geq t)
&\leq \mathbb{P}\left(\| 2b \sigma^{-1} u \| + 2b\sigma^{-1}\sqrt{d} \geq t\right)\\
 &\leq 2\exp\left(-\frac{t^2}{8\left(\frac{2b\sqrt{d}}{\sigma}\right)^2}\right)  \qquad \forall t \geq 0,
\end{align}
where the last inequality holds by a standard concentration inequality for Gaussians \cite{hanson1971bound}.
\end{proof}

\begin{proposition}[{Low-magnitude stochastic gradient}] \label{prop_SGVariance}

With probability at least $1-\omega$ the stochastic gradient $\mathcal{H}_{\epsilon_{\ai^\star}}$ at the outputs $(x^\star, y^\star)$ of Algorithm \ref{alg:LocalMin-max} satisfies
\be
\mathbb{E}\left[\left\|\mathcal{H}^{x^\star}_{\epsilon_{\ai^\star}}(x^\star, y^\star)\right\| \, \bigg | \, (x^\star, y^\star) \right] \leq 10 b \gamma_1 \sqrt{d}\sigma^{-1} \log\frac{2}{\gamma_1}.
\ee
\end{proposition}

\begin{proof}
Suppose that for any $\ai \in [\ai_{\max}]$ we have that
\be \label{eq:f11}
\mathbb{E}\left[\|\mathcal{H}_{\epsilon_{\ai}}^{x_{\ai}}(x_{\ai}, y_{\ai})\| \big | (x_{\ai}, y_{\ai}) \right] > 10 b \gamma_1 \sqrt{d}\sigma^{-1} \log\frac{2}{\gamma_1}.
\ee

\noindent Then, 
\be
  {}& 10 b \gamma_1 \frac{\sqrt{d}}{\sigma} \log\frac{2}{\gamma_1} < \mathbb{E}\left[\|\mathcal{H}_{\epsilon_{\ai}}^{x_{\ai}}(x_{\ai}, y_{\ai})\| \big | (x_{\ai}, y_{\ai}) \right] \\
 &= \mathbb{E}[  \| \zeta \sigma^{-1}(\mathbbm{h}_{\epsilon_{\ai}}^{x_{\ai}}(x_{\ai}  + \sigma \zeta,y_{\ai}) - \mathbbm{h}_{\epsilon_{\ai}}^{x_{\ai}}(x_{\ai} ,y_{\ai}))  \| \, \, \big | \,\, (x_{\ai}, y_{\ai}) ] \\
 &= \mathbb{E}\left[ \| \zeta \sigma^{-1}(\min(h_{\epsilon_{\ai}}(x_{\ai}  + \sigma \zeta,y_{\ai}), h_{\epsilon_{\ai}} (x_{\ai}, y_{\ai})) - \min(h_{\epsilon_{\ai}}(x_{\ai} ,y_{\ai}), h_{\epsilon_{\ai}}(x_{\ai} ,y_{\ai})))  \|  \, \, \big | \, \,  (x_{\ai}, y_{\ai}) \right]\\
  &= \mathbb{E}\left[ \| \zeta \sigma^{-1} (\min(h_{\epsilon_{\ai}}(x_{\ai}  + \sigma \zeta,y_{\ai}), h_{\epsilon_{\ai}} (x_{\ai}, y_{\ai})) - h_{\epsilon_{\ai}}(x_{\ai} ,y_{\ai}))  \| \,  \big | \,  (x_{\ai}, y_{\ai}) \right]\\
    &= \mathbb{E}\left[ \left\| \zeta \sigma^{-1}  \right\| \times |\min(h_{\epsilon_{\ai}}(x_{\ai}  + \sigma \zeta,y_{\ai}), h_{\epsilon_{\ai}} (x_{\ai}, y_{\ai})) - h_{\epsilon_{\ai}}(x_{\ai} ,y_{\ai})) | \, \, \,  \big | \, \, \, (x_{\ai}, y_{\ai}) \right]\\
     &\leq \mathbb{E}\left[  \sqrt{d} \sigma^{-1} \log\frac{2}{\gamma_1}\times | \min(h_{\epsilon_{\ai}}(x_{\ai}  + \sigma \zeta,y_{\ai}), h_{\epsilon_{\ai}} (x_{\ai}, y_{\ai})) - h_{\epsilon_{\ai}}(x_{\ai} ,y_{\ai})| \, \, \,  \big | \,  \, \, (x_{\ai}, y_{\ai}) \right]\\
     &\qquad \qquad \qquad +\mathbb{P}\left(\|\zeta\| \geq  \sqrt{d} \log\frac{2}{\gamma_1} \right) \times \mathbb{E}\left[ \| \zeta \sigma^{-1} \| \times 2b  \, \, \,  \big | \, \,  \,  \|\zeta\| \geq  \sqrt{d} \log\frac{2}{\gamma_1} \right]\\
     &\leq \mathbb{E}\left[  \sqrt{d} \sigma^{-1} \log\frac{2}{\gamma_1}\times |\min(h_{\epsilon_{\ai}}(x_{\ai}  + \sigma \zeta,y_{\ai}), h_{\epsilon_{\ai}} (x_{\ai}, y_{\ai})) - h_{\epsilon_{\ai}}(x_{\ai} ,y_{\ai})| \, \, \, \big | \, \, \, (x_{\ai}, y_{\ai}) \right]\\
        &\qquad \qquad \qquad + 2 \frac{\gamma_1}{2} \times \mathbb{E}\left[ \| \zeta \sigma^{-1} \| b  \, \, \, \big | \, \, \,  \|\zeta\| \geq  \sqrt{d} \log\frac{2}{\gamma_1} \right]\\
     &\leq \sqrt{d} \sigma^{-1} \log\frac{2}{\gamma_1} \mathbb{E}[ | (\min(h_{\epsilon_{\ai}}(x_{\ai}  + \sigma \zeta,y_{\ai}), h_{\epsilon_{\ai}} (x_{\ai}, y_{\ai})) - h_{\epsilon_{\ai}}(x_{\ai} ,y_{\ai})) |  \, \, \, \big | \, \, \, (x_{\ai}, y_{\ai}) ] + 2b \gamma_1 \sigma^{-1} \sqrt{d} \log\frac{2}{\gamma_1}.
 \ee
 where the last two inequalities hold by a standard concentration bound for Gaussians \cite{hanson1971bound}. 

\noindent Thus,
 \be
\mathbb{E}\left[ | (\min(h_{\epsilon_{\ai}}(x_{\ai}  + \sigma \zeta,y_{\ai}), h_{\epsilon_{\ai}} (x_{\ai}, y_{\ai})) - h_{\epsilon_{\ai}}(x_{\ai} ,y_{\ai})) |\right] >  8 b \gamma_1.
 \ee

\noindent Since $|h_{\epsilon_{\ai}}|$ is uniformly bounded by $b$, this implies that
 \be
  \mathbb{P}\left[ | (\min(h_{\epsilon_{\ai}}(x_{\ai}  + \sigma \zeta,y_{\ai}), h_{\epsilon_{\ai}} (x_{\ai}, y_{\ai})) - h_{\epsilon_{\ai}}(x_{\ai} ,y_{\ai})) | > 2b\gamma_1\right] \geq 2 \gamma_1,
  \ee
and hence since $b \geq 1$, that
 \be  \label{eq:f10}
 \mathbb{P}[ | (\min(h_{\epsilon_{\ai}}(x_{\ai}  &+ \sigma \zeta,y_{\ai}), h_{\epsilon_{\ai}} (x_{\ai}, y_{\ai})) - h_{\epsilon_{\ai}}(x_{\ai} ,y_{\ai})) | > 2\gamma_1] \geq 2 \gamma_1,
\ee
 for any $\ai \in [\ai_{\mathrm{max}} ]$ for which Inequality \eqref{eq:f11} holds.

For any $\ai \in [\ai_{\mathrm{max}} ]$, let  $E_{\ai}$ be the ``bad" event that we have both 
\be  \label{eq:f9}
\mathbb{E}[\|\mathcal{H}_{\epsilon_{\ai}}^{x_{\ai}} (x_{\ai}, y_{\ai})\| \, \, \, \big | \, \, \, (x_i, y_i)] > 10 b \gamma_1 \sqrt{d} \sigma^{-1} \log\frac{2b}{\gamma_1},
\ee
and
\be \label{eq:f6}
\left| \min(h_{\epsilon_{\ai}}(x_{\ai} + \sigma \zeta_{\ai \aj}, y_{\ai}), h_{\epsilon_{\ai}}(x_{\ai}, y_{\ai}) ) - h_{\epsilon_{\ai}}(x_{\ai}, y_{\ai}) \right| \leq 2\gamma_1,
\ee
for every $j \in [\mathcal{I}_3]$.
Then Inequality \eqref{eq:f10} implies that 
\be \label{eq:f7}
\mathbb{P}(E_i) \leq (1 - \gamma_1)^{\mathcal{I}_3} \leq \frac{\omega}{i_{\max}} \qquad \forall i\in [i_{\max}].
\ee
since $\mathcal{I}_3 \geq \frac{1}{\gamma_1} \log(\frac{i_{\max}}{\omega})$, where $i_{\max} = O(\frac{b}{\gamma_1})$ is an upper bound given in the proof of Lemma \ref{lemma:RunningTime} on the number of iterations of the While loop in Algorithm \ref{alg:LocalMin-max}.

Therefore,
\be \label{eq:f8}
{\textstyle \mathbb{P}(\bigcup_{i=1}^{i_{\max}} E_i) \leq \sum_{i=1}^{i_{\max}}  \mathbb{P}(E_i) \stackrel{\textrm{Eq.} \ref{eq:f7}}{\leq} i_{\max} \times \frac{\omega}{i_{\max}}  = \omega.}
\ee

\noindent
Now, whenever the Algorithm outputs a point $(x_{\ai^\star}, y_{\ai^\star}) = (x^\star, y^\star)$, it first checks that the inequality in Line \ref{HillClimbing} does not hold for the point $(x^\star + \sigma \zeta_{\ai^\star \aj},y^\star)$ for the random vector $\zeta_{\ai^\star \aj} \sim N(0,I_d)$, and it repeats this check $\mathcal{I}_3$ times before stopping.  In other words, we have that
\be \label{eq:f1}
f(x^\star + \sigma \zeta_{\ai^\star \aj}, \mathcal{Y}_{\aj}) > f(x^\star, y^\star) - \gamma_1 \qquad \forall j \in [\mathcal{I}_3],
\ee
for a sequence of independent random vectors $\zeta_{\ai^\star 1}, \ldots, \zeta_{\ai^\star \mathcal{I}_3} \sim N(0,I_d)$.  Here we denote by $\mathcal{Y}_{\aj}$ the output of Algorithm \ref{alg:InnerMaxLoop} for inputs  $\mathsf{x} \leftarrow x^\star + \sigma \zeta_{\ai^\star \aj}$,  $\mathsf{y}^0 \leftarrow y^\star$, and $\epsilon' \leftarrow \epsilon_{\ai^\star}(1+\delta)$. 

Inequality \eqref{eq:f1} and Equation \eqref{eq:z1} together imply that
\be \label{eq:f3}
h_{\epsilon_{\ai^\star}}(x^\star + \sigma \zeta_{\ai^\star \aj}, y^\star) > f(x^\star, y^\star) - \gamma_1  \qquad \forall \aj \in [\mathcal{I}_3].
\ee
Therefore, by Proposition \ref{Prop_fixed_point},  Inequality \eqref{eq:f3} implies that
\be
h_{\epsilon_{\ai^\star}}(x^\star + \sigma \zeta_{\ai^\star \aj}, y^\star) > h_{\epsilon_{\ai^\star}}(x^\star, y^\star) - \gamma_1  \qquad \forall \aj \in [\mathcal{I}_3],
\ee
and hence that
\be  \label{eq:f4}
h_{\epsilon_{\ai^\star}}(x^\star + \sigma \zeta_{\ai^\star \aj}, y^\star) - h_{\epsilon_{\ai^\star}}(x^\star, y^\star) > - \gamma_1  \qquad \forall \aj \in [\mathcal{I}_3].
\ee

\noindent Inequality \eqref{eq:f4} then implies that, for every $\aj \in [\mathcal{I}_3]$,
\begin{equation} \label{eq:f5}
\left| \min(h_{\epsilon_{\ai^\star}}(x^\star + \sigma \zeta_{\ai^\star \aj}, y^\star), h_{\epsilon_{\ai^\star}}(x^\star, y^\star) ) - h_{\epsilon_{\ai^\star}}(x^\star, y^\star) \right| <  \gamma_1.
\end{equation}

\noindent
Therefore, Inequalities \eqref{eq:f8} and \eqref{eq:f5}, together with our definition of the ``bad" events $E_i$ (definitions \eqref{eq:f6} and \eqref{eq:f9}) together imply that
\be
\mathbb{E}\left [\|\mathcal{H}^{x^\star}_{\epsilon_{\ai^\star}}(x^\star, y^\star)\| \big | (x^\star, y^\star) \right] \leq 10 b \gamma_1 \frac{\sqrt{d}}{\sigma} \log\frac{2b}{\gamma_1},
\ee
with probability at least $1-\omega$.
\end{proof}

\bigskip
\subsubsection{Showing $x^\star$ is an approximate local minimum of $h_\eps(\cdot, y^\star)$}

\begin{proposition}[{First-order stationary condition}] \label{Prop_first_order_x}
\be \label{eq_summary6} 
\mathbb{P}\left(\left\| \nabla_{{x}} \mathfrak{h}_{\epsilon_{\ai^\star}, \sigma}^{x^\star}(x^\star,y^\star)\right\| > 10 b \gamma_1 \frac{\sqrt{d}}{\sigma} \log\frac{2b}{\gamma_1} \right) \leq \omega.
\ee
\end{proposition}

\begin{proof}
By Proposition \ref{prop_SGVariance} we have that, with probability at least $1-\omega$,
\be 
10 b \gamma_1 \frac{\sqrt{d}}{\sigma} \log\frac{2b}{\gamma_1}
&\stackrel{\textrm{Prop.} \ref{prop_SGVariance}}{\geq} \mathbb{E}_{\zeta \sim N(0,I_d)}\left[ \left\| \frac{\zeta}{\sigma}(\mathbbm{h}_{\epsilon_{\ai^\star}}^{x^\star}(x^\star  + \sigma \zeta,y^\star)- \mathbbm{h}_{\epsilon_{\ai^\star}}^{x^\star}(x^\star ,y^\star))\right\| \,\,\, \bigg | \,\,\, (x^\star, y^\star) \right]\\
&\geq  \left\| \mathbb{E}_{\zeta \sim N(0,I_d)}\left[ \frac{\zeta}{\sigma}(\mathbbm{h}_{\epsilon_{\ai^\star}}^{x^\star}(x^\star  + \sigma \zeta,y^\star) - \mathbbm{h}_{\epsilon_{\ai^\star}}^{x^\star}(x^\star ,y^\star)) \,\,\, \bigg | \,\,\,  (x^\star, y^\star) \right] \right\|\\
&\stackrel{\textrm{Prop.} \ref{lemma_SG}}{=} \| \nabla_{{x}} \mathfrak{h}_{\epsilon_{\ai^\star}, \sigma}^{x^\star}(x^\star,y^\star)\|,
\ee
and hence that, with probability at least $1-\omega$,
\begin{equation} \label{eq:g12}
\| \nabla_{{x}} \mathfrak{h}_{\epsilon_{\ai^\star}, \sigma}^{x^\star}(x^\star,y^\star)\| \leq 10 b \gamma_1 \frac{\sqrt{d}}{\sigma} \log\frac{2b}{\gamma_1}.
\end{equation}

\noindent
Therefore, Inequality \eqref{eq:g12} Implies that 
\begin{equation*}
\mathbb{P}\left(\| \nabla_{{x}} \mathfrak{h}_{\epsilon_{\ai^\star}, \sigma}^{x^\star}(x^\star,y^\star)\| > 10 b \gamma_1 \frac{\sqrt{d}}{\sigma} \log\frac{2b}{\gamma_1}\right) \leq \omega. \qedhere
\end{equation*}
\end{proof}

\begin{proposition}[{Second-order stationary condition}] \label{Prop_NoisySGD}
With probability at least $1-2\omega$, we have that
\be \label{eq:g1}
 \lambda_{\mathrm{min}}\left(\nabla^2_{{x}} \mathfrak{h}_{\epsilon_{\ai^\star}, \sigma}^{x^\star}(x^\star, y^\star) \right) \geq - \frac{1}{5}\sqrt{\epsilon_{\ai^\star}}.
 \ee
\end{proposition}

\begin{proof}

In this proof, it will be convenient to write $X_{\ak}$ which appears inside the For loop (Lines \ref{RebootStart}-\ref{RebootEnd} in Algorithm \ref{alg:LocalMin-max}) with an index $j$ indicating the value that $X_{\ak}$ takes during the $j$th For loop.  Specifically, instead of $X_{\ak}$, we will write $X^{\aj}_{\ak}$.  In a simmilar manner, we right $u_{\ai \ak}$ in place of $u$.

From Line \ref{PSGD_start} of Algorithm \ref{alg:LocalMin-max} we have that
\be  \label{eq:g15}
 h^{\ak} = \min(f(X^{\aj}_{\ak-1} + \sigma u_{\ai \ak}, \mathcal{Y}), f(x_{\ai}, y_{\ai})),
\ee
where $\mathcal{Y} \leftarrow \mathsf{y}_{\mathrm{LocalMax}}$ is the output $\mathsf{y}_{\mathrm{LocalMax}}$ of Algorithm \ref{alg:InnerMaxLoop} with inputs $\mathsf{x} \leftarrow X^{\aj}_{\ak-1} + \sigma u_{\ai \ak}$,  $\mathsf{y}^0 \leftarrow y_{\ai}$, $\epsilon' \leftarrow \epsilon_{\ai}(1+\delta)$.

Therefore, by the definition of the function $h$ (Equation \eqref{eq:z1}), we have that
\be  \label{eq:g16}
f(X^{\aj}_{\ak-1} + \sigma u_{\ai \ak}, \mathcal{Y}) = h_{\epsilon_i} (X^{\aj}_{\ak-1} + \sigma u_{\ai \ak}, y_{\ai}).
\ee

\noindent
Moreover, by Proposition \ref{Prop_fixed_point}, we have that
\be  \label{eq:g17}
f(x_{\ai}, y_{\ai}) = h_{\epsilon_i}(x_{\ai}, y_{\ai}).
\ee

\noindent
  Therefore, Equations \eqref{eq:g15},  \eqref{eq:g16} and \eqref{eq:g17}, together imply that
\be  \label{eq:g2}
 h^{\ak} &= \min(f(X^{\aj}_{\ak-1} + \sigma u_{\ai \ak}, \mathcal{Y}),  \, \, f(x_{\ai}, y_{\ai}))\\
 & =  \min(h_{\epsilon_i} (X^{\aj}_{\ak-1} + \sigma u_{\ai \ak}, y_{\ai}),  \, \, h_{\epsilon_i}(x_{\ai}, y_{\ai})).
\ee

\noindent
In Line \ref{PSGD_end} of Algorithm \ref{alg:LocalMin-max} we have that  $\Gamma_{\ak} = (h^{\ak} - h^{\ak-1}) \frac{1}{\sigma} u_{\ai \ak}$.  Therefore, Equation \eqref{eq:g2} implies that
\be \label{eq:g3}
\Gamma_{\ak} &= (h^{\ak} - h^{\ak-1}) \frac{1}{\sigma} u_{\ai \ak}\\
&= [\min(h_{\epsilon_i} (X^{\aj}_{\ak-1} + \sigma u_{\ai \ak}, y_{\ai}), h_{\epsilon_i}(x_{\ai}, y_{\ai})) - h_{\ak-1}]\frac{1}{\sigma} u _{\ai \ak}\\
& = [\min(\mathbbm{h}_{\epsilon_i}^{x_{\ai}} (X^{\aj}_{\ak-1} + \sigma u_{\ai \ak}, y_{\ai})) - h_{\ak-1}]\frac{1}{\sigma} u _{\ai \ak}.
\ee

\noindent Therefore by Equation \eqref{eq:g3} and Proposition \ref{lemma_SG}, and since $h_{\ak-1}$ is independent of $\sigma u_{\ai \ak}$, we have that
\be \label{eq:g4}
\mathbb{E}[\Gamma_{\ak}]= \nabla_{{x}} \mathfrak{h}_{\epsilon_{\ai}, \sigma}^{x_{\ai}}(X^{\aj}_{\ak-1},y_{\ai}),
\ee
and
\begin{equation} \label{eq:g5}
\mathbb{P}(\|\Gamma_{\ak} - \nabla_{{x}} \mathfrak{h}_{\epsilon_{\ai}, \sigma}^{x_{\ai}}(X^{\aj}_{\ak-1},y_{\ai})\| \geq t) \leq 2\exp\left(-\frac{t^2}{8\left(\frac{2b\sqrt{d}}{\sigma}\right)^2}\right)  \qquad \forall t \geq 0.
\end{equation}
In other words, $\Gamma_{\ak}$ is a stochastic gradient for $\mathfrak{h}_{\epsilon_{\ai}, \sigma}^{x_{\ai}}(X^{\aj}_{\ak-1},y_{\ai})$ (Equation \eqref{eq:g4}) which satisfies a concentration property (Inequality \eqref{eq:g5}).

Since $\Gamma_{\ak}$ is a stochastic gradient  with concentration properties for a smooth function, we can apply results from \cite{NonconvexOptimizationForML} which, roughly speaking, say that stochastic gradient descent with added Gaussian noise can escape saddle points in polynomial time.

More specifically, Lemma 25 of \cite{NonconvexOptimizationForML}, together with Equations \eqref{eq:g4} and \eqref{eq:g5}, imply that if at any iteration $i$ of the For loop in Algorithm \ref{alg:LocalMin-max} (Lines \ref{RebootStart}-\ref{RebootEnd}) we have
\be \label{eq:g6}
\|\nabla_{{x}} \mathfrak{h}_{\epsilon_{\ai}, \sigma}^{x_{\ai}}(x_{\ai},y_{\ai})\| \leq  \epsilon_{\ai}/50,
\ee
and
\be \label{eq:g7}
\lambda_{\min}(\nabla_{{x}}^2 \mathfrak{h}_{\epsilon_{\ai}, \sigma}^{x_{\ai}}(x_{\ai},y_{\ai})) \leq -   \sqrt{\epsilon_{\ai}}/10,
\ee
then 
\be \label{eq:g14}
\mathbb{P}(\mathfrak{h}_{\epsilon_{\ai}, \sigma}^{x_{\ai}}( X^{\aj}_{\mathcal{I}_2}, y_{\ai}) - \mathfrak{h}_{\epsilon_{\ai}, \sigma}^{x_{\ai}}(x_{\ai}, y_{\ai}) \leq -\gamma_1) \geq \frac{1}{6},
\ee
as long as the hyper-parameters $\gamma_1, \mathcal{I}_2, \eta, \alpha$ satisfy 
$$\gamma_1 \leq \epsilon^{\nicefrac{3}{2}} \sigma^{15.5} (b^{2.5} \mathsf{c} \log^5( bd\sqrt{\sigma \epsilon}))^{-1},$$ 
$$\mathcal{I}_2 \geq \mathsf{c} \eta^{-1} \epsilon^{-\nicefrac{1}{2}} \log(bd\sqrt{\sigma \epsilon}) ,$$ 
$$\eta \leq \frac{\sigma^9}{b^4 d^2} (b^2 (1+ 10 b d \sigma^{-12} \epsilon^{-2}) \mathsf{c} \log^9( bd\sqrt{\sigma \epsilon}))^{-1},$$  
$$\alpha = \eta \mathsf{c}\log( bd\sqrt{\sigma \epsilon}) \sqrt{1+ b^2 d^2 \sigma^{-2}},$$
  where $\mathsf{c}$ is a universal constant.

For every $i \in [i_{\max}]$, let $\mathcal{E}_i$ be the ``bad" event that we have that both of the following Equations \eqref{eq:g8} and \eqref{eq:g9} hold:
\begin{equation} \label{eq:g8}
\|\nabla_{{x}} \mathfrak{h}_{\epsilon_{\ai}, \sigma}^{x_{\ai}}(x_{\ai},y_{\ai})\| \leq \frac{1}{50}\epsilon_{\ai}, \qquad  \lambda_{\min}(\nabla_{{x}}^2 \mathfrak{h}_{\epsilon_{\ai}, \sigma}^{x_{\ai}}(x_{\ai},y_{\ai})) \leq - \frac{1}{10} \sqrt{\epsilon_{\ai}}
\end{equation}
and
\be \label{eq:g9}
\mathfrak{h}_{\epsilon_{\ai}, \sigma}^{x_{\ai}}( X^{\aj}_{\mathcal{I}_2}, y_{\ai}) - \mathfrak{h}_{\epsilon_{\ai}, \sigma}^{x_{\ai}}(x_{\ai}, y_{\ai}) > -\gamma_1 \qquad \forall j\in \mathcal{I}_4.
\ee

\noindent Then Inequality \eqref{eq:g14} implies that
\be  \label{eq:g10}
\mathbb{P}(\mathcal{E}_i) \leq \left(1-\frac{1}{6}\right)^{\mathcal{I}_4} \leq \frac{\omega}{i_{\max}},
\ee
since $\mathcal{I}_4 \geq 6\log(\frac{i_{\max}}{\omega})$, where $i_{\max} = O(\frac{b}{\gamma_1})$ is an upper bound given in the proof of Lemma \ref{lemma:RunningTime} on the number of iterations of the While loop in Algorithm \ref{alg:LocalMin-max}.

 \noindent
Therefore,
\begin{equation} \label{eq:g11}
{\textstyle \mathbb{P}(\bigcup_{i=1}^{i_{\max}} \mathcal{E}_{\ai}) \leq \sum_{i=1}^{{\ai}_{\max}}  \mathbb{P}(\mathcal{E}_{\ai}) \stackrel{\textrm{Eq.} \ref{eq:f10}}{\leq} {\ai}_{\max} \times \frac{\omega}{{\ai}_{\max}}  = \omega.}
\end{equation}

\noindent Now, since we set $\gamma_1 \leq \frac{\epsilon \sigma}{10^3 b \sqrt{d} \log(\frac{2b}{\gamma_1}) }$,  Proposition \ref{Prop_first_order_x} Implies that 
\be  \label{eq:g13}
\mathbb{P}\left(\| \nabla_{{x}} \mathfrak{h}_{\epsilon_{\ai^\star}, \sigma}^{x^\star}(x^\star,y^\star)\| > \frac{1}{100}\epsilon\right) \leq \omega.
\ee

\noindent
Now, the condition in Line \ref{SuccessStart} of Algorithm \ref{alg:LocalMin-max} implies that Inequality \eqref{eq:g9}  holds for all $j\in \mathcal{I}_4$ when the value of $i$ is $i^\star$.  Therefore, Inequalities \eqref{eq:g13}, and \eqref{eq:g11}, together with the definition of $\mathcal{E}_i$ (\eqref{eq:g8} and \eqref{eq:g9}) imply that
\be \label{eq:g18}
 \lambda_{\mathrm{min}}(\nabla^2_{{x}} \mathfrak{h}_{\epsilon_{\ai^\star}, \sigma}^{x^\star}(x^\star, y^\star)) \geq - \frac{1}{10}\sqrt{\epsilon_{\ai^\star}},
 \ee
 with probability at least $1-2\omega$.
\end{proof}

\subsection{Concluding the proof of Theorem \ref{thm:GreedyMin-max}} \label{concluding_the_proof}

\begin{proof}[Proof of Theorem \ref{thm:GreedyMin-max}]
\textbf{Showing convergence, and bounding the number of oracle calls.}
By Lemma \ref{lemma:RunningTime}, we have that Algorithm \ref{alg:LocalMin-max} terminates and outputs a point $(x^\star, y^\star) \in \mathbb{R}^d$ after  $$O\left(\frac{b}{\gamma_1} \times (\mathcal{I}_2 \mathcal{I}_4 + \mathcal{I}_3) \times \frac{b}{ \mu_1 \mu_3^2 L_{}}\right) = \mathrm{poly}\left(\frac{1}{\epsilon}, d, b, L_{}, \frac{1}{\sigma}\right)$$ gradient, function, and Hessian evaluations.  In particular, if $b, L \geq 1$ and if $\sigma, \epsilon \leq 1$, the number of gradient, function, and Hessian evaluations can be simplified to $O\left(\frac{d^{8} L^2 b^{37}}{\epsilon^{19}\sigma^{132}}\right)$.

\bigskip
\noindent
\textbf{Showing that $x^\star$ is an approximate local minimum for greedy max function.}
By Proposition \ref{Prop_lower_bound}, we have that
\be \label{eq:h7}
h_{\epsilon_{\ai^\star}} (x,y) \leq g_{\epsilon_{\ai^\star}} (x,y)  \qquad \forall x,y\in \mathbb{R}^d.
\ee

\noindent By Proposition \ref{Prop_fixed_point}, we have that
\be \label{eq:h1}
  h_{\epsilon_{\ai^\star}} (x^\star,y^\star) = g_{\epsilon_i}(x^\star, y^\star) = f(x^\star, y^\star).
\ee

\noindent By Proposition \ref{prop_SGVariance} we have, with probability at least $1-\omega$, that
\be \label{eq:h3}
\mathbb{E}\left[\left\|\mathcal{H}^{x^\star}_{\epsilon_{\ai^\star}}(x^\star, y^\star)\right\| \big | (x^\star, y^\star) \right]  \leq  \frac{1}{8000}  \frac{\sigma^{14} \epsilon_{\ai^\star}^{1.5}}{b^2},
\ee
since $\gamma_1 = \frac{\epsilon^{2.1} \sigma^{16.6}}{10^4 (1+b^{3.1}) d^{0.6} \log(b d \sigma \epsilon)}$.

 \noindent
By Proposition \ref{Prop_first_order_x}, with probability at least $1-\omega$, we have that
\be \label{eq:h8}
\left\| \nabla_{{x}} \mathfrak{h}_{\epsilon_{\ai^\star}, \sigma}^{x^\star}(x^\star,y^\star)\right\| \leq \frac{\epsilon_{\ai^\star}^2}{8000}  \frac{\sigma^7}{b},
\ee
since $\gamma_1 = \frac{\epsilon^{2.1} \sigma^{16.6}}{10^4 (1+b^{3.1}) d^{0.6} \log(b d \sigma \epsilon)}$.

 \noindent
By Proposition \ref{Prop_NoisySGD}, with probability at least $1-2\omega$ we have
\be \label{eq:h4}
 \lambda_{\mathrm{min}}(\nabla^2_{{x}} \mathfrak{h}_{\epsilon_{\ai^\star}, \sigma}^{x^\star}(x^\star, y^\star)) \geq - \frac{1}{5}\sqrt{\epsilon_{\ai^\star}}.
 \ee

\noindent Thus, by Lemma \ref{Lemma_SharedLM},  \, \eqref{eq:h7}-\eqref{eq:h4} imply that, w.p. at least $1-4 \omega$,
 \begin{equation} \label{eq:h5}
\big\|\nabla_{{x}} \mathfrak{g}_{\epsilon_{\ai^\star}, \sigma}^{x^\star}(x^\star, y^\star)\big\| \leq \epsilon_{\ai^\star} \quad \textrm{and} \quad \lambda_{\mathrm{min}}(\nabla^2_{{x}} \mathfrak{g}_{\epsilon_{\ai^\star}, \sigma}^{x^\star}(x^\star, y^\star)) \geq - \sqrt{\epsilon_{\ai^\star}}.
\end{equation}

\noindent \textbf{Showing that $y^\star$ is an approximate local maximum for $f(x^\star, \cdot)$.}
We also have, by Lemma \ref{lemma:Player2} that
\begin{equation} \label{eq:h6}
\|\nabla_y f(x^\star, y^\star)\| \leq  \epsilon_{\ai^\star} \, \, \, \,  \textrm{ and }  \, \,  \, \,  \lambda_{\mathrm{max}}( \nabla^2_{{y}} f(x^\star, y^\star)) \leq \sqrt{L_{}  \epsilon_{\ai^\star}}.
\end{equation}

\noindent \textbf{Showing that $(x^\star, y^\star)$ is greedy adversarial equilibrium for $f$.}
 Inequalities \eqref{eq:h5} and \eqref{eq:h6} together imply that, with probability at least $1-4 \omega$, the point $(x^\star, y^\star)$ is an $(\epsilon^\star, \sigma)$-greedy adversarial equilibrium, where $\epsilon^\star = \epsilon_{\ai^\star} \leq \eps$.
\end{proof}

\section*{Acknowledgments}
This research was supported in part by NSF CCF-1908347.

\bibliographystyle{plain}
  \bibliography{GAN}

\newpage
\appendix

\section{Greedy adversarial equilibrium in the strongly convex-strongly concave setting} \label{sec:convex_concave}

\begin{theorem}\label{Thm_strongly_convex_concave}
Suppose that for some $\alpha, L >0$ the function $f: \mathbb{R}^d \times \mathbb{R}^d \rightarrow \mathbb{R}$ is $C^2$-smooth, $\alpha$-strongly convex in $x$, $\alpha$-strongly concave in $y$, and has $L$-Lipschitz gradients and $L$-Lipschitz Hessian. 
 And suppose that $(x^\star, y^\star)$ is an $(\epsilon, \sigma)$-greedy adversarial equilibrium, for any $\epsilon>0$ and any $0< \sigma \leq \frac{\epsilon^3 \alpha}{10d^{\nicefrac{3}{2}} L^3 (\alpha+1) \log(L\frac{\alpha+1}{\alpha})}.$  
 Then the duality gap at $(x^\star, y^\star)$ satisfies
\be
\max_{z \in \mathbb{R}^d} f(x^\star, z) - \min_{w\in \mathbb{R}^d} f(w, y^\star) \leq 26\frac{\epsilon^2}{\alpha}.
\ee

\end{theorem}

\begin{proof}

Since $(x^\star, y^\star)$ is an $(\epsilon, \sigma)$-greedy adversarial equilibrium, $\|\nabla_y f(x^\star, y^\star)\| \leq \epsilon$.
Therefore, since $f(x,y)$ is $\alpha$-strongly concave in $y$, we have that
\be \label{eq_strong_y}
\|y^\star - \mathrm{argmax}_z f(x^\star, z)\| \leq \frac{\epsilon}{\alpha}.
\ee

\noindent
Recall the notation
   \be
   S(x) :=  \mathbb{E}_{\zeta \sim N(0,I_d)}\left [\min(g_\epsilon(x  + \sigma \zeta, y^\star),  g_\epsilon(x^\star, y^\star)) \right].
   \ee
   from Definition \ref{def:local_min-max_formal}.
   
   From the proof of Proposition \ref{Prop_fixed_point} we have that $f(x^\star, y^\star) = g_\epsilon(x^\star, y^\star)$.  Thus,
   
\be
      S(x) =  \mathbb{E}_{\zeta \sim N(0,I_d)} \left[\min(g_\epsilon(x  + \sigma \zeta, y^\star), f(x^\star, y^\star)) \right].
\ee
   
   \noindent
   Let 
   \be
   \hat{S}(x) = \mathbb{E}_{\zeta \sim N(0,I_d)}\left [\min(g_\epsilon(x  + \sigma \zeta, y^\star),  f(x^\star, y^\star)) - \min(f(x  + \sigma \zeta, y^\star),  f(x^\star, y^\star))  \right].
   \ee
    \noindent
   Then, by the formula in Equation \eqref{eq_stochastic_gradient_formula}, we have
\be \label{eq_z1}
\|\nabla_x \hat{S}(x)\| &=  \left \|\mathbb{E}_{\zeta \sim N(0,I_d)} \left [ \frac{\zeta}{\sigma}\left [\min(g_\epsilon(x  + \sigma \zeta, y^\star),  f(x^\star, y^\star)) - \min(f(x  + \sigma \zeta, y^\star),  f(x^\star, y^\star))  \right] \right]\right \|\\
& \leq  \mathbb{E}_{\zeta \sim N(0,I_d)} \left [ \frac{\|\zeta\|}{\sigma} \times \left \| \min(g_\epsilon(x  + \sigma \zeta, y^\star),  f(x^\star, y^\star)) - \min(f(x  + \sigma \zeta, y^\star),  f(x^\star, y^\star))  \right\| \right].
\ee

 \noindent
Suppose now that $\zeta \in \mathbb{R}^d$ and let $\varphi$ be any $\epsilon$-greedy path starting at the point $y^\star$ which seeks to maximize $f(x^\star  + \sigma \zeta, \cdot)$.
Thus, since $f(x^\star  + \sigma \zeta, \cdot)$ increases at rate at least $\epsilon$ along the greedy path, and $f(x^\star  + \sigma \zeta, \cdot)$ is $\alpha$-strongly concave, we must have that the greedy path reaches a point $p$ where
$$\|\nabla_y f(x^\star  + \sigma \zeta, p)\| \leq \epsilon,$$
after the path has traveled a distance of at most $\|p-y^\star\| \leq (\|\nabla_y f(x^\star  + \sigma \zeta, y^\star)\|-\epsilon) \times \frac{1}{\alpha}$.  %

Thus, since the value of $f$ must continuously increase along the greedy path, and $f(x,y)$ is strongly concave in $y$, we have that
\be \label{eq_z2}
{}&\left | \min(g_\epsilon(x  + \sigma \zeta, y^\star),  f(x^\star, y^\star)) - \min(f(x  + \sigma \zeta, y^\star),  f(x^\star, y^\star))  \right |\\
&\leq \left | g_\epsilon(x  + \sigma \zeta, y^\star) - f(x  + \sigma \zeta, y^\star)  \right |\\ 
&\leq (\|\nabla_y f(x^\star  + \sigma \zeta, y^\star)\|-\epsilon) \times \frac{1}{\alpha} \times \|\nabla_y f(x^\star  + \sigma \zeta, y^\star)\|,
\ee
\noindent
Therefore, combining Inequalities \eqref{eq_z1} and \eqref{eq_z2}, we get
\be
{}&\|\nabla_x \hat{S}(x)\| 
 \stackrel{\textrm{Eq.}\ref{eq_z1}, \ref{eq_z2}}{\leq} \mathbb{E}_{\zeta \sim N(0,I_d)} \left [ \frac{\|\zeta\|}{\sigma} \times (\|\nabla_y f(x^\star  + \sigma \zeta, y^\star)\|-\epsilon) \times \frac{1}{\alpha} \times \|\nabla_y f(x^\star  + \sigma \zeta, y^\star)\| \right]\\
&\leq\mathbb{E}_{\zeta \sim N(0,I_d)} \bigg [ \frac{20\sqrt{d} \log\left(d\frac{L(\alpha+1)}{\alpha \epsilon}\right)}{\sigma} \times (\|\nabla_y f(x^\star  + \sigma \zeta, y^\star)\|-\epsilon) \times \frac{1}{\alpha}\\
&\qquad \qquad \qquad  \qquad \qquad \qquad \qquad  \qquad \qquad \qquad \qquad \times \left(\epsilon + L \sigma 20\sqrt{d} \log\left(d\frac{L(\alpha+1)}{\alpha \epsilon}\right)\right)\bigg]\\
&+ \mathbb{E}_{\zeta \sim N(0,I_d)} \left [ \frac{\|\zeta\|}{\sigma} \times (L\sigma\|\zeta\|) \times \frac{1}{\alpha} \times(\epsilon + L\sigma\|\zeta\|) \times \mathbbm{1}\bigg\{\|\zeta\| \geq 20\sqrt{d} \log\left(d\frac{L(\alpha+1)}{\alpha \epsilon}\right) \bigg\} \right ]\\
&
 \stackrel{\textrm{Lemma }\ref{lemma_strong}}{\leq} \frac{20\sqrt{d} \log(d(\frac{L(\alpha+1)}{\alpha \epsilon})}{\sigma} \times \left( \frac{1}{2 \epsilon}L^2  \sigma^2 d  +   \frac{1}{2}\sigma^2 L d \right ) \times \frac{1}{\alpha}\times \left(\epsilon + L \sigma 20\sqrt{d} \log\left(d\frac{L(\alpha+1)}{\alpha \epsilon}\right)\right)\\
&+ \frac{L^2 \sigma + \epsilon}{\alpha}\mathbb{E}_{\zeta \sim N(0,I_d)} \left [ \|\zeta\|^3 \times \mathbbm{1}\bigg\{\|\zeta\| \geq 20\sqrt{d} \log\left(d\frac{L(\alpha+1)}{\alpha \epsilon}\right)\bigg\} \right ]\\
& \leq \frac{\epsilon}{2} + \frac{L^2 \sigma + \epsilon}{\alpha}\mathbb{E}_{\zeta \sim N(0,I_d)} \left [ \|\zeta\|^4 \times \mathbbm{1}\bigg\{\|\zeta\| \geq 20\sqrt{d} \log\left(d\frac{L(\alpha+1)}{\alpha \epsilon}\right)\bigg\} \right ]\\
& \leq \frac{\epsilon}{2} + \int_{(20\sqrt{d} \log(d(\frac{L(\alpha+1)}{\alpha \epsilon}))^4 }^{\infty} \mathbb{P}(\|\zeta\|^4 \geq t) \mathrm{d}t\\
 & = \frac{\epsilon}{2} + \int_{(20\sqrt{d} \log(d(\frac{L(\alpha+1)}{\alpha \epsilon}))^4 }^{\infty} \mathbb{P}(|\zeta\| \geq t^{\nicefrac{1}{4}}) \mathrm{d}t\\
  & \leq \frac{\epsilon}{2} + \int_{(20\sqrt{d} \log(d(\frac{L(\alpha+1)}{\alpha \epsilon}))^4 }^{\infty} e^{-\frac{\sqrt{t} - d}{8}} \mathrm{d}t\\
    & \leq \frac{\epsilon}{2} + e^{\frac{d}{8}} \left( -16(\sqrt{t}+8) e^{-\frac{\sqrt{t}}{8}} \right) \bigg|_{20\sqrt{d} \log\left(d\frac{L(\alpha+1)}{\alpha \epsilon}\right)^4 }^{\infty} \\
   & \leq \epsilon,
\ee
where the third inequality is by  Lemma \ref{lemma_strong}, the fourth inequality holds since  since $\sigma \leq \frac{\epsilon^3 \alpha}{10 d^{\nicefrac{3}{2}} L^3 (\alpha+1) \log(L\frac{\alpha+1}{\alpha \epsilon})}$, and the sixth inequality is a standard concentration bound for Gaussians \cite{hanson1971bound}.

Define $R(x) :=  \mathbb{E}_{\zeta \sim N(0,I_d)}\left [\min(f(x  + \sigma \zeta, y^\star),  f(x^\star, y^\star)) \right]$ for all $x\in \mathbb{R}^d$.
Since $(x^\star, y^\star)$ is an $(\epsilon, \sigma)$-greedy adversarial point, we have $\|\nabla_x S(x)\| \leq \epsilon$.
Then
\be \label{eq_z4}
\| \nabla_x R(x^\star) \| = \| \nabla_x S(x) - \nabla_x \hat{S}(x) \| \leq  \| \nabla_x S(x)\| +  \| \nabla_x \hat{S}(x) \| &\stackrel{\textrm{Eq.}\ref{eq_z3}}{\leq} 2\epsilon.
\ee
 \noindent
  Since $f(\cdot, y^\star)$ is convex,
\be\label{eq_strong_2}
\mathbb{P}_{\zeta \sim N(0,I_d)}(\min(f(x^\star  + \sigma \zeta, y^\star) \geq  f(x^\star, y^\star)) \geq \frac{1}{2}.
\ee
 Thus, since $f$ is also $L$-Lipschitz,
\be
\| \nabla_x R(x^\star) \| \stackrel{\textrm{Eq.}\eqref{eq_strong_2}}{\geq}\frac{1}{2}\|\nabla_x f(x^\star, y^\star)\| - \mathbb{E}_{\zeta \sim N(0,I_d)}[L\sigma \|\zeta\|] =  \frac{1}{2}\|\nabla_x f(x^\star, y^\star)\| - L \sigma \sqrt{d}.
\ee
 \noindent
Thus,
\be \label{eq_strong_1}
\|\nabla_x f(x^\star, y^\star)\| &\leq  2\| \nabla_x R(x^\star) \|  + 2L \sigma \sqrt{d}\\
& \stackrel{\textrm{Eq.}\eqref{eq_z4}}{\leq} 4\epsilon +  2L \sigma \sqrt{d} \\
&\leq 5\epsilon,
\ee
since $\sigma \leq \frac{\epsilon}{2Ld}$.

\noindent
Since $f$ is $\alpha$-strongly convex, Inequality \eqref{eq_strong_1} implies that
\be \label{eq_strong_x}
\|x^\star - \argmin_w f(w, y^\star)\| \leq 5\frac{\epsilon}{\alpha}.
\ee

\noindent
Now, since $f(\cdot, y^\star)$ is $\alpha$-strongly convex, Inequality \eqref{eq_strong_x} implies that
\be \label{eq_z5}
f(x^\star, y^\star) - \min_{w\in \mathbb{R}^d} f(w, y^\star) \leq  25\frac{\epsilon^2}{\alpha}.
\ee
And since $f(x^\star, \cdot)$ is  $\alpha$-strongly concave, 
Inequality \eqref{eq_strong_y} implies that
\be \label{eq_z6}
\max_{z\in \mathbb{R^d}} f(x^\star, z) - f(x^\star, y^\star) \leq \frac{\epsilon^2}{\alpha}.
\ee

\noindent
Combining Inequalities \eqref{eq_z5} and \eqref{eq_z6}, we get
\be \label{eq_z8}
\max_{z \in \mathbb{R}^d} f(x^\star, z) - \min_{w\in \mathbb{R}^d} f(w, y^\star) \leq 26\frac{\epsilon^2}{\alpha}.
\ee

\end{proof}

\begin{lemma} \label{lemma_strong}
Suppose that $f: \mathbb{R}^d \times \mathbb{R}^d \rightarrow \mathbb{R}$ is $C^2$-smooth with $L$-Lipschitz gradient and $L$-Lipschitz Hessian, and that $\|\nabla_y f(x^\star, y^\star)\| \leq \epsilon$ for some $x^\star, y^\star \in \mathbb{R}^d$.  Then
\be
\mathbb{E}_{\zeta \sim N(0,I_d)} &[ \|\nabla_y f(x^\star  + \sigma \zeta, y^\star)\|] \leq \epsilon + \frac{1}{2 \epsilon}L^2  \sigma^2 d  +   \frac{1}{2}\sigma^2 L d.
\ee

\end{lemma}

\begin{proof}
Since $f$ is $C^2$-smooth, we have that

\be \label{eq_z3}
 \mathbb{E}_{\zeta \sim N(0,I_d)} &[ \|\nabla_y f(x^\star  + \sigma \zeta, y^\star)\|]\\
 &= \mathbb{E}_{\zeta \sim N(0,I_d)} \left[ \left \| \nabla_y f(x^\star, y^\star) + \int_{0}^{\sigma\|\zeta\|}  \nabla_{xy}^2   f(x^\star  + t \frac{\zeta}{\|\zeta\|}, y^\star)  \frac{\zeta}{\|\zeta\|}  \mathrm{d}t \right \| \right ]\\
&\leq \mathbb{E}_{\zeta \sim N(0,I_d)} \left[  \int_{0}^{\sigma\|\zeta\|} \left \| \nabla_y f(x^\star, y^\star)+   \nabla_{xy}^2   f(x^\star  + t \frac{\zeta}{\|\zeta\|}, y^\star) \frac{\zeta}{\|\zeta\|}  \right \|\mathrm{d}t \right ]\\
&\leq  \mathbb{E}_{\zeta \sim N(0,I_d)} \left[  \int_{0}^{\sigma\|\zeta\|} \left \| \nabla_y f(x^\star, y^\star) + \nabla_{xy}^2   f(x^\star, y^\star) \frac{\zeta}{\|\zeta\|}  \right \| \mathrm{d}t  \right ]\\
& \qquad \qquad+   \mathbb{E}_{\zeta \sim N(0,I_d)} \left[ \int_{0}^{\sigma\|\zeta\|}   \left \| [\nabla_{xy}^2   f(x^\star  + t \frac{\zeta}{\|\zeta\|}, y^\star)  - \nabla_{xy}^2   f(x^\star, y^\star)] \frac{\zeta}{\|\zeta\|}  \right \|\mathrm{d}t \right ]\\
&\leq   \mathbb{E}_{\zeta \sim N(0,I_d)} \left[ \left \| \nabla_y f(x^\star, y^\star) + \nabla_{xy}^2   f(x^\star, y^\star) \sigma \zeta  \right \|   \right ] + \mathbb{E}_{\zeta \sim N(0,I_d)} \left[ \int_{0}^{\sigma\|\zeta\|}   L \sigma \|\zeta\| \mathrm{d}t \right ]\\
&\leq  \mathbb{E}_{\zeta \sim N(0,I_d)} \left[ \left \| \nabla_y f(x^\star, y^\star) + L  \sigma \zeta  \right \|   \right ]  + \frac{1}{2}\sigma^2 L \mathbb{E}_{\zeta \sim N(0,I_d)} \left[ \|\zeta\|^2   \right ]\\
&=  \mathbb{E}_{\zeta \sim N(0,I_d)} \left[ \left \langle \nabla_y f(x^\star, y^\star) + L  \sigma \zeta, \, \, \, \, \nabla_y f(x^\star, y^\star) + L  \sigma \zeta   \right \rangle^{\nicefrac{1}{2}}   \right ]  + \frac{1}{2}\sigma^2 L d\\
&=  \mathbb{E}_{\zeta \sim N(0,I_d)} \left[ \left( \|\nabla_y f(x^\star, y^\star)\|^2+ 2 \left \langle \nabla_y f(x^\star, y^\star) , \, \, L  \sigma \zeta  \right \rangle, \, \, \, \,+ L^2  \sigma^2 \|\zeta\|^2  \right )^{\nicefrac{1}{2}} \right ]  + \frac{1}{2}\sigma^2 L d\\
&\leq \mathbb{E}_{\zeta \sim N(0,I_d)} \left[ \left( \|\nabla_y f(x^\star, y^\star)\|^2 + L^2  \sigma^2 \|\zeta\|^2  \right )^{\nicefrac{1}{2}} \right ]  + \frac{1}{2}\sigma^2 L d\\
&\leq \mathbb{E}_{\zeta \sim N(0,I_d)} \left[ \left( \epsilon^2 + L^2  \sigma^2 \|\zeta\|^2  \right )^{\nicefrac{1}{2}} \right ]  + \frac{1}{2}\sigma^2 L d\\
&\leq  \mathbb{E}_{\zeta \sim N(0,I_d)} \left[  \epsilon + \frac{1}{2 \epsilon}L^2  \sigma^2 \|\zeta\|^2 \right ]  + \frac{1}{2}\sigma^2 L d\\
&= \epsilon + \frac{1}{2 \epsilon}L^2  \sigma^2 d  +   \frac{1}{2}\sigma^2 L d,
\ee
where the third inequality is due to the fact that $f$ has $L$-Lipschitz Hessian.  The fourth ineqaulity holds because $-L I_{2d} \preceq \nabla^2   f(x^\star, y^\star) \preceq L I_{2d}$ because $f$ has $L$-Lipschitz gradient, $\nabla^2   f(x^\star, y^\star)$ is symmetric (and therefore is diagonalizable by an orthogonal matrix) and $\zeta \sim N(0,I_d)$.
The fifth inequality holds since $\sqrt{\cdot}$ is concave, and the distribution of $\langle \nabla_y f(x^\star, y^\star) , \, L  \sigma \zeta \rangle$ is symmetric about $0$.
And the seventh inequality is due to the fact that the linear approximation to the function $\sqrt{t}$ at $t=\epsilon^2$ gives an upper bound for the function $\sqrt{t}$ for all $t \geq \epsilon^2$, since the function $\sqrt{t}$ is concave.

\end{proof}

 \section{Hardness results} \label{sec:hardness}
 \subsection{Hardness of nonconvex optimization in the oracle model} \label{sec:hardness_oracle}
In this section we show a hardness result for global nonconvex optimization in the oracle model (that is, when one is only given access to an oracle for the value, gradient, and Hessian of the objective function). 
 Although we could not find a reference for such a result in the literature, we suspect that it is widely known to be true.
 For hardness results for {convex} optimization in the oracle model, see for instance \cite{Problem_Complexity_book}.

Define the bump function $\psi: \mathbb{R}^d \rightarrow \mathbb{R}$ by
\be
\psi(x) := \begin{cases} e^{-\frac{1}{1-\|x\|^2}} \quad \textrm{ if } \|x\| < 1,\\ 0 \quad \textrm{ otherwise.} \end{cases}
\ee
In particular, we note that $\psi$ is $1$-lipschitz with $8$-Lipschitz gradient, and that $\sup_{x\in \mathbb{R}^d}\psi(x) =\frac{1}{e}$ and $\inf_{x\in \mathbb{R}^d}\psi(x) = 0$.  We also note that $\psi$ is $\mathcal{C}^\infty$ with all its derivatives vanishing outside of the ball $B(0,1)$.

We first prove hardness for the case of deterministic algorithms (Theorem \ref{thm:hardness_deterministic}), then we generalize the result to randomized algorithms (Corollary \ref{thm:hardness_randomized}).

\begin{theorem}[Hardness of nonconvex optimization for deterministic algorithms] \label{thm:hardness_deterministic}
Let $\mathcal{A}(g)$ be any deterministic algorithm which takes as input any function $g: \mathbb{R}^d \rightarrow \mathbb{R}$,  and has output $\tilde{x}(g) \in \mathbb{R}^d$ \footnote{We only allow one output point, since then the algorithm could just output all the points in $\mathbb{R}^d$ without making any oracle calls.}, where $\mathcal{A}(g)$ can only access the function $g$ by zeroth- first- and second- order oracle access to $g$. 
 Then there exists an objective function $\mathfrak{f}: \mathbb{R}^d \rightarrow \mathbb{R}$ that is 1-Lipschitz, with 8-Lipschitz gradient, and for which $\sup_{x\in \mathbb{R}^d} \mathfrak{f}(x) - \inf_{x\in \mathbb{R}^d} \mathfrak{f}(x) \leq \frac{1}{e}$, with global minimizer $x^\star \in B(0,10)$, such that the algorithm $\mathcal{A}(\mathfrak{f})$ must make at least $2^d$ oracle calls to find an $\epsilon$-global minimizer $x^\circ$ of $\mathfrak{f}$ with $\epsilon = \frac{1}{2e}$ for which $\mathfrak{f}(x^\circ) - \mathfrak{f}(x^\star) \leq \epsilon$.
\end{theorem}
\begin{proof}
We will use the probablistic method.  Let $Z$ be a uniform random point on the ball $B(0,6)$. 
 Consider the candidate function $\mathfrak{\hat{\mathfrak{f}}}(x) = -\psi(x - Z)$, and let $\phi(x) = 0$ for all $x$. 
  Then $\hat{\mathfrak{f}}$ has all its derivatives equal to zero outside of the ball $B(Z,1)$. 
   Therefore, $\mathcal{A}(\phi)$ and  $\mathcal{A}(\hat{\mathfrak{f}})$ are exactly the same algorithm up to the point when $\mathcal{A}(\phi)$ makes an oracle query for a point in the ball $B(Z,1)$.

Let $\mathcal{I}$ be the number of oracle calls until the algorithm $\mathcal{A}(\phi)$ queries a point in $B(Z,1)$. 
 We will show that $\mathcal{I} \geq 2^d$ with probability at least $\frac{1}{2^d}$. 
  Let $x_1,\ldots, x_{\tau}$ be the sequence of points at which algorithm $\mathcal{A}(\phi)$ makes its first $x_{\tau}$ oracle calls, where $\tau := \min(2^d-1, \mathcal{I})$.  Let $x_{\tau+1}:= \tilde{x}(\phi)$ be the output of the algorithm. 
   Then for all $j>0$ we have
\be  \label{eq:hardness1}
\mathbb{P}(x_j \in B(Z,1)) \leq \frac{\mathrm{Vol}(B(0,1))}{\mathrm{Vol}(B(0,5))}= \frac{1}{5^d}.
\ee
Hence,
\be \label{eq:hardness2}
\mathbb{P}(x_j \notin B(Z,1) \forall i \in [\tau+1] ) \geq 1- 2^d \times  \frac{1}{5^d} -1 \geq 1- \frac{1}{2^d}.
\ee
Therefore, Inequality \eqref{eq:hardness2} implies, with probability at least $1-\frac{1}{2^d}$, that we have both $\mathcal{I} \geq \frac{1}{2^d}$ and $\tilde{x}(\phi) = x_{\tau+1} \notin B(Z,1)$. 
 Therefore, with probability at least $1-\frac{1}{2^d}$, $\mathcal{A}(\phi)$ and  $\mathcal{A}(\hat{\mathfrak{f}})$ are exactly the same algorithm for their first $2^d$ oracle calls (or until both algorithms terminate) and $\tilde{x}(\phi)\notin B(Z,1)$. 
  Therefore, we have that, if $\mathcal{A}(\hat{\mathfrak{f}})$ outputs a point before it makes $2^d +1$ oracle calls, this point must be $\tilde{x}(\hat{\mathfrak{f}}) = \tilde{x}(\phi) \notin  B(Z,1)$ with probability at least $1-\frac{1}{2^d}$.

Therefore, since all the global $\epsilon$-minimizers of $\mathfrak{f}$ are in the ball $B(Z,1)$, we must have that, with probability at least $1-\frac{1}{2^d}$, algorithm $\mathcal{A}(\hat{\mathfrak{f}})$  does not output an $\epsilon$-minimizer of $\hat{\mathfrak{f}}$ before it makes $2^d +1$ oracle calls.  Since the probability of this event is nonzero, there must exist a function $\mathfrak{f}$ for which the algorithm $\mathcal{A}(\mathfrak{f})$ must make at least $2^d$ oracle calls to find an $\epsilon$-global minimizer of $\mathfrak{f}$ for $\epsilon = \frac{1}{e}$.

\end{proof}

\begin{corollary}[Hardness of nonconvex optimization for randomized algorithms] \label{thm:hardness_randomized}
Let $\mathcal{A}(g)$ be any randomized algorithm which takes as input any function $g: \mathbb{R}^d \rightarrow \mathbb{R}$,  and has output $\tilde{x}(g) \in \mathbb{R}^d$, where $\mathcal{A}(g)$ can only access the function $g$ by zeroth- first- and second- order oracle access to $g$. 
 Then there exists an objective function $\mathfrak{f}: \mathbb{R}^d \rightarrow \mathbb{R}$ that is 1-Lipschitz, with 8-Lipschitz gradient, and for which $\sup_{x\in \mathbb{R}^d} \mathfrak{f}(x) - \inf_{x\in \mathbb{R}^d} \mathfrak{f}(x) \leq \frac{1}{e}$, with global minimizer $x^\star \in B(0,10)$, such that, with probability at least $1-\frac{1}{2^d}$, the algorithm $\mathcal{A}(\mathfrak{f})$ makes at least $2^d$ oracle calls before it outputs an $\epsilon$-global minimizer $x^\circ$ of $\mathfrak{f}$ with $\epsilon = \frac{1}{2e}$ for which $\mathfrak{f}(x^\circ) - \mathfrak{f}(x^\star) \leq \epsilon$.
\end{corollary}
\begin{proof}
The proof is the same as the proof of Theorem \ref{thm:hardness_deterministic}, with the following modifications: 
The points $x_1,\ldots, x_{\tau}$ queried by the algorithm $\mathcal{A}(\phi)$, and the output point $x_{\tau+1}$. 
Since $\mathcal{A}(\phi)$ does not depend on $\hat{\mathfrak{f}}$, this sequence of random points is jointly independent of the random vector $Z$ which defines the random function $\mathfrak{\hat{f}}$. 
Therefore, Inequality \eqref{eq:hardness2} still holds.

By the same reasoning  as the proof of Theorem \ref{thm:hardness_deterministic}, Inequality \eqref{eq:hardness2} implies that with probability at least $1-\frac{1}{2^d}$, Algorithm $\mathcal{A}(\hat{\mathfrak{f}})$  does not output an $\epsilon$-minimizer of $\hat{\mathfrak{f}}$ before it makes $2^d +1$ oracle calls. 
 Therefore, there must be a (non-random) function $\mathfrak{f}$ for which Algorithm $\mathcal{A}(\mathfrak{f})$,  with probability at least $1-\frac{1}{2^d}$,  makes at least $2^d$ oracle calls before it outputs an $\epsilon$-global minimizer $x^\circ$ of $\mathfrak{f}$.

\end{proof}

\begin{remark} [Hardness of finding exact local minima] \label{local_minimum_hardness}
One can also show that finding an exact local minimum can require a number of function calls which is exponential in the dimension $d$, using roughly the same ideas as in the proof  Theorem \ref{thm:hardness_deterministic}.  Consider the class of functions $f(x) = \mathrm{sigmoid}(x[1]) - 10\psi(x + c)$. for some $c \in \mathbb{R}^d$ such that $\|c\| \leq 10$.  This function has only one exact local minimum, and this exact local minimum is always in a ball $B(c,1)$ of radius $1$ around $c$.  But, for this class of functions, $c$ can be anywhere inside a ball $B(0,10)$ of radius 10 centered at the origin, and the only way to ``find" this ball using a gradient or function oracle is if one calls this oracle inside the radius 1 ball $B(c)$.  Then any algorithm will require at least $\frac{\mathrm{vol}(B(0,10)}{\mathrm{vol}(B(c,1)} = 10^d$ function or gradient evaluations to find a point within a distance of $1$ of the exact local minimum of $f$. 
\end{remark}

\subsection{Hardness of optimizing bounded Lipschitz RELUs}

When the objective function one wants to optimize is a neural network of Lipschtiz RELUs, \cite{goel2017reliably} show that optimizing the weights of this neural network is at least as hard as the {\em Learning Sparse Parity with Noise} problem. 
 Specifically, consider the class of depth-2 neural networks with $k$ RELUs with the restriction $\|w\|_1 \leq 2k$ on the weight vector and taking inputs on $\{0,1\}^n$. \cite{goel2017reliably}  show that the problem of finding such weights that globally optimize a $2k$-Lipschitz objective function of the outputs of the Neural network within error $\epsilon = \omega(1)$ is at least as hard as the {\em Learning Sparse Parity with Noise} problem for parity functions on size $k$ subsets of the vertices of the $n$-cube $\{0,1\}^n$. 
  This problem is conjectured \cite{goel2017reliably} to require time $n^{\Omega(k)}$ to solve, based on the best currently available bounds \cite{blum2003noise, valiant2012finding}.

\section{Proof of Lemma \ref{lemma_local_nash} } \label{Sec_Auxilliary}

\noindent
We give the proof Lemma \ref{lemma_local_nash} presented in Section \ref{Sec_Discussions_and_Limitations}:

\begin{proof}[Proof of Lemma \ref{lemma_local_nash}]
Since $x^\star$ is an exact local minimum for $\psi$, there exists $\delta>0$ such that
\be
\psi(x^\star) \leq \psi(x)   \qquad \qquad \forall x \in \mathbb{R}^d\textrm{ such that }\|x - x^\star \| \leq \delta.
\ee

\noindent
Choose $\sigma>0$ small enough such that 
\be
\int_{\zeta \in \mathbb{R}^d: \|x - \zeta\| > \delta} \frac{\|\zeta\|}{ \sqrt{2\pi}} e^{- \frac{\|\zeta\|^2}{2 \sigma^2}} \, \,\mathrm{d} \zeta \leq \frac{\epsilon}{b},
\ee
and
\be
\int_{\zeta \in \mathbb{R}^d:  \|x - \zeta\| > \delta} \left (\frac{\|\zeta\|^2}{\sigma^5 \sqrt{2\pi}} +  \frac{1}{\sigma^3 \sqrt{2 \pi}} \right) e^{- \frac{\|\zeta\|^2}{2 \sigma^2}} \, \,\mathrm{d} \zeta \leq \frac{\sqrt{\epsilon}}{b}.
\ee

 \noindent
Let $\phi$ be the probability density of the standard normal distribution $N(0,I_d)$.  Then
\be
\nabla_x  \mathbb{E}_{\zeta \sim N(0,I_d)}\left [\min(\psi(x  + \sigma \zeta),  \psi(x^\star)) \right] &=  \nabla_x  \int_{\mathbb{R}^d} \min(\psi(x  + \sigma \zeta),  \psi(x^\star))  \times \phi(\zeta)  \, \,\mathrm{d} \zeta\\
&=\nabla_x  \int_{\mathbb{R}^d} \min(\psi(\zeta),  \psi(x^\star))  \times \frac{1}{\sigma^d} \phi(\frac{1}{\sigma}(x - \zeta))  \, \,\mathrm{d} \zeta\\
&= \int_{\mathbb{R}^d} \min(\psi(\zeta),  \psi(x^\star))  \times \nabla_x \frac{1}{\sigma^d} \phi(\frac{1}{\sigma}(x - \zeta))  \, \,\mathrm{d} \zeta
\ee
where the third equation holds by the dominated convergence theorem since (1) $\psi$ is uniformly bounded and (2) $\frac{\mathrm{d}}{\mathrm{d}t} \phi(t)= \lim_{\delta \rightarrow 0} \frac{1}{\delta}(\phi(t) - \phi(t+\delta))$ and the finite difference $\frac{1}{\delta}(\phi(t) - \phi(t+\delta)$ for $0<\delta<1$ is dominated by the integrable function $D(x)$, where
\be
D(x) = \begin{cases}
8(x+2)e^{-(x+2)^2} \qquad x<-1, \\
8 \qquad \qquad \qquad \quad -1 \leq x \leq 1,\\
8(x-2)e^{-(x-2)^2} \qquad x>1.
\end{cases}
\ee

 \noindent
Therefore, since $|\psi | \leq b$,
\be
\|\nabla_x  \mathbb{E}_{\zeta \sim N(0,I_d)}\left [\min(\psi(x  + \sigma \zeta),  \psi(x^\star)) \right] \|&\leq \int_{\zeta \in \mathbb{R}^d: \|x - \zeta\| \leq \delta} 0  \times \left \|\nabla_x \frac{1}{\sigma^d} \phi(\frac{1}{\sigma}(x - \zeta))\right\|  \, \,\mathrm{d} \zeta\\
&\qquad +  \int_{\zeta \in \mathbb{R}^d: \|x - \zeta\| > \delta} b \times \left \|\nabla_x \frac{1}{\sigma^d} \phi(\frac{1}{\sigma}(x - \zeta)) \right\|  \, \,\mathrm{d} \zeta\\
& = 0 + b \times \int_{\zeta \in \mathbb{R}^d: \|x - \zeta\| > \delta} \frac{\|\zeta\|}{\sigma^3 \sqrt{2\pi}} e^{- \frac{\|\zeta\|^2}{2 \sigma^2}} \, \,\mathrm{d} \zeta\\
&\leq b \times \frac{\epsilon}{b} = \epsilon.
\ee
\noindent
This completes the proof of the left-hand Inequality \eqref{eq_local_lemma2}.

By a similar reasoning, we also have that
\be \label{eq_local_lemma1}
\|\nabla^2_x \mathbb{E}_{\zeta \sim N(0,I_d)}\left [\min(\psi(x  + \sigma \zeta),  \psi(x^\star)) \right] \|_{\mathrm{op}}&\leq \int_{\zeta \in \mathbb{R}^d: \|x - \zeta\| \leq \delta} 0  \times \left \|\nabla^2_x\frac{1}{\sigma^d} \phi(\frac{1}{\sigma}(x - \zeta))\right\|_{\mathrm{op}}  \, \,\mathrm{d} \zeta\\
&\qquad +  \int_{\zeta \in \mathbb{R}^d: \|x - \zeta\| > \delta} b \times \left \|\nabla^2_x\frac{1}{\sigma^d} \phi(\frac{1}{\sigma}(x - \zeta)) \right\|_{\mathrm{op}}  \, \,\mathrm{d} \zeta\\
& = 0 + b \times \int_{\zeta \in \mathbb{R}^d: \|x - \zeta\| > \delta} \left (\frac{\|\zeta\|^2}{\sigma^5 \sqrt{2\pi}} +  \frac{1}{\sigma^3 \sqrt{2 \pi}} \right) e^{- \frac{\|\zeta\|^2}{2 \sigma^2}} \, \,\mathrm{d} \zeta\\
&\leq b \times \frac{\sqrt{\epsilon}}{b} = \sqrt{\epsilon}.
\ee

\noindent
Therefore, since for any matrix $A$ we have that $|\lambda_{\mathrm{min}}(A)| \leq \|A \|_{\mathrm{op}}$,  Inequality \eqref{eq_local_lemma1} implies that
 \be
   \lambda_{\mathrm{min}}( \nabla^2_x \mathbb{E}_{\zeta \sim N(0,I_d)}\left [\min(\psi(x  + \sigma \zeta),  \psi(x^\star)) \right] ) \geq -\sqrt{ \epsilon}.
  \ee
This completes the proof of the right-hand Inequality \eqref{eq_local_lemma2}.
\end{proof}

\end{document}